\documentclass[nonacm,acmsmall,printccs=false,printacmref=false]{acmart}

\usepackage{booktabs}   
\usepackage{subcaption} 

\usepackage[noend]{algorithm2e}
\usepackage{listings}
\usepackage{amsmath}
\usepackage{stmaryrd}
\usepackage{wrapfig}

\usepackage{tikz-cd}
\usetikzlibrary{decorations.pathmorphing}

\usepackage{proof}
\usepackage[normalem]{ulem}
\makeatletter
\newcommand{\xRightarrow}[2][]{\ext@arrow 0359\Rightarrowfill@{#1}{#2}}
\makeatother
\newcommand\apxA{~A}

\newcommand\code[1]{\texttt{#1}}

\newcommand\Tool{{\sc Cion}}

\usepackage[utf8]{inputenc}
\usepackage{amsmath}
\usepackage{amsfonts}
\usepackage{program}
\usepackage{color}
\usepackage{graphicx}

\newcommand\ie{{\it i.e.,} }
\newcommand\eg{{\it e.g.,} }

\newcommand\added[1]{\textcolor{blue}{#1}}
\newcommand\removed[1]{ }

\newcommand\smartpar[1]{\medskip\noindent{\bf #1. $\;$ }}

\newcommand\exprSet{\mathsf{Abs}}
\newcommand\expr{\mathsf{expr}}

\newcommand\meth{\mathit{Meth}}

\newcommand{\As}{\vec{x}} 
\newcommand{\rvs}{\vec{v}} 
\newcommand{\avals}{\vec{u}}
\newcommand{\rvals}{\vec{w}}

\newcommand{\rv}{v} 
\newcommand{\fst}{\mathit{first}}
\newcommand{\lst}{\mathit{last}}
\newcommand{\Inv}[2]{\textsf{call}\;#2}
\newcommand{\Ret}[2]{\textsf{ret}(#2)}
\newcommand{\tid}{t}
\newcommand{\tids}{\mathcal{T}}
\newcommand{\lstate}{\sigma_l}
\newcommand{\supp}{\mathit{supp}}

\newcommand{\gstate}{\sigma_{g}}
\newcommand{\Lstate}{\Sigma_{lo}}
\newcommand{\Gstate}{\Sigma_{gl}}
\newcommand{\Conf}{\mathcal{C}}
\newcommand{\sem}[1]{\llbracket  #1 \rrbracket}
\newcommand{\semL}{\llbracket}
\newcommand{\semR}{\rrbracket}
\newcommand{\dom}[1]{\textsf{dom}(#1)}
\newcommand{\exec}{\rho}
\newcommand{\execs}{\mathcal{E}}

\newcommand\Pths{\Pi}
\newcommand\Lyrs{\Lambda}
\newcommand\lyr{\lambda}

\newcommand\eeq{\equiv}
\newcommand\symeq{\simeq}
\newcommand\neeq{\not\equiv}

\newcommand{\tr}{\tau}
\newcommand{\traces}{\mathit{Traces}}

\newcommand{\condLL}{\langle\!\!\langle}
\newcommand{\condLLd}{\langle\!\!\!\langle}
\newcommand{\condwr}[2]{\condLL #1\kdot #2 \rangle\!\!\rangle}
\newcommand{\condwrd}[2]{\condLLd #1\kdot #2 \rangle\!\!\!\rangle}
\newcommand{\ttcondwr}[2]{\condwr{\code{#1}}{\code{#2}}}
\newcommand{\ttcondwrd}[2]{\condwrd{\code{#1}}{\code{#2}}}

\newcommand{\lin}[1]{\mathit{lin}({#1})}

\newcommand\kdot{\xspace\cdot\xspace}
\newcommand\tightkdot{\!\xspace\cdot\xspace\!}

\newcommand\qq{\mathit{q}}
\newcommand\trans{\delta}

\newcommand\QQ{\mathcal{Q}}

\newcommand\kat{k}
\newcommand\Kat{\mathcal{K}}
\newcommand\bkat{\mathcal{B}}
\newcommand\Akat{\Sigma}
\newcommand\kAs{\textsf{A}}
\newcommand\kBs{\textsf{B}}

\newcommand\env{\mathcal{E}}
\newcommand\lbl{\ell}
\newcommand\Lbls{\mathcal{L}}

\newcommand\ttc{\code{c}}
\newcommand\ttctr{\code{ctr}}

\definecolor{bluekeywords}{rgb}{0,0,1}
\definecolor{greencomments}{rgb}{0,0.5,0}
\definecolor{redstrings}{rgb}{0.64,0.08,0.08}
\definecolor{xmlcomments}{rgb}{0.5,0.5,0.5}
\definecolor{types}{rgb}{0.17,0.57,0.68}
\newcommand\dt[1]{\textbf{\emph{#1}}}

\newcommand\equo[1]{\langle\!\lfloor #1 \rfloor\!\rangle}

\newcommand{\seqlbl}{\omega}
\newcommand\ttdelete{\code{delete}}
\newcommand\ttinsert{\code{insert}}
\newcommand\ttlocate{\code{locate}}

\newcommand\xrsquigarrow[1]{%
    \mathrel{%
        \begin{tikzpicture}[%
            baseline={(current bounding box.south)}
            ]
        \node[%
            ,inner sep=.44ex
            ,align=center
            ] (tmp) {$\scriptstyle #1$};
        \path[%
            ,draw,<-
            ,decorate,decoration={%
                ,zigzag
                ,amplitude=0.7pt
                ,segment length=1.2mm,pre length=3.5pt
                }
            ] 
        (tmp.south east) -- (tmp.south west);
        \end{tikzpicture}
        }
    }
\newcommand\exectrans[1]{\xrsquigarrow{#1}}

\newcommand\QuoteAuthor[2]{%
  \begin{center}%
  \fbox{\begin{minipage}{0.93\columnwidth}\emph{#1} -- #2\end{minipage}}%
  \end{center}}
\newcommand\QuoteAuthorCiteless[1]{%
  \begin{center}%
  \fbox{\begin{minipage}{0.93\columnwidth}\emph{#1}\end{minipage}}%
  \end{center}}
\newcommand\TAOMPP{\citet{TAOMPP}}

\newif\ifarxiv  \arxivtrue
\newif\ifconf   \conffalse
\newcommand\ARXIV[1]{#1}
\newcommand\CONF[1]{}
\newcommand\ignoreme[1]{}
\newcommand\ttpop{\code{pop}}
\newcommand\ttpush{\code{push}}

\newcommand\intheextended{in the extended version~\cite{arxiv}}
\newcommand\inExtOrApx[1]{\CONF{\intheextended{}}\ARXIV{in Apx.~\ref{#1}}}
\newcommand\ExtOrApx[1]{\CONF{\cite{arxiv}}\ARXIV{Apx.~\ref{#1}}}
\newcommand\ExtOrApxt[1]{\CONF{\citet{arxiv}}\ARXIV{Apx.~\ref{#1}}}

\definecolor{bluekeywords}{rgb}{0,0,1}
\definecolor{greencomments}{rgb}{0,0.5,0}
\definecolor{redstrings}{rgb}{0.64,0.08,0.08}
\definecolor{xmlcomments}{rgb}{0.5,0.5,0.5}
\definecolor{types}{rgb}{0.17,0.57,0.68}
\begin{document}

\ifconf
\title{Scenario-Based Proofs for Concurrent Objects}
\fi
\ifarxiv
\title{Scenario-Based Proofs for Concurrent Objects [Extended Version]}
\fi

\author{Constantin Enea}
\orcid{0000-0003-2727-8865}
\affiliation{%
  \institution{LIX - CNRS - École Polytechnique}
  \city{Paris}
  \country{France}
}
\email{cenea@lix.polytechnique.fr}

\author{Eric Koskinen}
\orcid{0000-0001-7363-634X}
\affiliation{%
  \institution{Stevens Institute of Technology}
  \city{Hoboken}
  \country{USA}
}
\email{eric.koskinen@stevens.edu}

\vspace{-10mm}
\begin{abstract}

Concurrent objects form the foundation of many applications that exploit multicore architectures and their importance has lead to informal correctness arguments, as well as formal proof systems.
Correctness arguments (as found in the distributed computing literature) give intuitive descriptions of  a few canonical executions or ``scenarios'' often each with only a few threads, yet it remains unknown as to whether these intuitive arguments have a formal grounding and extend to arbitrary interleavings over unboundedly many threads.

We present a novel proof technique for concurrent objects,
based around identifying a small set of scenarios (representative, canonical interleavings), formalized as the commutativity quotient of a concurrent object. 
We next give an expression language for defining abstractions of the quotient in the form of regular or context-free languages that enable simple proofs of linearizability. 
These quotient expressions organize unbounded interleavings into a form more amenable to reasoning and make explicit the relationship between implementation-level contention/interference and ADT-level transitions.

We evaluate our work on numerous non-trivial concurrent objects from the literature (including the Michael-Scott queue, Elimination stack, SLS reservation queue, RDCSS and Herlihy-Wing queue).
We show that quotients capture the diverse features/complexities of these algorithms, can be used even when linearization points are not straight-forward, correspond to original authors' correctness arguments, and provide some new scenario-based arguments. 
Finally, we show that discovery of some object's quotients reduces to two-thread reasoning and give an implementation that can derive candidate quotients expressions from source code.
\end{abstract}



\maketitle


\section{Introduction}
\label{sec:intro}

Efficient multithreaded programs typically rely on optimized implementations of
common abstract data types ({\sc adt}s) like stacks, queues, and sets, whose operations execute in
parallel to maximize
efficiency. 
Synchronization between operations must be
minimized to increase
throughput~\cite{TAOMPP}. Yet 
this minimal amount of synchronization must also be adequate to ensure that
operations behave as if they were executed atomically, 
so that client programs can rely on their (sequential) {\sc adt} specification;
this de-facto correctness criterion is known as
\emph{linearizability}~\cite{DBLP:journals/toplas/HerlihyW90}. These opposing
requirements, along with the general challenge in reasoning about 
interleavings, make concurrent data structures a ripe source of insidious
programming errors. 

Algorithm designers (\eg researchers defining new concurrent objects) argue about correctness by considering some number of ``scenarios'', \ie interesting ways of interleaving steps of different operations, and showing for instance, that each one satisfies some suitable invariant (which is not necessarily inductive). 
For example, a scenario of the~\citet{conf/podc/MichaelS96} queue is described as: many threads concurrently reading, one enqueuer thread taking a specific read path finding a tail pointer to be outdated, and then succeeding a compare-and-swap (CAS) operation, causing others to fail their compare-and-swap (paraphrasing from~\citet{TAOMPP}). Such scenario descriptions are powerful because they describe unboundedly many threads and often generalize to cover many executions that are equivalent due to commutative re-orderings. Consequentially, informal correctness arguments need only consider a few representative scenarios. Furthermore, another critical benefit of scenario-based reasoning is that scenarios are more readily explainable to software developers, who need not have a background in formal logic.

Despite the intuitive benefit of these operational, scenario-based proofs---which continue to be widely used in the concurrent algorithms literature---it remains unknown as to whether they have a formal grounding. This has lead to cases where objects thought to be linearizable~\cite{DBLP:conf/wdag/DetlefsFGMSS00} where later determined to contain bugs in unconsidered scenarios~\cite{DBLP:conf/spaa/DohertyDGFLMMSS04}. 

\subsection{Formalizing Scenarios with Quotients}
\label{subsec:intro-quo}

In this paper, we show that operational, scenario-based correctness arguments can be formally grounded.
To that end, we propose a new proof methodology that is based on formal arguments while keeping the intuition of scenario-based reasoning. This methodology relies on a reduction to reasoning about a subset of \emph{representative} interleavings (i.e. a formal version of informal scenarios), which cover the whole space of interleavings modulo repeatedly swapping adjacent commutative steps. The latter corresponds to the standard \emph{equivalence up to commutativity} between the executions of an object (e.g., Mazurkiewicz traces~\cite{DBLP:conf/ac/Mazurkiewicz86}). 

Reductions based on commutativity arguments have been formalized in previous work, \eg~Lipton's reduction theory~\cite{DBLP:journals/cacm/Lipton75}, QED~\cite{DBLP:conf/popl/ElmasQT09}, CIVL~\cite{DBLP:conf/cav/HawblitzelPQT15}, and they generally focus on identifying \emph{atomic sections}, \ie~sequences of statements in a single thread that can be assumed to execute without interruption (without sacrificing completeness). 
Relying on atomic sections for reducing the space of interleavings has its limitations, especially in the context of concurrent objects. These objects rely on intricate algorithms where almost every step is an access to the shared memory that does not commute with respect to other steps. 

Our reduction argument reasons about a \emph{quotient} of the set of object executions, which is a subset of executions that contains a representative from each equivalence class. In general, an execution of an object interleaves an unbounded number of invocations to the object's methods, each from a different thread\footnote{Typically, it can be assumed w.l.o.g. that each thread performs a single invocation in an execution.}. These executions can be seen as a word over an infinite alphabet, each symbol of the alphabet representing a statement in the code and the thread executing that statement\footnote{Such a sequence will be called a \emph{trace} in the formalization we give later in the paper.}. We show that when abstracting away thread ids from executions, carefully chosen quotients become \emph{regular or context-free languages}. This is not true for any quotient since representatives of equivalence classes can be chosen in an adversarial manner to make the language more complex. 

The principal benefit of quotients is that reasoning about correctness can be done by considering only a few representative execution interleavings, yet those conclusions generalize to all executions. 
For some kinds of concurrent object implementations (defined later), deriving representative traces can be  reduced via induction to two-thread reasoning.

\emph{Proofs with program logics.}
Our work is inspired by the success of many prior works on proofs for concurrent objects based on program logics such as
\citet{DBLP:journals/cacm/OwickiG76},
Rely/Guarantee~\cite{DBLP:conf/ifip/Jones83},
Concurrent separation logic~\cite{DBLP:journals/tcs/Brookes07,DBLP:journals/tcs/OHearn07},
RGSep~\cite{DBLP:conf/concur/VafeiadisP07},
Deny-Guarantee~\cite{DBLP:conf/esop/DoddsFPV09},
Views~\cite{DBLP:conf/popl/Dinsdale-YoungBGPY13},
Iris~\cite{DBLP:conf/popl/JungSSSTBD15,DBLP:journals/jfp/JungKJBBD18}
and interactive proof tools for such logics.

The goal of this paper is orthogonal and focuses on finding a formal grounding for the operational, scenario-based 
correctness arguments present in the algorithms literature.
To this end, our methodology is based on taking representative interleaved traces upfront and using commutativity-based equivalence classes for modularity/generalization rather than exploiting the program structure and invariants for modularity/generalization.
Achieving this alternative reasoning strategy nonetheless requires careful formalization of what is meant by ``representative traces'', as well as how those classes of traces can be expressed abstractly, which we outline below. 
Our results show that (i) scenario-based reasoning can be done formally through quotients, (ii) quotients can be given for a variety of concurrent objects with subtle differences including non-fixed linearization points, (iii) quotients improve the correctness arguments from the literature, and (iv) 
for some cases, quotients---which represent interleavings of unboundedly many threads---can be automatically discovered through a reduction to two-thread reasoning.

\subsection{Example: Scenario-Based Proofs of the Michael-Scott Queue}

For the sake of concreteness, we now show how quotients make concurrent reasoning simpler, using the canonical Michael-Scott Queue (MSQ) as an example. Ultimately the theory and algorithms in this paper lead to an implementation that is able to automatically derive the representation discussed below, from the object's source code.
The MSQ is implemented as a linked-list, with head and tail pointers and a sentinel head node, as depicted to the left below.
\begin{center}
\includegraphics[width=5in]{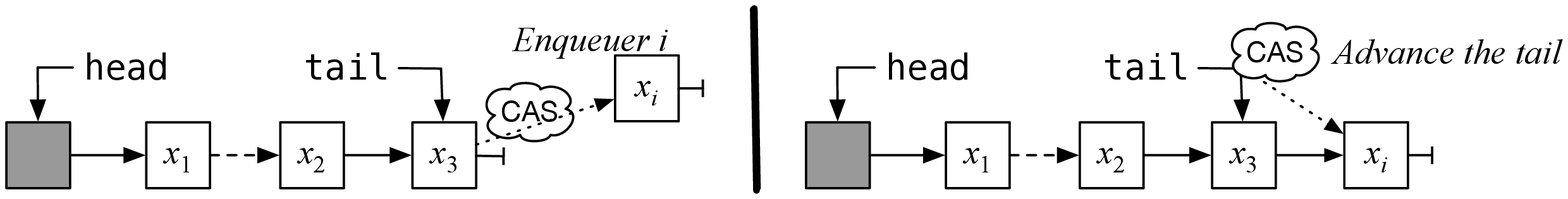}
\end{center}

An enqueue (\texttt{enq}) operation, such as $Enqueuer \; i$ in the diagram above, repeatedly attempts to enqueue a new element by using an atomic compare-and-swap (CAS) operation on the \texttt{tail} element’s \texttt{next} pointer, replacing null with the address of the new node ($x_i$ in the diagram above). It is possible that this CAS operation will fail due to a concurrent enqueuer (of which there can be unboundedly many). Nonetheless, due to the CAS, one enqueuer will succeed. At this point, although the element is linked, it is not logically in the queue because the \texttt{tail} pointer is lagging. The enqueuer will thus perform a second CAS operation, as shown on the digram above to the right, to advance \texttt{tail} to point to $x_i$. To ensure progress, concurrent enqueuers will also check to see if the \texttt{tail} lags and, if so, attempt to advance the \texttt{tail} before they attempt to enqueue their elements (i.e. helping).
A dequeue (\texttt{deq}) operation 
repeatedly attempts to advance the head pointer to make $x_1$ the new sentinel node,
but also has to check that the queue is non-empty and that other threads have not recently dequeued. (To achieve all of these cases, \texttt{deq} must begin by reading the \texttt{head} pointer, the \texttt{tail} pointer and \texttt{head}’s \texttt{next} pointer and validating to see which case applies.)

To verify the correctness of objects like the MSQ, one has to consider all of the ways in which concurrent invocations of unboundedly many methods could interleave. 
One strategy to tackle this problem has been through the aforementioned program logics such as rely-guarantee where, roughly, one defines state-based invariants and then shows they are preserved and threads don't interfere with other threads' actions.
Nevertheless, the correctness arguments laid out by algorithm designers (\eg~in the distributed computing community) 
typically are organized in a more operational manner and instead focus on discussing various ``scenarios.''
Consider the following excerpt from The Art of Multiprocessor Programming~\cite{TAOMPP} regarding the MSQ:

\QuoteAuthorCiteless{An enqueuer creates a new node, reads tail, and finds the node that appears to be last. To verify that node is indeed last, it checks whether that node has a successor. If so, the thread attempts to append the new node with CAS. (A CAS is required because other threads may be trying 
the same thing.) [Assume that] the CAS succeeds.}

\noindent
Such sentences describe scenarios that involve unboundedly many threads executing some portion of their programs. They are chosen to highlight tricky situations and describe why those situations are still acceptable. The above example can be thought of as the sequence:

\begin{enumerate}
\item Unboundedly many threads are reading the data structure.
\item There is a distinguished thread, let's call $\tau_{enq}$.
\item $\tau_{enq}$ reads the \texttt{tail} and the \texttt{tail}'s \texttt{next} pointer.
\item $\tau_{enq}$ finds that \texttt{tail}'s \texttt{next} is null.
\item $\tau_{enq}$atomically updates \texttt{tail}'s \texttt{next} to point to its new node.
\item The other (unboundedly many) threads fail their CASes on \texttt{tail}'s \texttt{next} and restart.
\end{enumerate}
This scenario has a particular shape about it: unboundedly many threads read, then a single thread performs a write, then the remaining threads react to that write. This is a common setup in many non-blocking concurrent algorithms and a useful pattern (although, in general, we will describe scenarios beyond those of this shape). One might think of it as a regular expression denoted $r_\texttt{next}$:
\[
r_\texttt{next} \equiv (\tau \in T: read + \tau_{enq}:read)^* \cdot 
  (\tau_{enq}: \texttt{cas/succeed}) \cdot
  (\tau \in T: restart)^*
\]
where $T$ is the (unbounded) set of all threads excluding $\tau_{enq}$.
Above $r_\texttt{next}$ expresses that some unboundedly many threads from set $T$ (including $\tau_{enq}$) perform only $read$-path actions, then $\tau_{enq}$ succeeds its \texttt{cas}, then those unboundedly many threads restart.
This expression is more powerful than it may first appear. There are a few important considerations:
\begin{itemize}
\item \emph{Conciseness}. The entirety of MSQ's concurrent execution behaviors can be represented with \emph{this and only two other} similarly concise representative interleavings, along with four even simpler read-only interleavings.
Expressions $r_\texttt{tail}$ and $r_\texttt{head}$ are similarly defined and represent advancing the tail pointer and the \texttt{head} pointer (due to a dequeuer), respectively.
\item \emph{Unbounded}. With these concise descriptions, the interleavings between an unbounded number of enqueuers and dequeuers can be seen as an unbounded alternation
$(r_{\texttt{next}} + r_{\texttt{tail}} + r_{\texttt{head}})^*$. 
(Below we will further refine this approximation with stateful automata.)
\end{itemize}

This starred-union description does not include all possible ways of interleaving steps of enqueuers, \eg~it does not include interleavings where a thread restarts after two successful CASs since it last read the shared memory. It includes just a subset of representatives that we call a quotient, which is succinct enough to correspond to the designer’s intuition and large enough to cover the whole space of interleavings modulo repeatedly swapping adjacent commutative steps (\ie~the standard equivalence up to commutativity between executions known as Mazurkiewicz traces~\cite{DBLP:conf/ac/Mazurkiewicz86}). For instance, an interleaving where a thread restarts after two successful CASs (since it last read the shared memory) is equivalent to one where the restart step is reordered to the left to occur immediately after the first CAS. This is because the restarting condition is fulfilled after this first CAS as well and the restart step does not perform any writes. 

The MSQ falls into a special class of objects for which quotients can be expressed in this inductive way, as a sequence of what we call ``layers'' (above 
$r_{\texttt{next}}$, $r_{\texttt{tail}}$ and $r_{\texttt{head}}$ are layers) wherein only a single shared memory \emph{write} action occurs per layer, and all other actions are thread-local/read-only (perhaps restarting due to a failed CAS). Consequently, it is possible via induction to reduce reasoning to a collection of two-threaded arguments (one writer, one reader). 
While quotients and their abstractions are a much broader class, layers are nonetheless an important subclass since they apply to many lock-free implementations and can be automated, as discussed below.

\subsection{Challenges and Contributions}

We now identify several challenges toward enabling scenario-based reasoning and discuss how we address them in this paper.

{\bf 1. Concurrent Object Quotients.}
\emph{How can scenario-based reasoning be done formally?}
(Sec.~\ref{sec:quotients})
We show that scenario-based reasoning can be made formal through a 
methodology wherein reasoning about all executions of a concurrent object is reduced to reasoning only about a smaller set of representative interleavings. At the technical core is the definition of an object's execution \emph{quotient} which collapses executions that are equivalent up to swapping commutative adjacent actions. A quotient is parameterized by this equivalence relation and has both a minimality constraint (no two executions are equivalent) and a completeness constraint (all executions are equivalent to some execution in the quotient). We prove that linearizability of the quotient is sufficient to show linearizability of the object. The upshot is that concurrent object correctness is now accomplished via reasoning about a collection of scenarios (as in typical informal proofs).

\smallskip
{\bf 2. Expressing Quotients.} \emph{How can a quotient set be described?}
(Sec.~\ref{sec:finite_reps})
A next question is how to \emph{finitely express} a quotient, which can have unboundedly many interleavings.
In Sec.~\ref{sec:quotients}, we introduce a \emph{quotient expression language} that permits a mixture of regular expressions (\eg~Kleene-star iterations of subexpressions) and context-free grammars (\eg~unbounded but balanced subexpressions). We then give an interpretation/semantics for these expressions that maintains the \emph{minimality} condition: there will only be one interleaving (with threads organized in a canonical order) for every unboundedly many unrolling.
The MSQ expression $(r_{\texttt{next}} + r_{\texttt{tail}} + r_{\texttt{head}})^*$ above provides an intuition for the quotient expression for the MSQ. (Technically, the $read$ actions are paths and the $*$-iterations within the $r_\texttt{x}$ subexpressions are replaced with a context-free form of iteration.)

As we will show later, quotients and their abstractions are \emph{expressive} and can capture canonical concurrent objects as well as more complicated ones such as the \citet{DBLP:journals/toplas/HerlihyW90} queue and the elimination stack of~\citet{DBLP:conf/spaa/HendlerSY04}, each having different kinds of non-fixed linearization points. These are notoriously hard cases for today's proof methodologies. 
We note that, while the idea of reasoning about execution quotients is generic, identifying precise limits for the applicability of the particular class of quotients expressions is hard in general. This is similar to using abstract domains in the context of static analysis: it is hard to determine precisely the class of programs for which interval or polyhedra abstractions are effective.

\smallskip
{\bf 3. Layer Quotient Expressions and Automata.}
(Sec.~\ref{sec:layers})
\emph{In addition to broad expressivity, are there 
classes of objects whose quotients have a simpler structure?}
To increase accessibility and automation, we next describe certain kinds of quotient expressions for which reasoning can actually be reduced, via induction, to two-thread reasoning.
Specifically, 
for objects whose implementation can be written as a collection of (possibly restarting) read-only/local paths and paths that have only a single atomic read-write, we define \emph{layer quotients} to more conveniently and inductively capture the quotient.
Although this does not apply to all objects, it does apply to canonical examples such as the MSQ, Treiber's Stack, and even the~\citet{Scherer2006} synchronous reservation queue.
For these objects, executions can be decompiled into a sequence of \emph{layers}, each described by context-free quotient expressions of the form $(a_1+b_1+\ldots)^n\cdot w\cdot (a_2+b_2+\ldots)^n$ where $a_1 \cdot a_2$ is a read-only path through the method implementation (possibly restarting), and $w$ is a path with a successful atomic read-write. The exponents in both expressions indicate the unbounded replication of local paths ($n$ is not fixed; it ensures prefix/suffix balancing). 
Then an overall quotient expression can be made from regular compositions of these context-free layers, leading to an inductive argument.
Furthermore, each layer can be discovered with two-thread reasoning: considering how each write, treated atomically, impacts each other read-only/local path.

\removed{We show that each such layer can be described by context-free quotient expressions of the form $(a_1+b_1+\ldots)^n\cdot w\cdot (a_2+b_2+\ldots)^n$ where $a_1 \cdot a_2$ is a read-only path through the method implementation (possibly restarting), and $w$ is a path with a successful atomic read-write. The exponents in both expressions indicate the unbounded replication of local paths ($n$ is not fixed; it ensures prefix/suffix balancing). The MSQ can be defined this way, where each subexpression $r_\texttt{x}$ is a layer. }

We describe how layer expressions can be conveniently represented as finite-state \emph{automata} (and further below also used for automation). 
\removed{The automata are defined by a number of states that, roughly, represent pre-conditions of layers, and a number of layer-labeled edges corresponding to a transition made by that layer. Thus a run of the automaton is an unfolding of a layer expression.} The layer automaton for the Michael-Scott Queue is shown in Fig.~\ref{fig:aut-msq}. We will discuss it in detail in Sec.~\ref{ssec:MSQ} but, roughly, the states track whether the queue is empty and whether the tail is lagging. The layer-labeled edges define the local/read-only (unbold) control-flow paths and how they are impacted by the write path (bold). There are also \emph{read-only} layers, which we will describe later.

\begin{figure}[t]\centering
\vspace{2mm}
\includegraphics[width=0.98\columnwidth]{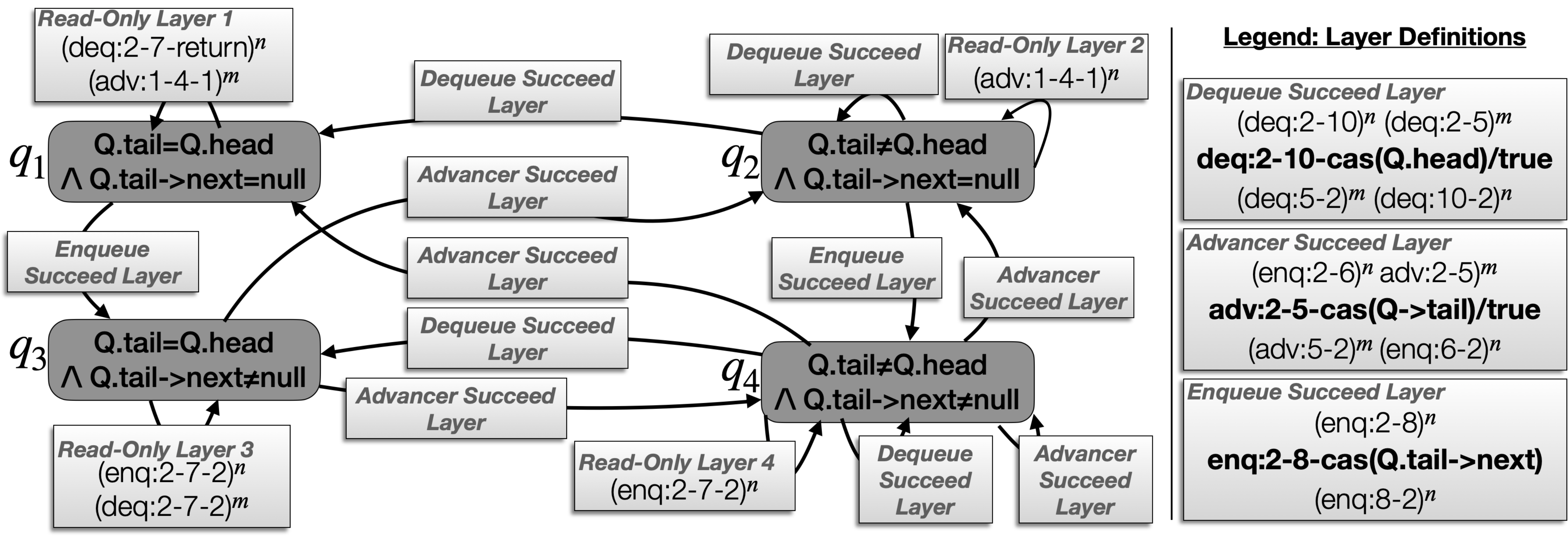}
\caption{\label{fig:aut-msq} Layer automaton for the Michael/Scott Queue.}
\end{figure}

\smallskip
{\bf 4. Evaluation: Verifying Concurrent Objects.} 
(Sec.~\ref{sec:cases})
We consider a broad range of concurrent objects including
Treiber's stack~\cite{IBMTR:Treiber86}, 
the~\citet{conf/podc/MichaelS96} queue,
the~\citet{Scherer2006} synchronous reservation queue,
the~\citet{DBLP:journals/toplas/HerlihyW90} queue,
the~\citet{DBLP:conf/spaa/HendlerSY04} elimination stack, and
the Restricted Double-Compare Single-Swap (RDCSS)~\cite{DBLP:conf/wdag/HarrisFP02}.
Each object has its own subtleties, including complications like multiple CAS steps and non-fixed linearization points. For each object we (i) show that its behavior and linearizability can be captured through a quotient and (ii) revisit the object's authors' correctness arguments. We find that quotients capture those intuitive scenarios and make scenarios explicit and comprehensive.

\smallskip
{\bf 5. Generating Candidate Quotient Expressions.}
(Sec.~\ref{sec:automation})
Automating quotient-based proofs of concurrent objects
is a rather large question (perhaps warranting new forms of induction) which we mostly leave to future work.
Nonetheless, we present an algorithm and prototype implementation \Tool{} for 
generating candidate quotient expressions, directly from a concurrent object's source code.
We manually confirmed that these expressions are sound abstractions of those objects' quotients.
\removed{Finally, we explore initial steps toward automation, showing that we can derive an object's quotient automatically from its source code. We give an algorithm for computing layer automata automatically, and then  
a prototype tool called \Tool{} that uses this algorithm to automatically derive automaton-based layer expressions from concurrent object implementations}
We applied \Tool{} to layer-compatible objects such as Treiber's Stack and the Michael/Scott Queue, 
finding that candidate layer expressions can be discovered in a few minutes.

\smallskip
\CONF{\emph{For lack of space, some detail has been omitted and is available \intheextended{} of this paper.}}
\ARXIV{\emph{This is an extended version of~\citet{OOPSLA2024}.}}
\emph{Our implementation \Tool{} is available on GitHub\footnote{\url{https://github.com/quotientprovers/cion}}, along with benchmark sources.}

\section{Preliminaries}\label{sec:prelim}

{\bf Running example: A simple concurrent counter.}
Fig.~\ref{fig:counter} lists a concurrent counter with methods

\begin{wrapfigure}{r}{4.2cm} 
\vspace{-5mm}
\hspace{1mm}
\begin{minipage}[t]{3.8cm}
\lstset{xleftmargin=1.0ex,numbersep=2pt}
\begin{lstlisting}
int increment() {
  while (true) {
    int c = ctr; `\label{ln:inc_read}`
    if (CAS(ctr,c,c+1)) `\label{ln:inc_cas}`
       return c;    
  } }
\end{lstlisting}
\vspace{-2mm}
\lstset{firstnumber=8}
\begin{lstlisting}
int decrement() {
  while (true) {
    int c = ctr; `\label{ln:dec_read}`
    if ( c == 0 ) return 0;
    if (CAS(ctr,c,c-1)) `\label{ln:dec_cas}`
      return c;    
  }  }
\end{lstlisting}
\end{minipage}
\vspace{-2mm}
\caption{\label{fig:counter} A concurrent counter.}
\vspace{-2mm}
\end{wrapfigure}
for incrementing and decrementing.
Both methods of the counter return the value of the counter before modifying it, and the counter is 
decremented only if it is strictly positive. 

Each method consists of a retry-loop that reads the shared variable \code{ctr} representing the counter and tries to update it using a Compare-And-Swap (CAS). 
A CAS atomically tests whether \code{ctr} equals the second argument and if this is the case, then it assigns the value specified by the third argument. If the test fails, then the CAS has no effect. The return value of CAS represents the truth value of
the equality test. If the CAS is unsuccessful, \ie~it
returns \emph{false}, then the method retries the same steps in another iteration.

The executions of the concurrent counter are interleavings of an arbitrary number of increment or 
decrement invocations from an arbitrary number of threads. 
Each invocation executes a number of retry-loop iterations until reaching the \code{return}. An iteration corresponds to a control-flow path that starts at the beginning of the loop and ends with a \code{return} or goes back to the beginning. For instance, the increment method consists of two possible iterations:
\begin{enumerate}
\item \code{c} = \code{ctr}; \code{CAS(ctr, c, c+1)}; \code{return c}, and
\item \code{c} = \code{ctr}; \code{assume(ctr != c)}.
\end{enumerate}
Iteration \#1 is called \emph{successful} because it contains a successful CAS, and the unsuccessful CAS in the iteration \#2 is written as an \code{assume} that blocks if the condition is not satisfied.

An invocation can execute more iterations if \code{ctr} is modified by another thread in between reading it at line~\ref{ln:inc_read} or~\ref{ln:dec_read} and executing the CAS at line~\ref{ln:inc_cas} or~\ref{ln:dec_cas}, respectively. Fig.~\ref{fig:counter-execution} pictures an execution with 3 increments that execute between 1 and 3 retry-loop iterations. The first iteration of threads 2 and 3 contains unsuccessful CASs because thread 1 executed a successful CAS and modified \code{ctr}, and these invocations must retry, execute more iterations. Note that there are unboundedly many such executions and, even with bounded threads, exponentially many interleavings.

\begin{figure}[t]
\vspace{2mm}
\includegraphics[width=0.95\textwidth]{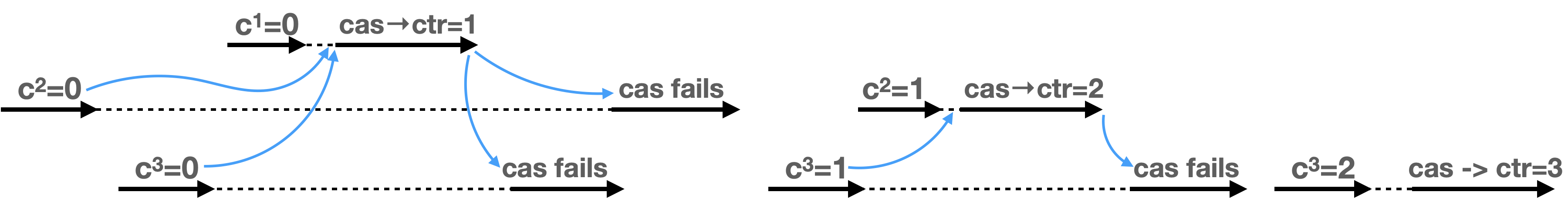}
\caption{
The steps of an execution with three increment-only threads whose actions are aligned horizontally. For readability, we rename the local variable \code{c} in thread $i$ to $\code{c}^i$. The curved blue arrows depict data-flow dependencies between reads/writes of \code{ctr}.}
\label{fig:counter-execution}
\end{figure}

\noindent
{\bf Concurrent Object Syntax.}
We model concurrent objects using Kleene Algebra with Tests~\cite{kozen97} (KAT). Intuitively, a KAT represents the code of an object method using regular expressions over symbols that represent conditionals (tests) or statements (actions).

\begin{definition}\label{def:KAT}[Kleene Algebra with Tests]
A KAT $\Kat$ is a two-sorted structure $(\Akat,\bkat,+,\;\kdot\;,*,\overline{\;},0,1)$, where $(\Akat,+,\;\kdot\;,*,0,1)$ is a Kleene algebra, $(\bkat,+,\;\kdot\;,\overline{\;},0,1)$ is a Boolean algebra, and the latter is a subalgebra of the former.
There are two sets of symbols: $\kAs$ for primitive actions,
and $\kBs$ for primitive tests. The grammar of boolean test expressions is
$BExp ::= b \in \kBs \mid b \kdot b \mid b + b \mid \overline{b} \mid 0 \mid 1$, and the grammar of KAT expressions is
$
KExp ::= a\in\kAs \mid b \in BExp \mid k \kdot k \mid k + k \mid k^* \mid 0 \mid 1
$.
For $k_1,k_2 \in \Kat$, we write $k_1 \leq k_2$ if $k_1 + k_2 = k_2$.
\end{definition}

\newcommand\ttx{\code{x}}
\newcommand\tty{\code{y}}
The primitive actions and tests used in examples in this paper will be along the lines of
$\kAs = \{ \ttx := \tty, \code{x.f := y}, \ldots \}$ and $\kBs = \{  \code{x = y}, \code{x.f = y}, \ttx=\code{null}, \code{x.f}=\code{null} \ldots \}$.

{\emph Atomic read-write (ARW).}
We conservatively extend KAT with a syntactic notation $\ttcondwr{b}{a}$, used to indicate a condition $b$ and action $a$, between which no other actions can be interleaved. Apart from restricting interleaving (defined below), this does not impact the semantics so it can be represented with two special symbols ``$\langle\!\!\!\langle$'' and ``$\rangle\!\!\!\rangle$'' whose semantics are the identity relation. For example a compare-and-swap \texttt{cas(x,v,v')} can be represented as 
$
(\ttcondwrd{[x=v]}{x:=v'}\kdot k) +
(\overline{\code{[x=v]}\kdot k')},
$
where $[x=v]$ is a primitive test and the assignment is a primitive action. Overline indicates negation, as in KAT notation. $k$ is the code to be executed when \texttt{cas} succeeds and $k'$ when it fails.

\noindent
\textbf{Methods of a concurrent object.}
We define a method signature $m(\As)/\rvs$ with a vector of arguments $\As$ and return values $\rvs$ (often a singleton $\rv$). For a vector $\vec{x}$, $x_i$ denotes its $i$-th component.
An \emph{implementation} of a method $m$ is a KAT expression $\kat_m$, whose actions may refer to argument values, \eg~$\ttx := \mathit{args}_i$. 
%
A \emph{concurrent object} $O$ is 
a set of methods $O=\{m_1(\As_1)/\rvs_1:\kat_{m_1},\ldots\}$, associating signatures with implementations. The set of method names in an object $O$ is denoted by $\meth(O)$.

\begin{example}\label{ex:ctrobj} The counter from Sec.~\ref{sec:prelim} is formalized as $O_{ctr} = \{ \texttt{inc}()/v:k_{inc}, \texttt{dec}()/u:k_{dec}\}$ 
\[\begin{array}{lll}
k_{inc} &=& 
(\code{c:=ctr}\kdot \left( (\ttcondwr{[c=ctr]}{ctr:=c+1}\kdot \Ret{\tid}{\code{c}}) + (\overline{[\code{c=ctr}]}) \right) )^* \\
k_{dec} &=& 
(\code{c:=ctr}\tightkdot 
\left(
  ([\code{c=0}]\tightkdot\Ret{\tid}{0})
 \! + \!
  (\overline{[\code{c=0}]}\tightkdot \ttcondwr{[c=ctr]}{ctr:=c-1}\tightkdot \Ret{\tid}{\code{c}})\!+\!(\overline{[\code{c=ctr}]})
\right))^*
\end{array}
\]
The outer $^*$ in $k_{inc}$ corresponds to the \code{while (true)} loop in the method \code{increment} while the inner $+$ corresponds to the two branches of the conditional. The KAT expression $k_{inc}$ represents every control-flow path of \code{increment} which goes a number of times through the assignment \code{c:=ctr} and the ``\code{false}'' branch of the conditional before succeeding the atomic read-write and returning (other sequences represented by this regular expression, \eg, iterating multiple times through the atomic read-write and return will be excluded when defining the semantics).
\end{example}


\noindent
{\bf Concurrent Object Semantics.}
A full semantics for these concurrent objects is given \inExtOrApx{apx:semantics}.
In brief, the semantics involves 
local states $\lstate \in \Lstate$,
shared states $\gstate \in \Gstate$,
and nondeterministic thread-local transition relation
$\lstate,\gstate,k \downarrow_\lbl \lstate',\gstate',k'$, which optionally involve label $\lbl$ ($k$ and $k'$ are KAT expressions representing code to be executed).
These \dt{labels} are taken from the set of possible labels 
$\Lbls \subseteq \kAs \cup \kBs \cup \Inv{\tid}{m(\vec{v})} \cup \Ret{\tid}{\rvs} \cup \condwr{b}{a}$
which includes primitive actions, primitive tests, call actions, return actions or ARWs. (We here write $\Inv{\tid}{m(\vec{v})}$ with free variables to refer to the set of all call actions and similar for returns and ARWs.)
Next, a configuration $C = (\gstate,T)$ where $T : \tids \rightharpoonup (\Lstate\times (\Kat\cup\{\bot\}))$ comprises a shared state $\gstate\in\Gstate$ and a mapping for each active thread to its local state and current code.  We use $\tids$ to denote the set of thread ids, which is equipped with a total order $<$. 
%
Configurations of an object transition according to the relation  $\exectrans{\_}: \Conf \times (\tids\times\Lbls) \times \Conf$, labeled with a thread id and a label. 


%

An object $O$ is acted on by a finite \dt{environment} $\env: \tids \rightarrow O \times \vec{\mathit{Val}}$, specifying which threads invoke which methods, with which argument values. $\mathit{Val}$ denotes a set of values and $\vec{\mathit{Val}}$ denotes the set of tuples of values. We assume that object methods can not access thread identifiers (which is true for concurrent objects defined in the literature) and therefore, each invocation is assumed to be executed by a different thread.
An \dt{execution} of $O$ 
in the environment $\env$ 
is a sequence of labeled transitions between configurations $C_0 \exectrans{\_} \ldots \exectrans{\_} C_n$
that starts in the initial configuration $C_0$ w.r.t. $\env$ and ends in configuration $C_n$. 
A  configuration $C_f=({\gstate}^f,T^f)$ is \dt{final} iff $T^f(\tid) = (\lstate,\bot)$, for some $\lstate$, for all $\tid \in \dom{T^f}$. An execution is \dt{completed} if it ends in a final configuration.
%
$\sem{O \otimes \env}$ denotes the set of completed executions of $O$ in the environment $\env$.
A \dt{trace} $\tr\in\traces$ is a sequence of $\tids\times\Lbls$ pairs, \ie~thread-indexed labels $\tid_0\!:\!\lbl_0,\ldots,\tid_n\!:\!\lbl_n$.
A trace of an execution $\exec$ denoted $\tr_\exec$ is a projection of the thread-indexed labels out of the transitions in the execution.

The \dt{semantics} $\sem{O}$ of a concurrent object $O$ is defined as the set of traces under all possible environments (\ie for any number of threads invoking any methods with any inputs). Formally,
$\sem{O} = 
\{ \tr_\exec \mid \mbox{$\exec\in\sem{O\otimes \env}$, for some environment $\env$}\}.
$

\noindent
{\bf Linearizability} 
For an object $O$, 
an \dt{operation} symbol (or operation for short) $o=m(\avals)/\rvals$ represents an invocation of a method $m\in \meth(O)$
with signature $m(\As)/\rvs$, where $\avals$ is a vector of values for the corresponding arguments $\As$, and $\rvals$ is a vector of values for the corresponding returns $\rvs$.
A \dt{sequential specification} $S$ for an object $O$ is a set of sequences over 
operation symbols. For instance, the sequential specification for the counter object includes sequences of increments and decrements corresponding to executions where each invocation executes in isolation, \eg $\texttt{inc}()/0\cdot \texttt{inc}()/1 \cdot \texttt{inc}()/2$ or $\texttt{inc}()/0\cdot \texttt{dec}()/1 \cdot \texttt{dec}()/0$.

A trace $\tr$ of an object $O$ is \dt{linearizable} w.r.t. a specification $S$ if there exists a (linearization-point) mapping $lp(\tr) : \tids \rightarrow \mathbb{N}$ where the label at position (index) $lp(\tr)$ in $\tr$ is considered to be the so-called \emph{linearization point} of $\tid$'s invocation, and must satisfy the following:
\begin{enumerate}
	\item\label{item:lin1} the position $lp(\tr)$ is after $\tid$'s invocation label and before $\tid$'s return,
	\item the (linearization) sequence $\lin{\tr,lp}$ of operation symbols $m(\avals)/\rvals$, where the $i$-th symbol represents the invocation of the $i$-th thread $\tid$ w.r.t. the positions $lp(\tr,\tid)$, belongs to $S$.
\end{enumerate}
For example, Fig.~\ref{fig:counter-execution} pictures a trace which is linearizable w.r.t. the counter specification described above because there exists a linearization-point mapping $lp$ which associates each thread $i$ with the position of the $i$-th successful CAS. The linearization $\texttt{inc}()/0\cdot \texttt{inc}()/1 \cdot \texttt{inc}()/2$ induced by this mapping is admitted by the specification.

For simplicity, we omit invocation labels from traces and consider the first instruction in an invocation to play the same role. 
%
Object $O$ is \dt{linearizable} wrt a spec.~$S$ if all traces in $\sem{O}$ are linearizable wrt $S$.


\section{Object Quotients}
\label{sec:quotients}

To formalize scenarios, we introduce the concept of a \emph{quotient} of an object which is a subset of its traces that represents every other trace modulo reordering of commutative steps or renaming thread ids. For an expert reader, the quotient is a partial order reduction~\cite{DBLP:books/sp/Godefroid96} composed with a symmetry reduction~\cite{DBLP:conf/cav/ClarkeEJS98} of its set of traces. In general, an object may admit multiple quotients, but as we show later, there exist quotients which can be finitely-represented using regular expressions or extensions thereof. We interpret scenarios as components (sub-expressions) of these finite representations.

Two executions $\exec_1$ and $\exec_2$ are \dt{equivalent up to commutativity}, denoted as $\exec_1 \eeq \exec_2$, if $\exec_2$ can be obtained from $\exec_1$ (or vice-versa) by repeatedly swapping adjacent commutative steps. An execution $\exec_2$ is obtained from $\exec_1$ through one swap of adjacent commutative steps, denoted as $\exec_1\eeq_1\exec_2$, if
\[
\begin{array}{l}
\exec_1 = 
C_0^\env  \cdots C_i \exectrans{(\tid:\lbl)} C_{i+1} 
\exectrans{(\tid':\lbl')} C_{i+2} \cdots C_n, \text{ and }
\exec_2 = 
C_0^\env  \cdots C_i \exectrans{(\tid':\lbl')} C'_{i+1} 
\exectrans{(\tid:\lbl)} C_{i+2} \cdots C_n \\[-.5mm]
\end{array}
\]
($\exec_2$ is obtained from $\exec_1$ by re-ordering the steps labeled by $\tid:\lbl$ and $\tid':\lbl'$). When there exist executions $\exec_1$ and $\exec_2$ as above, we say that the re-ordered labels $\lbl$ and $\lbl'$ are \dt{possibly commutative}.

\begin{definition} 
The equivalence relation $\eeq \subseteq \execs \times \execs$ between executions is the least reflexive-transitive relation that includes $\eeq_1$. 
\end{definition}

\noindent
The relation $\eeq$ is extended to traces as expected: $\tr_1\eeq \tr_2$ if $\tr_1$ and $\tr_2$ are traces of executions $\exec_1$ and $\exec_2$, respectively, and $\exec_1 \eeq \exec_2$.

For example, the Counter executions below are equivalent up to commutativity (related by $\eeq_1$):
\[
\exec = C_0 \cdots C_1
\exectrans{(t:\overline{[\code{c}_{t} \code{=ctr}]})}
C_2
\exectrans{(t':\code{c}_{t'} \code{:=ctr})}
C_3 \cdots \mbox{ and }
\exec' = 
C_0 \cdots C_1
\exectrans{(t':\code{c}_{t'} \code{:=ctr})}
C'_2
\exectrans{(t:\overline{[\code{c}_{t} \code{=ctr}]})}
C_3 \cdots
\]
assuming that $\code{ctr}>0$ at configuration $C_1$ (recall that $\overline{[\code{c}_{t} \code{=ctr}]}$ represents an unsuccessful CAS).

\begin{definition} 
Two traces $\tr_1$ and $\tr_2$ are \emph{equivalent up to thread renaming}, denoted as $\tr_1\symeq\tr_2$, if there is a bijection $\alpha$ between thread ids in $\tr_1$ and $\tr_2$, resp., s.t. $\tr_2$ is the trace obtained from $\tr_1$ by replacing every thread id label $\tid$ with $\alpha(\tid)$. 
\end{definition}

\noindent
For example, $C_0\exectrans{(t:a)}C_1\exectrans{(t':b)}C_2$
and $C_0\exectrans{(t':a)}C_1\exectrans{(t:b)}C_2$ are equivalent up to thread renaming.

We define a quotient of an object as a subset of its traces that is \emph{complete} in the sense that it represents every other trace up to commutative reorderings or thread renaming, and that is \emph{optimal} in that sense that it does not contain two traces that are equivalent up to commutativity. Optimality does \emph{not} include equivalence up to thread renaming (symmetry reduction) because the finite representations we define later abstract away thread ids.

\begin{definition}[Quotient]\label{def:quotient}
A \emph{quotient} of object $O$ is a set of traces $\equo{O} \subseteq \sem{O}$ such that:
\begin{itemize}
\item  $\forall \tr\in\sem{O}.\ \exists \tr', \tr''.
\tr \symeq \tr' \wedge \tr' \eeq \tr''
 \wedge \tr''\in\equo{O}$ (completeness), and
 \item $\forall \tr, \tr'\in\equo{O}.\; \tr \neeq \tr'$ (optimality)
\end{itemize}
\end{definition}

Note that an object admits multiple quotients since representatives of equivalence classes w.r.t. $\equiv$ can be chosen arbitrarily.

For a quotient $\equo{O}$, a set $\mathsf{Swaps}$ of pairs of possibly-commutative labels (in $\Lbls\times\Lbls$) is called \dt{$\equo{O}$-sufficient} if all the swaps needed to establish $\tr' \eeq \tr''$ in Def.~\ref{def:quotient} are between pairs of labels in $\mathsf{Swaps}$.

\begin{example}[Quotient and representative/canonical traces for the Counter]\label{ex:counter-phases}
The trace of three increment-only threads from Fig.~\ref{fig:counter-execution} represents many other traces of the Counter modulo commutative reorderings or thread renaming. It can be thought of as a sequence of three canonical phases, depicted with stacked parallelograms as follows:
\begin{center}
\includegraphics[width=0.95\textwidth]{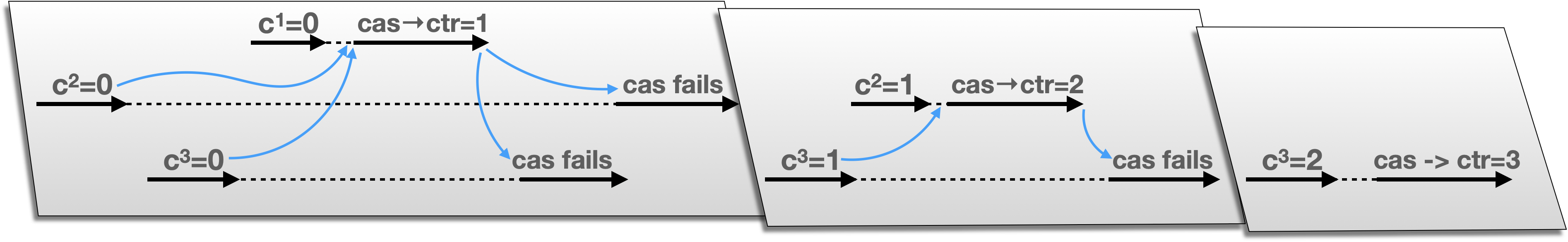}
\end{center}
Each phase above groups together the retry-loop iterations that interact with each other: a single successful CAS instruction causes the other attempts to fail. 
For instance, it represents another trace where the first ``cas fails'' step occurs after the second successful CAS:
\begin{center}
\includegraphics[width=0.80\textwidth]{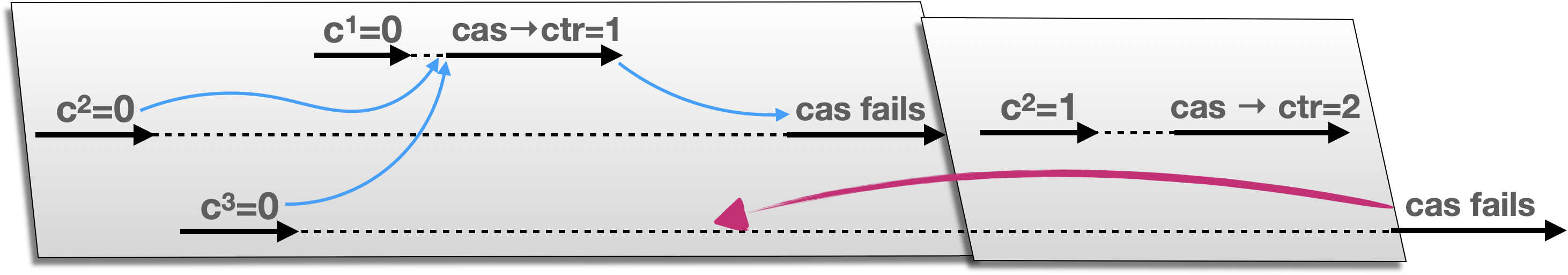}\\
\end{center}
This ``late'' CAS failure would also fail if moved to the left as shown above.
Similarly, it also represents traces where the action $c^2=0$ is swapped with $c^3=0$ and even $c^1=0$, 
or traces where thread ids change from 1, 2, 3 to 4, 5, 6 for instance.

One can define a quotient $\equo{O_{ctr}}$ of Counter which includes representative traces of this form. The representative traces only differ in the number of incrementers/decrementers and the order in which they succeed their CASs. $\equo{O_{ctr}}$ will contain similar canonical traces for, say, an environment with 4 incrementers, 2 decrementers acting in the sequence $incr;decr;decr;incr;incr;incr$ (wherein the second $decr$ does nothing). See Example~\ref{ex:counter-expr} for a more precise description. 
\end{example}

\noindent
{\bf Preserving Linearizability Through Commutative Reorderings.}
Our goal is to reduce the problem of proving linearizability for all traces of an object to proving linearizability only for traces in a quotient. Therefore, given two traces $\tr$ and $\tr'$ that are equivalent up to commutativity ($\tr\equiv\tr'$), where for instance, $\tr$ would be part of a quotient, an important question is whether the linearizability of $\tr$ implies the linearizability of $\tr'$. We show that this holds provided that the reordering allowed by the equivalence $\equiv$ is consistent with a commutativity relation between operations in the specification.

Given a specification $S$, two operations $o_1$ and $o_2$ are \dt{$S$-commutative} when $\eta_1 \cdot o_1\cdot o_2\cdot \eta_2\in S$ iff $\eta_1 \cdot o_2\cdot o_1\cdot \eta_2\in S$,  for every $\eta_1$, $\eta_2$ sequences of operations. 
Given a set of pairs of labels $\mathsf{Swaps}\subseteq \Lbls\times \Lbls$, a linearization point mapping $lp(\tr)$ of a trace $\tr$ is \dt{robust against $\mathsf{Swaps}$-reorderings} if for every two threads $\tid_1$ and $\tid_2$, if the linearization points of $\tid_1$ and $\tid_2$ form a pair in $\mathsf{Swaps}$, then the operations of $\tid_1$ and $\tid_2$ are $S$-commutative.

\begin{theorem}\label{th:com_lin}
Let $\tr\eeq \tr'$ be two equivalent traces, such that $\tr'$ is obtained from $\tr$ by swapping pairs of labels in some set $\mathsf{Swaps}$. If $\tr$ is linearizable w.r.t. some specification $S$ via a linearization point mapping $lp(\tr)$ that is robust against $\mathsf{Swaps}$-reorderings, then $\tr'$ is linearizable w.r.t. $S$.
\end{theorem}
The above holds by defining $lp(\tr')$ by 
$lp(\tr')(\tid)$ = the index in $\tr'$ of the label $lp(\tr)(\tid)$, for every $\tid$.

Theorem~\ref{th:com_lin} implies that proving linearizability for an object $O$ reduces to proving linearizability only for the traces in a quotient $\equo{O}$ of $O$, provided that the used linearization point mappings are robust against $\mathsf{Swaps}$-reorderings for a set $\mathsf{Swaps}$ which is $\equo{O}$-sufficient (thread renaming does not affect this reduction because specifications are agnostic to thread ids).

\section{Finite Abstract Representations of Quotients}\label{sec:finite_reps}

We define finite representations of sets of traces, quotients in particular, which resemble regular expressions and which denote context-free languages over a finite alphabet. The finite alphabet is obtained by projecting out thread ids from labels in a trace. As we show in the evaluation section, scenarios in previous informal proofs of many concurrent objects correspond to components of these expressions, and linearization points can be identified directly within such expressions.

Let $\exprSet$ be the set of expressions $\expr$ defined by the following grammar
\begin{align*}
\expr & = \seqlbl\mid \seqlbl_1^n\cdot \expr\cdot \seqlbl_2^n \mid\ \expr^* \mid\ \expr + \expr\mid \expr \cdot \expr 
\end{align*}
such that $\seqlbl,\seqlbl_1,\seqlbl_2\in (\kAs \cup \kBs \cup \condwr{b}{a})^*$ are finite sequences of labels, and for every application of the production rule $\seqlbl_1^n\cdot \expr\cdot \seqlbl_2^n$, $n$ is a fresh variable not occurring in $\expr$ (this ensures context-free abstractions). Therefore, for every expression in $\exprSet$, a variable $n$ is used exactly twice.

Such expressions have a natural interpretation as context-free languages by interpreting $^*$, $+$, and $\cdot$ as the Kleene star, union, and concatenation in regular expressions, and interpreting every $\seqlbl_1^n\cdot \expr\cdot \seqlbl_2^n$ as sequences $\seqlbl_1,\ldots,\seqlbl_1\cdot \sem{\expr} \cdot \seqlbl_2,\ldots,\seqlbl_2$ where the number of $\seqlbl_1$ repetitions on the left of $\expr$'s interpretation, denoted as $\sem{\expr}$, equals the number of $\seqlbl_2$ repetitions on the right. 

We define an interpretation $\sem{\expr}$ of expressions $\expr$ as sets of \emph{traces}, which differs from the above only in the interpretation of $\seqlbl$, $\seqlbl^*$, and $\seqlbl_1^n\cdot \expr\cdot \seqlbl_2^n$, for finite sequences of labels $\seqlbl,\seqlbl_1,\seqlbl_2$. 

\begin{definition}[Interpretation of an expression]\label{def:interpexpr}
For an expression $\expr$, 
\[\begin{array}{llp{4.0in}}
\sem{\seqlbl}&=&$\{\tid:\seqlbl\mid \tid\in\tids\}$, where $\tid:\seqlbl$ means that all the labels in $\seqlbl$ are associated with the same thread id $\tid$,\\
\sem{\seqlbl^*}&=& $\{\tid_0:\seqlbl, \ldots, \tid_k:\seqlbl \mid k\in\mathbb{N}, \tid_0<\ldots<\tid_k\}$, sequences of labels associated with increasing thread ids,\\
\sem{\seqlbl_1^n\cdot \expr\cdot \seqlbl_2^n}&=&$\{\tid_0:\seqlbl_1, \ldots, \tid_k:\seqlbl_1, \sem{\expr},\tid_k:\seqlbl_2, \ldots, \tid_0:\seqlbl_2\mid k\in\mathbb{N}, \tid_0<\ldots<\tid_k\}$, sequences of labels where the same sequence of increasing thread ids is associated to $\seqlbl_1$ and $\seqlbl_2$ repetitions (in reverse order), respectively,\\
\sem{\expr^*}&=&$\sem{\expr},\ldots,\sem{\expr}$, sequences of repetitions of $\sem{\expr}$,\\
\sem{\expr_1+\expr_2}&=&$\sem{\expr_1} \cup \sem{\expr_2}$, union of interpretations, and\\
\sem{\expr_1 \cdot \expr_2}&=&$\sem{\expr_1}, \sem{\expr_2}$, concatenation of interpretations.
\end{array}\]
\end{definition}

\noindent
For example, in the first case of Def.~\ref{def:interpexpr}, $\{ (t:\code{x:=v}), (t:\code{x++}) \} \in \sem{\code{x:=v} \cdot \code{x++}}$.
For an expression $
(\code{x:=r}^n
\cdot
\code{y:=s}^m
\cdot
\code{skip}
\cdot
\code{s:=y+1}^m
\cdot
\code{r:=x+1}^n
)
$, its interpretation includes traces such as
\[
(t_1:\code{x:=r}), 
(t_2:\code{x:=r}), 
(t_3:\code{y:=s}), 
(t_4:\code{skip}),
(t_3:\code{s:=y+1}), 
(t_2:\code{r:=x+1}), 
(t_1:\code{r:=x+1})
\]

\begin{definition}[Abstractions of quotients]\label{def:expr_abs_quo}
An expression $\expr\in\exprSet$ is called an \dt{abstraction} of an object quotient $\equo{O}$ if $\equo{O}\subseteq \sem{\expr}$. 
\end{definition}

\begin{figure}[t]\centering
\vspace{2mm}
\includegraphics[width=0.95\textwidth]{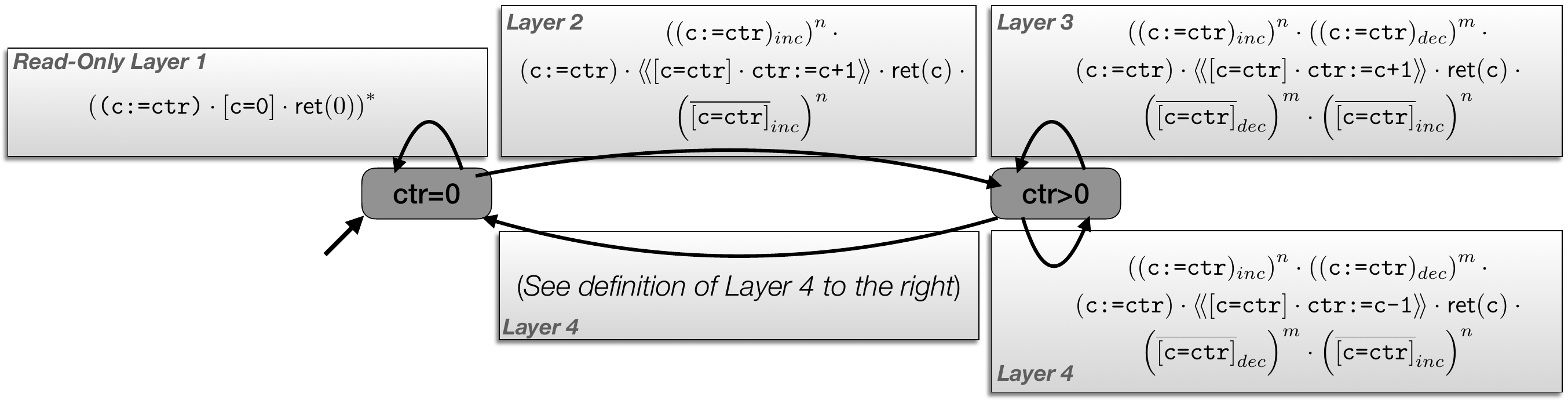}
\caption{An expression representing a quotient of the Counter. 
For readability we present it as four sub-expressions called ``layers'' whose composition with regular expression operators (concatenation, union, star) is represented using an automaton (all states are accepting). 
The full formal definitions of an example layer---from the quotient expression grammar---is given in Example \ref{ex:layer}.
In this figure, for conciseness, we subscript the primitives to indicate whether they were from increment-vs-decrement. Layer 1 represents decrements acting alone and finding the counter to be 0, Layer 2 corresponds to the first successful increment, Layer 3 and Layer 4 represent successful increments and decrements. For Layers 2 -- 4, some number $x$ of threads begin to read 
then a single different thread performs its complete write path, and then all $x$ threads fail their CAS instructions.
Technically, Layer 2 is a specialization of Layer 3, by letting $m=0$. However, treating them as separate layers provides a more refined representation.}
\label{fig:layers_automaton_formal}
\end{figure}

\begin{example}[Abstraction of a Quotient of the Counter]\label{ex:counter-expr}
An expression representing a quotient of the counter is given in Figure~\ref{fig:layers_automaton_formal}. 
%
The following trace is in the interpretation of this expression (for readability, we split the trace across lines, with segments labeled by layer names):
\newcommand\tightcolon{\!:\!}
\[\begin{array}{ll}
\text{Layer 2}: &
 \tid_2\tightcolon(c:=ctr) \tightkdot \tid_3\tightcolon(c:=ctr) \tightkdot \left(
\tid_1\tightcolon(c:=ctr) \tightkdot  \tid_1\tightcolon\condwr{[c=ctr]}{ctr:=c+1}  \tightkdot \tid_1\tightcolon\Ret{\tid_1}{0}  \right)\tightkdot\\
&\;\;\;    
\tid_3\tightcolon\overline{[c = ctr]} \tightkdot 
\tid_2\tightcolon\overline{[c = ctr]}\ \tightkdot
\\
\text{Layer 3}: &
\tid_3\tightcolon(c:=ctr) \kdot
\tid_2\tightcolon(c:=ctr) \kdot
  \tid_2\tightcolon\condwr{[c=ctr]}{ctr:=c+1} \kdot
\tid_2\tightcolon\Ret{\tid_2}{1}\kdot
\tid_3\tightcolon\overline{[c = ctr]} \kdot\\
\text{Layer 3}: &
\tid_3\tightcolon(c:=ctr) \kdot
  \tid_3\tightcolon\condwr{[c=ctr]}{ctr:=c+1} \kdot
  \tid_3\tightcolon\Ret{\tid_3}{2}
  \end{array}
\]

\emph{Linearizability.} Each layer corresponds to linearizing a single \emph{effectful} invocation, \ie~an increment invocation or a decrement invocation when the counter is non-zero, or an arbitrary number of \emph{read-only} invocations, \ie~decrement invocations when the counter is zero.
\end{example}

\section{Layers: An Inductive Quotient Language}
\label{sec:layers}

We show that, for a broad class of objects, 
we can provide a subclass of quotient abstraction expressions---that we will call \emph{layer expressions}---which, via an inductive argument, reduce reasoning about all executions (and about linearizability) to two-threads. 
This applies to numerous canonical examples such as Treiber Stack, the Michael-Scott Queue, a linked-list Set, and even the SLS Reservation Queue. For illustrative purposes, we will continue to use the concurrent Counter, whose quotient can also be expressed with layers.


Many lock-free\footnote{Lock-freedom requires that at least one thread makes progress, if threads are run sufficiently long. A slow/halted thread may not block others, unlike when using locks.} objects rely on a form of optimistic concurrency control where an operation repeatedly reads the shared-memory state in order to prepare an update that reflects the specification and tries to apply a possible update using an atomic read-write. The condition of the atomic read-write checks for possible interference from other threads since reading the shared-memory state.
The executions of such objects can be seen as sequences of what we call
``layers,'' each one being a triple consisting of 
(i) many threads all performing 
commutative local (\eg~read) actions,
(ii) a single non-commutative atomic read-write ARW on the shared state,
and (iii) those same initial threads reacting to the ARW with
more local commutative actions.
For example, incrementing the counter involves a successful cas operation on the shared variable,
which leads to other threads' old reads to go down a failure/restart path.
In fact, with this layer language one can consider
an arbitrary number of control-flow paths executed by an arbitrary number of threads where at most one can contain an atomic read-write. In the remainder of this section we discuss this in detail and then discuss automated discover of layers in Sec.~\ref{sec:automation}.

\subsection{Local-vs-Write Paths}

For an implementation $\Inv{\tid}{m(\As)}\kdot\kat_m\in\Kat$ of a method $m(\As)/\rvs$, a \emph{full (control-flow) path} of $\kat_m$ is a KAT expression $\kat$ such that $\kat \leq\kat_m$ and $\kat$ contains only primitive actions, tests or ARWs, composed together with $\kdot$ ($\kat$ contains no $+$ or $^*$ constructor). In a representation with control-flow graphs of $m$'s code, $\kat$ corresponds to a path from the entry point to the exit point.
A \emph{path} is any contiguous subsequence $\kat'$ of a full path $\kat$, i.e., there exists (possibly empty) $\kat_1$ and $\kat_2$ such that $\kat=\kat_1\kdot\kat'\kdot\kat_2$.
The set of paths of method $m$ is denoted by $\Pths(m)$, 
and as a straightforward extension, the set of paths of an object $O$ defined by a set of methods $m_i$ with $1\leq i\leq n$ is defined as $\Pths(O)=\bigcup_{1\leq i\leq n}\Pths(m_i)$. $\Pths_f(O)$ denotes the subset of \emph{full} paths in $\Pths(O)$.

A primitive action is called \emph{local} when it cannot affect actions or tests executed by another thread (atomic read-writes included), e.g., it represents a read of a shared variable or it reads/writes a memory region that has been allocated but not yet connected to a shared data structure (this region is still \emph{owned} by the thread). 
Formally, let $\sem{a}:(\Lstate \times \Gstate) \rightarrow (\Lstate \times \Gstate)$ and $\sem{b}:(\Lstate \times \Gstate) \rightarrow \{\mathit{true},\mathit{false}\}$ denote the functions defining the semantics of actions $a\in \kAs$ and tests $b\in \kBs$. Then, an action $a\in \kAs$ is \emph{local} iff for every $(\lstate',\gstate') = \sem{a}(\lstate,\gstate)$ and every $s\in \kAs\cup\kBs$ that occurs in some method implementation, $\sem{s}(\lstate'',\gstate)=\sem{s}(\lstate'',\gstate')$, for every local state $\lstate''$. 

A path is called \emph{local} if it contains only local actions, and a \emph{write path}, otherwise. 
Given a KAT expression $\kat'$ that represents a path,
we use $\fst(\kat')$ and $\lst(\kat')$ to denote the first and the last action or test in $\kat'$, respectively. 

\begin{example}\label{ex:counter_full_paths} Returning to the counter object $O_{ctr}$, the full paths are as follows:
\[
\begin{array}{ll}
\code{(c:=ctr)} \cdot \overline{[\code{c=ctr}]} & \code{(c:=ctr)} \cdot [\code{c=0}] \cdot \Ret{\tid}{0}\\
\code{(c:=ctr)} \cdot \ttcondwrd{{[c=ctr]}}{ctr:=c+1} \cdot \Ret{\tid}{\ttc} & \code{(c:=ctr)} \cdot \overline{[\code{c=ctr}]} \\
& \code{(c:=ctr)} \cdot \ttcondwrd{{[c=ctr]}}{ctr:=c-1} \cdot \Ret{\tid}{\ttc}\\
\end{array}
\]
The first two paths are from $k_{inc}$ and the last three are from $k_{dec}$. Paths without ARWs consist of only \emph{local} actions, that may read global \ttctr, but they do not mutate any global variables. 
\end{example}

\newcommand\prek{\overleftarrow{\kat}}
\newcommand\sufk{\overrightarrow{\kat}}
\newcommand\canonLayer{\left((\prek_1)^{n_1}\tightkdot (\prek_2)^{n_2}\cdots (\prek_N)^{n_N}\right) \kdot \kat_w\kdot \left((\sufk_N)^{n_N}\tightkdot (\sufk_{N-1})^{n_{N-1}}\cdots (\sufk_1)^{n_1}\right)}

\subsection{The Language of Layers}

We now define layer expressions and discuss how they represent an object's quotient.

\begin{definition}[Basic Layer Expressions]\label{def:layer_exprs}
A \emph{basic layer expression} $\lambda$ has one of two forms:
\begin{itemize}
\item \emph{local layer}: $(\kat_l)^*$ where $\kat_l$ is a \emph{local} path in  $\Pths(O)$.
\item \emph{write layer}: $\canonLayer$, where 
\begin{enumerate}
\item $\kat_w$ is a write path in $\Pths(O)$,
\item for each $j\in[1,N], \prek_{j}\cdot \sufk_{j}$ is a \emph{local} path in $\Pths(O)$ and the prefix and suffix are each repeated $n_j$ times,
\item $\lst(\prek_{j})$ and $\fst(\sufk_{j})$ do not commute with respect to the ARW in $\kat_w$.
\end{enumerate}

\end{itemize}
\end{definition}
\noindent
The first type, \emph{local layers}, represent 
unboundedly many threads executing a local path $\kat_l$.
Since each instance of the path is local, they all commute with each other, so the interpretation puts them into a single, canonical order which follows the increasing order between their thread ids (by the interpretation of $*$ in quotient expressions; see Def.~\ref{def:interpexpr}).

The second type, \emph{write layers}, represents an interleaving where 
threads execute $n_j$ read-only prefix $\prek_{j}$ of paths (in a canonical, serial order), then 
a different thread executes a non-local path $\kat_w$, and then $n_j$ corresponding suffixes 
$\sufk_{j}$ occur, finishing their iteration reacting to the write of $\kat_w$.
Again, the interpretation $\sem{\lyr}$ of a write layer associates these KAT action labels with increasing thread ids.
Prefixes and suffixes of local paths can be assumed to execute serially as in the first type of layer. The non-commutativity constraint ensures that such an interleaving is ``meaningful'', \ie~it is not equivalent to one in which complete paths are executed serially. 

A \dt{layer expression} is a collection of basic layer expressions, combined in a regular way via $\kdot, +,$ or $*$ (defined in Sec.~\ref{sec:finite_reps}). That is, a layer expression represents complete traces as sequences of layers. 


\begin{example}\label{ex:layer} 
The expression given in Fig.~\ref{fig:layers_automaton_formal} representing a quotient of the Counter is a layer expression. It combines a single read-only layer with other three write layers.
One layer $\lambda_{inc}$ pertains to the increment write path, along with the local paths that fail their CAS attempts. Here, we consider full paths. This basic expression $\lambda_{inc}$ is:
\[\lambda_{inc} \equiv \left[ \begin{array}{l}
  \left(\odot\begin{array}{l}
        (\code{c:=ctr})_{inc}\\
        (\code{c:=ctr})_{dec}
        \end{array}\right)
\tightkdot (\code{c:=ctr})\cdot\condwr{[\code{c=ctr}]}{\code{ctr:=c+1}}\tightkdot\Ret{\tid}{\code{c}} \tightkdot
  \left(\odot\begin{array}{l}
        \overline{[\code{c=ctr}]}_{dec}\\
        \overline{[\code{c=ctr}]}_{inc}
        \end{array}\right)
\end{array}\right]^\circledast \]
This layer interleaves the write path between prefixes/suffixes of the two local paths. We subscript the primitives to indicate whether they were from increment-vs-decrement.  
The $\fst$ and $\lst$ actions/tests do not commute with the interleaved writer's ARW.

\end{example}


\noindent
\textbf{Support of a Layer.} 
The \emph{support} of a basic layer expression $\lambda$, denoted by $\supp(\lambda)$, is defined as a set of KAT expressions where a single prefix/suffix local path is concretized to a single occurrence, and interleaved with the write path.
Intuitively, the support of a write layer characterizes all of the
pair-wise interference by representing interleavings of two paths executed by different threads.

\begin{definition} For basic layer expression $\lambda$, $\supp(\lambda)$ is defined as:
\begin{itemize} 
\item If $\lambda$ is a local layer $\lambda=(\kat_l)^*$,
then $\supp(\lambda) = \{ \kat_l \}$.
\item If $\lambda$ is a write layer $\lambda=\canonLayer$,\\
then $\supp(\lambda) = \{ 
\prek_{j}\kdot \kat_w\kdot \sufk_{j}
\mid j \in [1,n]
\}$.
\end{itemize}
\end{definition}

\begin{example}\label{ex:support} 
For Layer 3 in Fig.~\ref{fig:layers_automaton_formal} involving the increment write path 
$\kat_w=(\code{c:=ctr})\cdot\condwrd{[\code{c=ctr}]}{\code{ctr:=c+1}}\cdot\Ret{\tid}{\code{c}}$, $\supp(\text{Layer 3})=
\{(\ttc\code{:=ctr})_{inc}\kdot\kat_w\kdot \overline{[\ttc\code{=ctr}]}_{inc}, 
(\ttc\code{:=ctr})_{dec} \kdot\kat_w\kdot \overline{[\ttc\code{=ctr}]}_{dec}\}$.
Here there are only two elements of the support, the first being a local path through increment and the second being a local path through decrement. 
\end{example}

The paths $\Pths(\lambda)$ of a basic layer expression $\lambda$ are defined from its support: (1) if $\lambda$ is a local layer, then $\Pths(\lambda)=\supp(\lambda)$, and (2) if $\lambda$ is a write layer, then $\{\kat_w, \prek_{j}\kdot \sufk_{j}\}\subseteq \Pths(\lambda)$ iff $\prek_{j}\kdot \kat_w\kdot \sufk_{j}$ is included in $\supp(\lambda)$. The paths $\Pths(\expr)$ of a layer expression $\expr$ is obtained as the union of $\Pths(\lambda)$ for every basic layer expression $\lambda$ in $\expr$.


\subsection{Proof Methodology with Two-Thread Reasoning}

Recall that layer expressions represent languages of traces so we now ask whether a given expression is an abstraction of an object's quotient (Def.~\ref{def:expr_abs_quo}).
That is: whether each execution $\exec$ of an object is equivalent to some execution $\exec'\eeq\exec$, where the trace of $\exec'$ is in the interpretation of the expression.

Interestingly, this can be done by considering only two threads at a time,
since local paths do not affect the feasibility of a trace. Therefore, it is sufficient to focus on interleavings between a \emph{single} local or write path $\kat$ (on a first thread) and a sequence $\vec{\kat}_w$ of (possibly different) write paths (on a second thread), and show that they can be reordered as a sequence of layers, i.e., $\kat$ executes in isolation if it is a write path, and interleaved with at most one other write path in $\vec{\kat}_w$, otherwise (it is a local path). Applying such a reordering for each path $\kat$ while ignoring other local paths makes it possible to group paths into layers. The reordering must preserve a stronger notion of equivalence defined as follows: two executions $\rho$ and $\rho'$ are \emph{strongly equivalent} if they are $\eeq$-equivalent, they start and resp., end in the same configuration, and they go through the same sequence of shared states modulo stuttering. This notion of equivalence guarantees that any local path enabled in the context of an arbitrary interleaving between $\kat$ and $\vec{\kat}_w$ remains enabled in the context of an interleaving where for instance, $\kat$ executes in isolation. 
A more detailed proof for the following theorem is given \inExtOrApx{apx:proof_methodology_sound}.

\begin{theorem}
\label{def:completeness_methodology}
Let $O$ be an object defined by a set of methods $m_i$ with implementations $\Inv{\tid}{m_i(\As)}\kdot\kat_{m_i}\in\Kat$.
A layer expression $\expr=(\lambda_1+\ldots+\lambda_n)^*$ is an abstraction of a quotient of $O$ if 
\begin{itemize}
\setlength\itemsep{1mm}
\item the layers cover all statements in the implementation: $\Pths(\expr)\subseteq \Pths(O)$ and for each primitive action, test or ARW $\kat_p$ in $\kat_{m_i}$ for some $i$, there exists a path in $\Pths(\expr)$ which contains $\kat_p$,
\item for every path $\kat\in\Pths(\expr)$ and every execution $\exec$ of $O$ starting in a reachable configuration that represents\footnote{An execution $\exec$ represents an interleaving $\kat\,\mid\mid\,\vec{\kat}_w$ if it interleaves two sequences of steps labeled with symbols in $\kat$ and $\vec{\kat}_w$, respectively (in the same order). An execution $\rho$ represents a path sequence $\vec{\kat}$ when it is a sequence of steps labeled with symbols in $\vec{\kat}$ (in the same order).}  an interleaving $\kat\,\mid\mid\,\vec{\kat}_w$, where $\vec{\kat}_w$ is a sequence of write paths in $\Pths(\expr)$, 
\begin{itemize}
\item Write Path Condition (WPC): if $\kat$ is a write path, there is an exec.~$\exec'$ of $O$ s.t. $\exec'$ is strongly equivalent to $\exec$, and 
$\exec'$ represents a 
write path sequence $\vec{\kat}_w^1 \kdot \kat\kdot \vec{\kat}_w^2$ where $\vec{\kat}_w=\vec{\kat}_w^1\kdot\vec{\kat}_w^2$,
\item Local Path Condition (LPC): if $\kat$ is a local path, there exists an execution $\rho'$ of $O$ such that $\exec'$ is strongly equivalent to $\exec$ and 
\begin{itemize}
\item $\exec'$ 
represents a path sequence $\vec{\kat}_w^1 \kdot \kat\kdot \vec{\kat}_w^2$ where $\vec{\kat}_w=\vec{\kat}_w^1\kdot\vec{\kat}_w^2$ ($\kat$ executes in isolation) and $\kat$ is the support of a local layer $\lambda_j$, $1\leq j\leq n$, or 
\item a sequence $\vec{\kat}_w^1 \kdot \kat_l^1\kdot \kat_w\kdot \kat_l^2\kdot \vec{\kat}_w^2$ where $\vec{\kat}_w=\vec{\kat}_w^1\kdot\kat_w\kdot\vec{\kat}_w^2$ and $\kat_w$ is a write path ($\kat$ interleaves with a single write path $\kat_w$), and $\kat_l^1\kdot \kat_w\kdot \kat_l^2\in \supp(\lambda_j)$ for some write layer $\lambda_j$, $1\leq j\leq n$.
\end{itemize}
\end{itemize}
\end{itemize}
\end{theorem}

\begin{example}[Counter layers via two-thread reasoning]
We now proceed to show that the \emph{starred union} of the basic layer expressions defined in Fig.~\ref{fig:layers_automaton_formal} is an abstraction of a quotient.
Concerning WPC, a write path is of the form $(\code{c:=ctr})\cdot\condwrd{[\code{c=ctr}]}{\code{ctr:=c+1}}\cdot\Ret{\tid}{\code{c}}$. Such paths can be reordered to execute in isolation because the ARW is enabled only if the counter did not change its value since the read, and therefore, the read $\code{c:=ctr}$ can be reordered after any step of another thread that may occur until the ARW. Also, the return action is local and can be reordered to occur immediately after the ARW.
%
LPC holds because any ``late'' CAS failure (that occurs after more than one successful CAS) would also fail if moved to the left (as explained in Example~\ref{ex:counter-phases}).
\end{example}

\subsection{Automaton Representation of Layer Quotients}

A layer expression comprised simply of a starred union of basic layer expressions is not always appealing since some layers are not enabled from some configurations.
For instance, as shown in Figure~\ref{fig:layers_automaton_formal} for the Counter, the read-only ``decrement returning 0'' layer cannot occur after one successful increment layer. 
(In formal notation, layer $\lambda_{dec0}$ of $O_{ctr}$ in Example~\ref{ex:layer} is enabled only when $\code{ctr}$ is 0.)
In other words, the starred starred union composition of layers
can be refined further to enforce certain orders in which layers can occur, by taking into account reachability.

We therefore describe a more convenient representation as a \emph{layer automaton}, in which the automaton states represent abstractions (sets) of concrete configurations in executions (as defined in Sec.~\ref{sec:prelim}) and the transitions are labeled by basic layer expressions. 
Another example of such an automaton was seen for the Michael-Scott queue
in Fig.~\ref{fig:aut-msq} in Sec.~\ref{sec:intro}.
Briefly, the control states correspond to the
configurations of the objects (\eg, whether the MSQ is empty, tail is lagged, etc.),
and the transitions are labeled by basic layer expressions (\eg, the ``\textsf{Dequeue Succeed}'' layer from Fig.~\ref{fig:aut-msq}, in which one thread succeeds a CAS on the \code{head} pointer and other threads fail their CAS). 
These layer automata are a convenient representation of the quotient and, as shown in Sec.~\ref{sec:automation}, we can derive candidate 
layer quotients represented as layer automata automatically from source code.




\begin{definition}[Layer automaton]
Given an object $O$, a \emph{layer automaton} is a tuple 
$
\mathcal{A} = (\QQ, \QQ_0,\Lyrs, \trans)
$
where $\QQ$ is a finite set of states representing abstractions (sets) of configurations of $O$, $\QQ_0\subseteq \QQ$ is the set of initial states, and $\delta \subseteq \QQ \times 2^{\Lyrs} \times \QQ$ is a set of transitions labeled with basic layer expressions (elements of $\Lyrs$) with the constraint that an edge $q \xrightarrow{\alpha} q'$ can only be one of two types:
\begin{enumerate}
\item Unique self-loop: $\alpha=\lyr_1\cdots\lyr_n$ is a sequence of $n\geq 1$ local layers,  
$q'=q$, and there are no other self-loops $q\xrightarrow{\alpha'}q$. 
\item Single write layer edges: $\alpha=\lyr$ is a single write layer.
\end{enumerate}
\end{definition}

\noindent

The \dt{interpretation} of the automaton, denoted by $\sem{\mathcal{A}}$, as a layer expression is defined as expected, except that the label of a self-loop is not starred. For instance, the interpretation of an automaton consisting of a single state $q$ and self-loop $q\xrightarrow{\alpha}q$ is defined as $\alpha$ instead of $\alpha^*$.

\begin{theorem}\label{th:layer_automaton}
Given an object $O$ and a layer automaton $\mathcal{A}=(\QQ, \QQ_0,\Lyrs, \trans)$, the layer expression $\sem{\mathcal{A}}$ is an abstraction of a quotient of $O$ if 
\begin{itemize}
	\item the starred union of the basic layer expressions labeling transitions of $\mathcal{A}$ is an abstraction of a quotient of $O$ (Theorem~\ref{def:completeness_methodology}), 
	\item every initial configuration of $O$ is represented by some abstract state in $\QQ_0$, and every reachable configuration is represented by some abstract state in $\QQ$,
	\item for every layer $\lambda$ in $\sem{\mathcal{A}}$, if there exists an execution $\rho$ representing $\lambda$ from a reachable configuration $C$ to a configuration $C'$, then $\mathcal{A}$ contains a transition $q\xrightarrow{\alpha'}q$ where $q$ is an abstraction of $C$ and $q'$ is an abstraction of $C'$.
\end{itemize}
\end{theorem}
%
%

The automaton in 
Fig.~\ref{fig:aut-msq} is a layer automaton for the MSQ (see Section~\ref{ssec:MSQ} for more details).


\begin{corollary} (To Thm.~\ref{th:com_lin})
If a layer expression $\expr$ is an abstraction of a quotient and 
there is a linearization point mapping for every 
trace in $\sem{\expr}$ that is robust against re-ordering,
then the object is linearizable.
\end{corollary}


\section{Evaluation: Verifying Concurrent Objects}
\label{sec:cases}

As discussed in Sec.~\ref{sec:intro}, our goal is to provide a formal foundation for the scenario-based linearizability correctness arguments found in the distributed computing literature. 
To evaluate whether quotients serve that purpose, we examined several diverse and challenging concurrent objects,
listed below. 

\noindent
\begin{center}
\begin{tabular}{l|l|l}
\toprule
{\bf Concurrent Object} & {\bf Quotient} & {\bf Features} \\
\midrule
Atomic counter  & Sec.~\ref{sec:prelim} & simple cas loop \\
\citet{conf/podc/MichaelS96} queue & {\bf Sec.~\ref{ssec:MSQ}} & many cas, cleanup helping \\
\citet{Scherer2006} queue & {\bf Sec.~\ref{subsec:resqueue}} & synchronous, mult.~writes, LP helping \\
\cite{IBMTR:Treiber86}'s stack & {\bf Sec.~\ref{subsec:elim}} & simple cas loop \\
\citet{DBLP:conf/spaa/HendlerSY04} stack   & {\bf Sec.~\ref{subsec:elim}} & elimination, submodule, LP helping\\
\citet{DBLP:conf/wdag/HarrisFP02} RDCSS & {\bf Sec.~\ref{subsec:rdcss}} & mult.~cas steps, phases \\
\citet{DBLP:journals/toplas/HerlihyW90} queue & {\bf Sec.~\ref{subsec:hwq}} & future-dependent LPs \\
\citet{DBLP:conf/podc/OHearnRVYY10} set  & \ARXIV{Apx.~\ref{apx:listset}}\CONF{Ext.~Ver.} & lock-free traversal\\
\bottomrule
\end{tabular}
\end{center}
\noindent
For each object, we (i) determine whether quotients can be used for verification and (ii) 
revisit the scenario-based correctness arguments given by the object's authors and compare those arguments to the quotient. We discuss the quotients of most in this section (with bold {\bf Sec 6.\_} in the {\bf Quotient} column); further detail can be found 
\ARXIV{in Apx.~\ref{apx:resqueue-quo}--\ref{apx:revisit}.}\CONF{in the appendix of the extended version~\cite{arxiv}.}

\emph{Results summary.} As we show, all above algorithms can be captured with quotient expressions. These expressions 
(i) capture the diverse features/complexities of these algorithms (per the {\bf Features} column),
(ii) provide a succinct, formal foundation for the scenario-based arguments used by those objects' authors,
(iii) organize unbounded interleavings into a form more amenable to reasoning,
(iv) make explicit the relationship between implementation-level contention/interference and ADT-level transitions,
and (v) provide a scenario proof for HWQ which did not have scenario arguments.

\subsection{The Michael/Scott Queue}\label{ssec:MSQ}

Recall the implementation of MSQ, stored as a linked list from global pointers \code{Q.head} and \code{Q.tail}, and manipulated as follows. (Some local variable definitions omitted for lack of space.)
\begin{center}
\lstset{xleftmargin=3.5ex,numbersep=2pt}
\begin{tabular}{l|l|l}
\begin{minipage}{1.56in}\footnotesize
\begin{lstlisting}[morekeywords={ret,loop}]
int enq(int v){ loop {
 node_t *node=...;
 node->val=v;
 tail=Q.tail;  
 next=tail->next;
 if (Q.tail==tail) {  
  if (next==null) {
   if (CAS(&tail->next,
          next,node))
      ret 1;
} } } }
\end{lstlisting}
\end{minipage}
&
\begin{minipage}{1.7in}\footnotesize
\begin{lstlisting}[morekeywords={ret,loop}]
int deq(){ loop { 
 int pval;
 head=Q.head;tail=Q.tail;
 next=head->next;
 if (Q.head==head) { 
  if (head==tail) {
   if (next==null) ret 0;
  } else {
    pval=next->val;
    if (CAS(&Q->head,
            head,next))
      ret pval;
  } } } } }
\end{lstlisting}
\end{minipage}
&
\begin{minipage}{1.7in}\footnotesize
Factored out 
tail advancement:\\
(see notes below)\\
\begin{lstlisting}[morekeywords={ret,loop}]
adv(){ loop {
 tail=Q.tail; 
 next=tail->next;
 if (next!=null){
  if (CAS(&Q->tail,
    tail,next)) 
   ret 0;
 }
} }
\end{lstlisting}
\end{minipage}\\
\end{tabular}
\end{center}

\newcommand\ttenq{\code{enq}}
\newcommand\ttdeq{\code{deq}}
\newcommand\ttadv{\code{adv}}
\newcommand\ttQtail{\code{Q.tail}}
\newcommand\ttQtailnext{\code{Q.tail->next}}
\newcommand\ttQhead{\code{Q.head}}
\newcommand\ttQheadnext{\code{Q.head->next}}
Values are stored in the nodes between \ttQhead\ and \ttQtail, with \ttenq\ adding new elements to the \ttQtail, and \ttdeq\ removing elements from \ttQhead.
During a successful CAS in \ttenq, the \ttQtailnext\ pointer is changed from null to the new node. However, this new item cannot be dequeued until \ttadv\ advances \ttQtail\ forward to point to the new node. A \ttdeq\ on an empty list (when \ttQhead=\ttQtail) returns immediately. Otherwise, \ttdeq\ attempts to advance \ttQhead\ and, if success, returns the value in the now-omitted node. The original MSQ implementation includes the \ttadv\ CAS inside \ttenq\ and \ttdeq\ iterations. 
We have done this for expository purposes and it is not necessary. As we will see in Sec.~\ref{subsec:resqueue}, the SLS queue performs this tail (and head) advancing directly in the enqueue/dequeue method implementation.


\emph{Quotient.}
The layer automaton that abstracts a quotient of MSQ, mentioned briefly in Sec.~\ref{sec:intro}, is shown in Fig.~\ref{fig:aut-msq}.
The automaton states track whether \code{Q.tail=Q.head} and whether \code{Q.tail->next} is \texttt{null}, in rounded dark boxes. 
Edges are labeled with layers (discussed below), defined to the right in Fig.~\ref{fig:aut-msq}.
The write operations in those layers induce the automaton state changes as shown by the various edges between automaton states. For example, the \textsf{Dequeue Succeed} layer can move from automaton state $q_2$ to $q_1$. 
The three layers of the MSQ characterize three forms of interference:
%
\begin{description}
\item[The \textsf{Dequeue Succeed} layer] occurs when a dequeue thread successfully advances the \ttQhead\ pointer, causing concurrent dequeue CAS attempts to fail, as well as dequeue threads checking on Line 5 whether \ttQhead\ has changed. (We abbreviate local paths using line numbers rather than KAT expressions.)
\item[The \textsf{Advancer Succeed} layer] occurs when an advancer moves forward the \ttQtail\ pointer, causing concurrent advancer CAS attempts to fail, and causing concurrent \ttenq\ threads to find \ttQtail\ changed on Line 6.
\item [The \textsf{Enqueue Succeed} layer] occurs when an \ttenq\ thread successfully advances the \ttQtail\ pointer, causing concurrent \ttenq\ threads to fail.
\end{description}

\ARXIV{States are labeled $q_1,\ldots,q_4$.}
\ARXIV{Several transitions in the automaton are labeled with the same layer, so we have given those layers names (\eg~Dequeue Succeed Layer) and provided their definitions in the Legend to the right. These three layers characterize three forms of interference:}

\noindent
Naturally, some edges are not enabled. For example, there is no edge from $q_1$ to $q_2$, because the latter is not reachable from the former via a single write path/layer.
Also, while there are outbound edges from $q_1$, there is no layer involving a \ttdeq\ write operation (since the queue is empty).
Some non-local layers self-loop, such as the \textsf{Dequeue Succeed} layer self-loop at $q_4$.
There are also four \emph{local} layers that self-loop. These involve local paths that return (\eg~Read Only Layer 1 where \ttdeq\ returns because the queue is empty) or paths that loop while waiting (\eg~Read Only Layer 3 where \ttenq\ awaits the advancer thread).
%



\begin{theorem}\label{th:msq-abstraction}
The above layer automaton is an abstraction of a quotient
for Michael-Scott Queue.
\begin{proof}
Proof by the methodology of Def.~\ref{def:completeness_methodology}. 

The WPC condition requires that all write paths (that include successful CASs) can be reordered to execute in isolation. This is a direct consequence of the semantics of a successful CAS which checks that the value did not change since the last read of the written location. The \ttdeq{} successful CAS on \ttQhead{} insures that \ttQhead{} did not change since it was read at Line 3, which also means that its \texttt{next} pointer did not change (this pointer is written only once in \ttenq() for every node in the list). Therefore all actions on the \ttdeq{} path that includes the successful CAS can be reordered to execute together at the place of reading \ttQtail. Similarly the \ttenq{} successful CAS ensures that the actions between Line 5 and Line 8 can be reordered to occur together. Then, since the value of \ttQtail{} could not have changed without \ttQtailnext{} first having been changed, Lines 2-4 can also be reordered to occur together with the rest of the actions on this path. The case of the \ttadv{} write path is similar. 

The LPC condition follows from the fact that CAS operations always change the value so it is always possible to move a late ``failing'' CAS to the left so that it occurs after the first successful CAS following the previous reads in the same iteration. \qed
%
%
\end{proof}
\end{theorem}

\begin{theorem}\label{thm:linearizable-msq} The Michael-Scott Queue is linearizable.
\begin{proof}
We show that the traces in the quotient are linearizable via a linearization-point mapping which is robust against reorderings. Given a trace in the quotient (represented by the automaton in Fig.~\ref{fig:aut-msq}), the linearization points are the successful CAS operations in the \textsf{\{Dequeue,Advancer\} Succeed} layers
(also in {\bf bold} in the Fig.~\ref{fig:aut-msq} layer definitions), as well as the action corresponding to Line 7 in \texttt{deq()} which occurs in \textsf{Read-Only Layer 1}. The successful CAS operations are linearization points of dequeues returning some enqueued value and enqueues, respectively, and Line 7 is the linearization point of a dequeue returning empty. The validity of these linearization points can be proved by induction on the number of layers. The induction hypothesis will relate the last configuration of the quotient execution with a queue ADT state that is the sequence of elements reachable from \ttQhead. For instance, the successful CAS in the \textsf{Dequeue Succeed} layer will remove the first element in such a sequence which by the induction hypothesis is the oldest element in the queue. 
\\
By the proof of the quotient's completeness (Theorem~\ref{th:msq-abstraction}), successful CAS operations are never reordered. The only linearization point labels that can be reordered are those corresponding to Line 7 in \texttt{deq()} for a dequeue returning empty. It is easy to see that dequeues returning empty commute in the queue specification, which implies that the above linearization-point mapping is robust against a set of reorderings which is sufficient for this quotient. \qed
\end{proof}
\end{theorem}

\noindent
{\bf Comparison with the Authors' Proof.}
We evaluated the quotient by comparing with the correctness arguments from~\citet{TAOMPP}.
For lack of space, the following table gives example elements of the correctness argument/proof from~\citet{TAOMPP}, and identifies 
where they occur in the quotient proof (see \ExtOrApx{apx:revisit}  for more details).

\medskip
\noindent
\begin{tabular}{|p{1.2in}|p{1.9in}|p{1.9in}|}
\hline
{\bf Proof Element} & \citet{TAOMPP} & {\bf Quotient Proof} \\
\hline
\hline
ADT states & \footnotesize ``queue is nonempty,'' ``tail is lagged'' & ADT states, e.g. (\texttt{Q.tail=}\texttt{Q.head} $\wedge$ \texttt{Q.tail->}\texttt{next} $\neq$ \texttt{null})\\
\hline
Concurrent threads & \footnotesize ``some other thread'' & Superscripting $(...)^n$ \\
\hline
Event order & \footnotesize ``only then'' & Arcs in the quo automaton\\
\hline
Thread-local step seq. & \footnotesize ``reads tail, and finds the node that appears to be last (Lines 12–13)'' & Layer paths, e.g., \textsf{enq:2-6}\\
\hline
Linearization pts. & \footnotesize ``If this method returns a value, then its linearization point occurs when it completes a successful [CAS] call at Line 38, and otherwise it is linearized at Line 33.'' & The successful CAS in the Dequeue Succeed Layer or Read-Only Layer 1 \\
\hline
\end{tabular}

\medskip
\noindent
The layer quotient and, especially, the layer automaton helps make the \citet{TAOMPP} proof more explicit, without sacrificing the organization of the proof, for a few reasons.
First, all of the important ADT states are explicitly identified.
Second, it can be determined, from each of them, which layers are enabled as well as the target ADT states that are reached after each such layer transition.
This ensures that all cases are considered.
Finally, \emph{linearization points} are explicit in the layer quotient, occurring once with each layer transition.

\subsection{The SLS Synchronous Reservation Queue}
\label{subsec:resqueue}

The~\citet{Scherer2006} (SLS) queue builds on MSQ, 
but has some complications:
queue operations are synchronous (blocking),
a single invocation can involve multiple sequentially composed write paths that necessitate different layers, and linearization points must account for dequeuers arriving before their corresponding enqueuer.

\emph{Implementation.} 
Like  MSQ, SLS has paths that read the \code{head} or \code{tail} pointer and subsequent pointers, perform read validations and then attempt a CAS. Also like MSQ, enqueuers arriving at an empty list (or list of items), attempt to append \emph{item} nodes (and then try to advance the tail pointer). Dequeuers arriving at a list of items, attempt to swap item node contents for null (and then try to advance the head pointer).

SLS then has some further complexities.
Dequeuers arriving at an empty list (or list of reservation nodes) attempt to append \emph{reservation} nodes (and attempt to advance tail).
Enqueuers arriving at a list of reservations, attempt to \emph{fulfill} those reservations by swapping null for an item  (and attempt to advance head).
The list never contains both items and reservations; when the list becomes empty it can then transition from an item list to a reservation list (or vice-versa).
Finally, SLS is \emph{synchronous:} dequeuers with reservations \emph{block} until those reservations have been fulfilled and enqueuers with items \emph{block} until those items have been consumed.
(For the sake of comprehensiveness, the implementation is \inExtOrApx{apx:resqueue}, but not necessary for a general understanding.) As noted, unlike MSQ where paths have at most 1 write operation, a single SLS invocation can perform multiple write operations (\eg~a dequeue path inserting a reservation, advancing tail, awaiting fulfillment, advancing head). Despite conceptual simplicity, the implementation is non-trivial with many restart paths when validations or CAS operations fail.


\usetikzlibrary{shapes.misc, positioning}
\newcommand*\circled[1]{\tikz[baseline=(char.base)]{
            \node[draw,rounded rectangle,fill,inner sep=1pt] (char) {\textcolor{white}{\tt #1}};}}
\newcommand*\ropath[1]{\tikz[baseline=(char.base)]{
            \node[draw,rounded rectangle,color=black,draw,inner sep=1pt] (char) {\textcolor{black}{\tt #1}};}}

\begin{figure}[t]
\vspace{2mm}
\includegraphics[width=0.95\columnwidth]{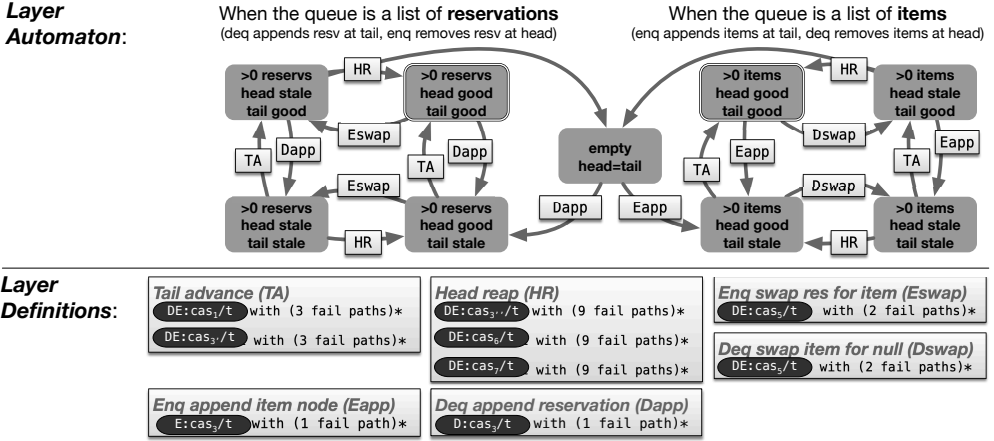}
\caption{\label{fig:resautomaton} Layer automaton for the synchronous
SLS queue. Layers' acronyms and their definitions are given in the lower half of the figure. \ARXIV{States are depicted in dark boxes.} For conciseness, layer definitions do not split the prefix/suffix of the read paths.
}
\end{figure}

\newcommand\Eapp{\textsf{Eapp}}
\newcommand\Dapp{\textsf{Dapp}}
\newcommand\Eswap{\textsf{Eswap}}
\newcommand\Dswap{\textsf{Dswap}}
\newcommand\TA{\textsf{TA}}
\newcommand\HR{\textsf{HR}}

\emph{Quotient.}
The quotient expression for the SLS queue is depicted as a layer automaton in Fig.~\ref{fig:resautomaton}. 
In the upper portion, the automaton \emph{states} differentiate between whether the queue is 
empty or whether the queue consists of reservations (left hand region) or of items (right hand region). In each of those regions, it is relevant as to whether the head pointer is stale or not, as well as whether the tail pointer is stale or not.
When the queue is a list of reservations, the head or tail could be stale (hence four states) and similar when the queue is a list of items. 

The \emph{basic layers} of the quotient expression are defined at the bottom of Fig.~\ref{fig:resautomaton}. 
The black circles (\eg~$\circled{DE:CAS$_\ell$/t}$) represent a write path in which a \texttt{D}equeuer or \texttt{E}nqueuer has successfully performed a CAS at some program location $\ell$. 
Along with the write path, we simply summarize the number of competing read-only paths, which are star-iterated.  
\ARXIV{When there are two write paths (\eg~$\circled{1t}$ and $\circled{3't}$) that have the same effect, they are grouped together (in this case, the EITA layer).}
Two layers are enq/deq-agnostic: advancing the tail pointer in \textsf{TA} and advancing the head pointer (and ``reaping'' the head node) in \textsf{HR}. These helping operations happen in many places in the code, with corresponding read-only ``\texttt{\_f}'' failure paths.
Enqueue can either append an item node (\Eapp) when in the RHS states of the automaton or else swap an item into a reservation node (\Eswap) in the LHS. These layers have a single CAS operation (\eg~$\circled{E:CAS$_5$/t}$) along with read-only paths\ignoreme{(\eg~$\ropath{4+5f}$)} where concurrent competing threads fail. The dequeue layers \Dapp{} and \Dswap{} are similar. 
\removed{There is one further read-only layer \textsf{Notice}, which occurs when a thread notices the node it's spinning on has changed, as when a dequeuer finds its reservation has been fulfilled. As discussed below, this layer provides a linearization point, ensuring the FIFO queue semantics.}

Finally, these (context-free) basic layer expressions are connected into an overall expression, represented here as an automaton or (below) as a star-/plus-/or-combination of layer expressions.


\begin{theorem} The SLS queue is linearizable.
  \begin{proof}  We associate linearization points with layers:
  \Dswap{} is an LP for dequeue, \Eapp{} is an LP for enqueue, and
  \Eswap{} is an LP for a combination of an enqueue followed by a dequeue. Next, we project the linearization points out of the quotient to obtain simply $(E\!\cdot\! D)^* \cdot (E^*\! +\! D^*)$. Combining this with a lemma that this expression is an abstraction of the quotient, we obtain that all executions in the quotient meet the sequential spec.~of a queue. This linearization point mapping is also robust because successful CASs (linearization points) do not have to be swapped in order to prove the completeness of the quotient.
(Detail 
\CONF{\intheextended{}}\ARXIV{in Apx.~\ref{apx:resqueue}, Thm.~\ref{thm:apxresqueue-linearizability}}.) \qed
  \end{proof}

 \end{theorem}



\newcommand\SLS[1]{\QuoteAuthor{#1}{\citet{Scherer2006}}}
\newcommand\SLSCiteless[1]{\QuoteAuthorCiteless{#1}}
\newcommand\Quo{}

\noindent
{\bf Comparison with the Authors' Proof.} We evaluated the SLS quotient expression by revisiting the authors' proof in \citet{Scherer2006}. Line numbers in the authors' quotes below refer to a reproduction of the source code given in \inExtOrApx{apx:resqueue}. For lack of space, some discussion of the authors' quotes can be found \inExtOrApx{apx:resqueue-eval}. 

The authors split the enqueue operation into two linearization points: a ``reservation linearization point'' and
a later ``follow up linearization point,'' so that synchronous, blocking enqueue implementations are a single reservation LP and then repeated follow-up LPs (as if the client is repeatedly checking whether the operation has completed).

\newcommand\refAuthorLine[1]{\CONF{[...]}\ARXIV{\ref{#1}}}

\SLS{[Regarding enqueue,] the reservation linearization point for this code path occurs at line~\refAuthorLine{ln:insertOffer} when we successfully insert our offering into the queue}

\noindent
This prose describes a scenario, (i) identifying an alleged linearization point at $\circled{\ignoreme{E3t}E:cas$_3$/t}$, involving a specific change 
to shared memory (a CAS on the tail's next pointer), and (ii) identifying the important ADT state transition (inserting an offer node into the queue). This scenario is formalized by the \Eapp{} layer in the quotient expression. The successful CAS $\circled{\ignoreme{E3t}E:cas$_3$/t}$ in \Eapp{} is the linearization point, with competing concurrent threads abstracted away by the starred fail path expression\ignoreme{ \texttt{(DE3f)*}}, and the state transition is given in the automaton as the downward \Eapp-labeled arcs in the righthand region of the automaton.
The scenario and LP for dequeue on a list of reservation nodes is symmetric, and represented in the quotient expression as layer \Dapp{} involving $\circled{\ignoreme{D3t}D:cas$_3$/t}$ and competing fail path\ignoreme{\texttt{(DE3f)*}}.

The quotient expression makes the interaction between LPs and ADT states more explicit (\eg~through $LP$-marked layers) and comprehensive (\eg~the authors do not discuss the 9 different automaton ADT states and which transitions are possible from each). The quotient expression can be seen as an abstract view of an implementation of the sequential specification.

%


\noindent
\SLSCiteless{\begin{tabular}{ll}
\begin{minipage}[b]{2.7in}
The other case occurs when the queue consists of reservations
(requests for data), and is depicted [to the right]. 
In this case, after originally reading the head node (step A), we read its successor
(line~\refAuthorLine{ln:readHeadNext}/step B) and verify consistency (line~\refAuthorLine{ln:verifyHeadNext}). Then, we attempt to supply our data to the \underline{head-most reservation} (line~\refAuthorLine{ln:appendHead}/C). If this succeeds, we dequeue the former dummy node (\refAuthorLine{ln:removeDummy}/D) and return
\end{minipage} & 
\begin{minipage}[b]{1.5in}
\includegraphics[height=1.0in]{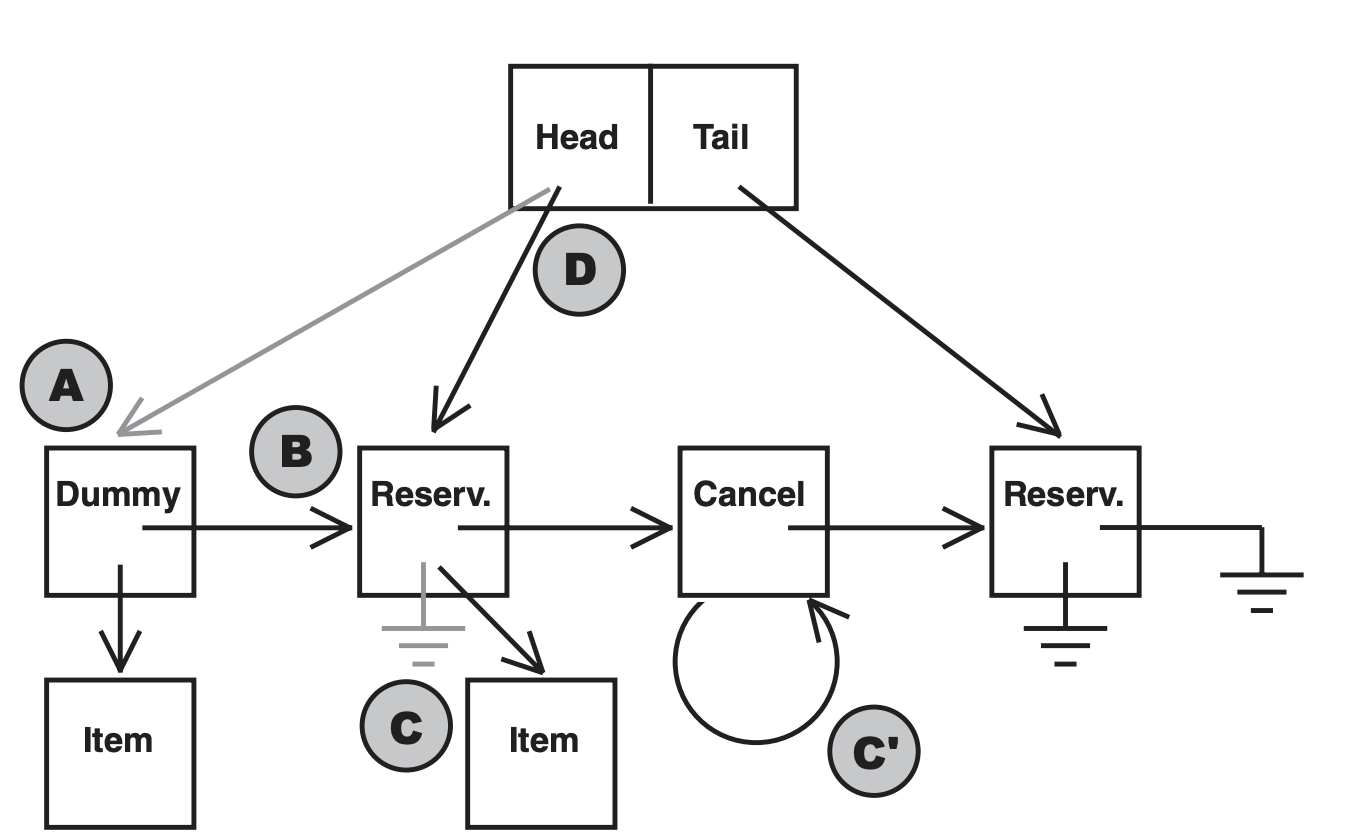}
\end{minipage}
\end{tabular}}

\noindent
This prose again indicates important mutations (\eg~swapping the node's contents pointer),
ADT state changes (\eg~supplying data) and that the head dummy node needs to be advanced.
These memory mutations and state changes are explicit in the quotient expression. For example,
\Eswap{} performs a memory CAS and makes a ADT state transition. The staleness of the head is also captured directly in the ADT states and the \HR{} layers' transitions.
%
The authors' prose also discusses failure paths (see \ExtOrApx{apx:resqueue-eval}) and retry, which are 
also captured in the layer definitions\ignoreme{, for example, $\circled{5f} \cdot (\circled{6bt} + \ropath{6bf}$)}.

\emph{Summary.}
The layer quotient expression/automaton provides a succinct formal foundation for the correctness arguments of~\citet{Scherer2006}, capturing the authors' discussions of LPs, ADTs, impacts of writes, CAS contention, etc.


\subsection{The Hendler et al. Elimination Stack}
\label{subsec:elim}

\begin{figure}
%
\begin{tabular}{ll}
\begin{minipage}{3.0in}
\lstset{xleftmargin=3.0ex,numbersep=2pt}
\begin{lstlisting}[morekeywords={ret}]
void push/pop(descriptor p){ while(1) {
  one iteration of Treiber stack
  location[mytid] = p; `\label{line:descriptor}`
  pos = nondet(); `\label{line:start_publish}`
  do { him = collision[pos]
  } while (!CAS(&collision[pos], him, mytid)) `\label{line:end_publish}`
  if him != NULL { `\label{line:first_test}`
     q = location[him]
     if ( q != NULL & q.id = him & p.op != q.op ) { `\label{line:second_test}`
        if (CAS(&location[mytid],p,NULL)) { `\label{line:first_CAS}`
           if ( CAS (&location[him], q, p/NULL) ) `\label{line:second_CAS}`
              return NULL/q.input
           else continue `\label{line:restart}`
        } else {
           val = NULL/location[mytid].input; `\label{line:start_finish}`
           location[mytid] = NULL;
           return val `\label{line:end_finish}`
  } } } 
  if (!CAS(&location[mytid],p,NULL)) { `\label{line:clear_CAS}`
     val = NULL/location[mytid].data; `\label{line:start_last_finish}`
     location[mytid] = NULL;
     return val `\label{line:end_last_finish}`
}}  }
\end{lstlisting}
\end{minipage}
&
\begin{minipage}{3.0in}
\vspace{2mm}
\includegraphics[width=5.7cm]{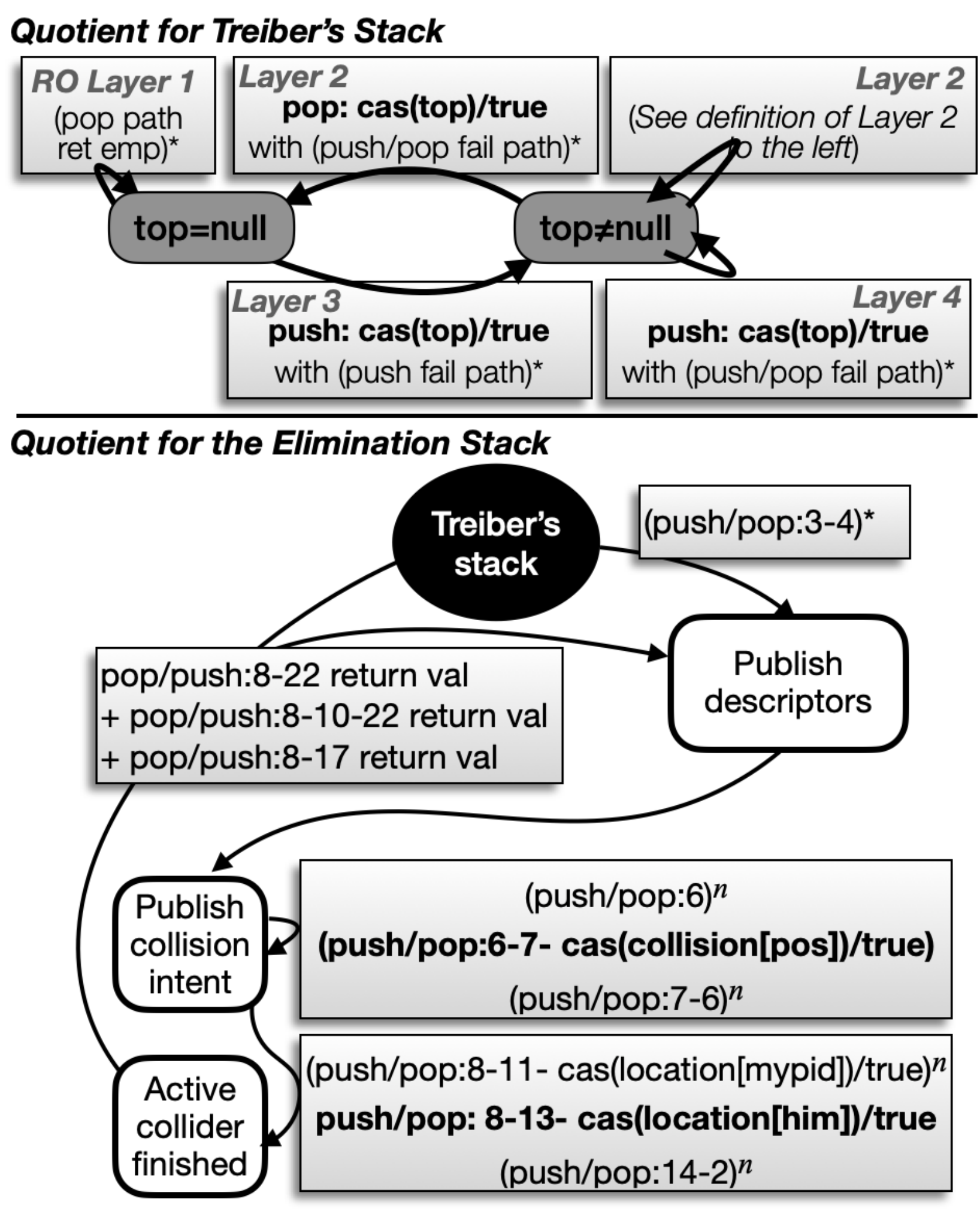}
\end{minipage}\\
(a) Elimination Stack source code
&
(b) Stack Quotients\\
\end{tabular}
\caption{Elimination Stack}
\label{fig:code_stacks}
\vspace{-4mm}
\end{figure}

The Elimination Stack of~\citet{DBLP:conf/spaa/HendlerSY04} is difficult because the linearization point of some invocation can happen in another (threads can awake to find they were linearized earlier) and it uses a submodule: Treiber's stack~\cite{IBMTR:Treiber86}.

We first show the Treiber's stack quotient, and then build elimination on top. 
Since Treiber's stack is simple, we explain only the basics here, with more detail \inExtOrApx{apx:treiber}.
The implementation of
\ttpush{} prepares a new node and then attempts a CAS to swing the \code{top} pointer,
while \ttpop{} attempts to advance the \code{top} pointer and return the removed node's value. 
The quotient for Treiber's stack is shown in the upper right of Fig.~\ref{fig:code_stacks} and is similar to the counter, but with ADT states tracking emptiness (rather than non-zeroness) and CAS contention on the top pointer (rather than the counter cell).
There is one read-only layer for a \ttpop{} and an empty stack, and other layers involve one successful CAS with failed competing CAS attempts. See \ExtOrApx{apx:treiber} for more detail, as well as \ARXIV{Lemma~\ref{lemma:treiber}}\CONF{a lemma} proving that this layer automaton is an abstraction of the quotient.

The Elimination Stack, listed in Fig.~\ref{fig:code_stacks}(a), augments Treiber's stack with a protocol for ``colliding'' push and pop invocations so that the push passes its input directly to the pop without affecting the underlying data structure. An invocation starts this protocol after performing a loop iteration in Treiber's stack and failing (due to contention on \texttt{top}). The protocol uses two arrays: (1) a \texttt{location} array indexed by thread ids where a push or pop invocation publishes a descriptor tuple (\texttt{op,id,input}) with fields \texttt{op} for the type of invocation (push or pop), \texttt{id} for the id of the invoking thread, and \texttt{input} for the input of a push operation, and (2) a \texttt{collision} array indexed by arbitrary integers which stores ids of threads announcing their availability to collide. 

Each invocation starts by publishing their descriptor in the \texttt{location} array (line~\ref{line:descriptor}). Then, it reads a random cell of the \texttt{collision} array while also trying to publish their id at the same index using a CAS (lines~\ref{line:start_publish}--\ref{line:end_publish}). If it reads a non-NULL thread id, then it tries to collide with that thread. A successful collision requires 2 successful CASs on the \texttt{location} cells of the two threads (we require CASs because other threads may compete to collide with one of these two threads): the initiator of the collision needs to clear its cell (line~\ref{line:first_CAS}) and modify the cell of the other thread (line~\ref{line:second_CAS}) to pass its input if the other thread is a pop. The first CAS failing means that a third thread successfully collided with the initiator and the initiator can simply return (lines~\ref{line:start_finish}--\ref{line:end_finish}). Failing the second CAS leads to a restart (line~\ref{line:restart}).
Succeeding the second CAS means there has been a successful collision and the thread returns, returning null for a push and otherwise using the descriptor to obtain the popped value (line~\ref{line:second_CAS}).
 If the invocation reads a NULL thread id from \texttt{collision}, then it tries to clear its cell before restarting (line~\ref{line:clear_CAS}). If it fails, then as in the previous case, a collision happened with a third thread and the current thread can simply return (line~\ref{line:start_last_finish}--\ref{line:end_last_finish}).

\emph{Quotient.}
We use the automaton in the lower right of Fig.~\ref{fig:code_stacks} to describe a sound abstraction of the quotient. Layers of Treiber's stack  interleave with layers of the collision protocol (some components are not exactly layers as in Definition~\ref{def:layer_exprs}, but quite similar).
Executions in the quotient \emph{serialize} collisions and proceed as follows: (1) some number of threads publish their descriptor and choose a cell in the \texttt{collision} array, (2) some number of threads publish their id in the \texttt{collision} array (there may be more than one such thread -- note the self-loop on the ``Publish collision intent'' state), (3) some number of threads succeed the CAS to clear their \texttt{location} cell but only one succeeds to also CAS the \texttt{location} cell of some arbitrary but fixed thread \texttt{him} and return, and (4) the thread \texttt{him} returns after possibly passing the tests at line~\ref{line:first_test} or~\ref{line:second_test}. 
(Note that, for succinctness, we have combined push/pop into the same method, which also makes the automaton succinct. The code and corresponding automaton could also have been written in a more verbose way where the bottommost layer is replaced with two layers: (1) a layer where a push's successful CAS takes with it a corresponding pop, and (2) a layer where a pop's successful CAS takes with it a corresponding push. For succinctness, we have combined those layers using the ``push/pop'' notation.) 
We emphasize that collisions happen in a serial order, i.e., at any point there is exactly one thread that succeeds on both CASs required for a collision and immediately after the collided thread returns (publishing descriptors or collision intent interleaves arbitrarily with collisions).

\begin{theorem} The Elimination Stack is linearizable.
\begin{proof}Follows from the fact that the above expression is an abstraction of the quotient (\ARXIV{Thm.~\ref{th:elimination}}\CONF{See~\citet{arxiv}}), with the {\bf bold} actions in the layers being the LPs. \qed\end{proof}
\end{theorem}

\newcommand\HSYprose[1]{\QuoteAuthor{#1}{\citet{DBLP:conf/spaa/HendlerSY04}}}
\newcommand\HSYproseCiteless[1]{\QuoteAuthorCiteless{#1}}

\noindent
{\bf Comparison with the Authors' Proof.}
A proof is given by~\citet{DBLP:conf/spaa/HendlerSY04} in that paper's Section 5. It is a lengthy proof so, for lack of space, the full review is in \inExtOrApx{apx:elimination-stack} and summarized here. 
Overall, the correctness 
argument requires numerous lemmas in the \citet{DBLP:conf/spaa/HendlerSY04} proof,
mostly focused on establishing a bijection between the active thread and its correspondingly collided passive thread. The authors lay out a few definitions, which are also captured by the quotient. For example, the authors' prose includes:

\HSYprose{[A] colliding operation op is \underline{active} if it executes a successful CAS in lines C2 or C7. We say that a colliding operation is \underline{passive} if op fails in the CAS of line S10 or S19. [underlines added]}

\noindent
Above the authors' intuitive concept of ``active'' is captured by the paths in a layer that succeed their CAS, denoted in {\bf bold} in the quotient automaton above. Likewise for ``passive'' and CAS failure. As mentioned above, the active thread is captured as the bold thread that succeeds its CAS in the bottommost layer; the passive thread is the thread that finds itself collided with in the layers on arcs exiting the bottommost layer.

%
%
%
%

\HSYproseCiteless{we show that push and pop operations are paired correctly during collisions. Lemma 5.7. Every passive collider collides with exactly one active collider.}

\noindent
The bottommost layer in the {\bf bold} action, a single push or pop succeeds, colliding with another operation of the oppose type, and passing the element from the push to the pop.



Authors' LPs are given for ``active'' threads as the time when the second CAS succeeds, and linearization points for ``passive'' threads ``the time of linearization of the matching active-collider operation, and the push colliding-operation is linearized before the pop colliding-operation.'' 
The linearization points in the quotient correspond to the bold successful CAS in the bottommost layer in the quotient automaton (this linearizes both a push and a pop).
Importantly, every run of the quotient automaton gives a serial linearization order that is a
repetition of pairs of active/passive threads. All other executions are equivalent to one such serialized run, upto commutativity.
%
%

In summary, as detailed \inExtOrApx{apx:elimination-stack}, the quotient naturally and succinctly captures the key concept of the Elimination stack: that a single successful CAS of one type of operation is the LP for that operation as well as the corresponding matched operation. 
The quotient captures ``active'' versus ``passive'' threads (in the automaton layers/states/transitions), as well as this bijection through the runs of the automaton: every run in the automaton contains some number of active/passive pairs and provides a representative serialization order (in each pair the push is serialized before the pop). 
\removed{This occurs due to the active thread completing its CAS successfully in the bold action in the bottom right layer, and then the passive thread finding itself to be collided in the layers exiting ``Active collider finished.''
All other executions are equivalent to some automaton run, upto commutativity. Lemma 5.7 is thus captured by this quotient (as are the other Lemmas; see Apx.~\ref{apx:elimination-stack}).}
Linearization points and other logistics of threads preparing/completing are similarly captured by the quotient automaton.

\subsection{The Harris et al. Restricted Double-Compare Single-Swap  (RDCSS)}
\label{subsec:rdcss}

RDCSS~\cite{DBLP:conf/wdag/HarrisFP02} is a restricted version of a double-word CAS 
which modifies a so-called data address provided that this address and another so-called control address have some given expected values (the tests and the write happen atomically). 
%
%
%
\code{RDCSS} attempts a standard CAS on the data address to change the old value into a pointer to a descriptor structure that stores the inputs of the operation. This fails if the data address does not have the expected value. A second standard CAS on the data address is used  to write the new value if the control address has the expected value or the old value, otherwise. Faster threads can help complete the operations of slower threads using the information stored in the descriptor. 

%

The traces in the quotient of \code{RDCSS} interleave successful attempts at modifying the data address with unsuccessful ones. A successful attempt consists of a thread succeeding the first CAS combined with competing threads that fail, followed by another thread succeeding the second CAS (this can be different from the first one in the case of helping) combined with other threads that fail. An unsuccessful attempt may contain just a thread failing the first CAS, or it can contain two successful CASs like a successful attempt (when the data address has the expected value but the control address does not). Proving linearizability of quotient traces is obvious because they make explicit the ``evolution'' of a data address, oscillating between storing values and descriptors, and which CAS is enabled depending on the value of the control address. See \ExtOrApxt{apx:rdcss} for more details.

\subsection{The Herlihy-Wing Queue}
\label{subsec:hwq}

\newcommand{\enqinc}{\code{enqI}}
\newcommand{\enqwr}{\code{enqW}}
\newcommand{\deqT}{\code{deqT}}
\newcommand{\deqF}{\code{deqF}}
\newcommand\ttback{\code{back}}

The quotients of some data structures cannot be represented using layer automata. The Herlihy-Wing Queue~\cite{DBLP:journals/toplas/HerlihyW90} is one such example and it is notorious for  linearization points that depend on the future and that \emph{cannot} be associated to fixed statements, see e.g.~\citet{DBLP:conf/cav/SchellhornWD12}!
The queue is implemented as an array of slots for items, with a shared variable \code{back} that indicates the last possibly non-empty slot. An \code{enq} atomically reads and increments \code{back} and then later stores a value at that location. A \code{deq} repeatedly scans the array looking for the first non-empty slot in a doubly-nested loop.
We show that the Herlihy-Wing queue quotient can
be abstracted by an expression 
$\left( \deqF^*\cdot (\enqinc)^+ \cdot\enqwr^* \cdot \deqT^*\right)^*$, 
where $\deqF$ captures dequeue scans that need to restart,
$\deqT$ scans succeed, $\enqinc$ reads/increments \ttback{} and 
$\enqwr$ writes to the slot.
For lack of space, a detailed discussion about how this expression abstracts
the quotient is given \inExtOrApx{apx:hwq}. Importantly, linearization points in executions represented by this expression are \emph{fixed}, drastically simplifying reasoning from the general case where they are non-fixed.

\begin{theorem}The Herlihy-Wing Queue is linearizable.
(see \ARXIV{Thm.~\ref{th:linearizability-hwq}}\CONF{\citet{arxiv}})
\end{theorem}


\noindent
{\bf Comparison with the Authors' Proof.}
\citet{DBLP:journals/toplas/HerlihyW90} give intuitions of scenarios:

\QuoteAuthor{Enq execution occurs in two steps, which may be interleaved with steps of other concurrent operations: an array slot is reserved by atomically incrementing back, and the new item is stored in items.}{Sec 4.1 of ~\citet{DBLP:journals/toplas/HerlihyW90}}

\noindent
This describes a scenario with unboundedly many threads, though is not yet an argument for why that scenarios is correct. This scenario appears in the quotient as the fact that $\enqinc$ and $\enqwr$ are distinct.
To cope with non-fixed LPs (in this and other objects), the authors introduce a proof methodology based on tracking all possible linearizations that could happen in the future.
This general methododology complicates the proof. The quotient, by contrast, allows one to consider scenarios along the lines of ``one or more enqueuers increment back, possibly some of them write to the array, and then some dequeuers succeed,'' following the quotient's regular expression. 
In summary, the quotient here provides the first scenario-based proof of correctness, through
representative executions that allow the linearization order to be \emph{fixed} and all other executions are equivalent to one such representative execution up to commutativity.




\section{Generating Candidate Quotient Expressions}
\label{sec:automation}
In Sec.~\ref{sec:cases} we showed quotients can be defined for a wide range of concurrent objects, including notoriously difficult ones. 
We leave the (rather large) question of automated quotient proofs for the general case as future work.
Here we take a first step asking,
\emph{Can candidate quotient expressions can be generated algorithmically?}

This section answers this question with an algorithm, implementation and
experiments showing that, from the source code of concurrent data-structures such as Treiber's stack and the MSQ, candidate quotients expressions (equivalent to those in Sec.~\ref{sec:cases}) can be automatically discovered.
We manually confirmed that these generated candidates are indeed sound abstractions of the quotient, a process that can also be automated (perhaps through new forms of induction) in future work.

\removed{The algorithm exploits our reduction to two-thread reasoning and automaton
representation of layer quotients (Apx.~\ref{apx:layerautomata}).}

\subsection{Computing Layer Automata}
\label{apx:layerautomata-alg}

Given a set of layers $\lambda_1$,$\ldots$,$\lambda_n$ whose starred union is an abstraction of an object quotient (cf. Theorem~\ref{def:completeness_methodology}), a layer automaton satisfying Theorem~\ref{th:layer_automaton} can be computed automatically. 
The algorithm consists of the following steps:
\begin{enumerate}
\item \emph{States}: Compute the automaton abstract states as boolean conjunctions of the weakest pre-conditions (and their negations) of traces in the \emph{support} of a layer $\lambda_i$ with $1\leq i\leq n$. We assume that the initial state can be determined from the object spec.

\item \label{alg:step2} \emph{Edges}: Whenever a state $q$ implies the precondition of a write layer $\lyr_i$ with write path $\kat_w$, compute every post-state $q'$ that can hold, and add an edge $q\xrightarrow{\lyr_i}q'$. This can be encoded as an assertion violation in a program that assumes $q;\kat_w$ and asserts the negation of $q'$.

\item \emph{Self-Loops}: For every state $q$ collect every local layer that is enabled from $q$ and create a single self-loop consisting of a concatenation of all these layers.
\end{enumerate}

\removed{
The algorithm is in Apx.~\ref{apx:layerautomata-alg} but, briefly, involves 
(i) computing automaton states using weakest preconditions,
(ii) computing the possible post-states of write paths $\kat_w$, and which 
local paths are feasible interleavings with those write paths 
(exploiting pair-wise reasoning about paths), and 
(iii) computing which automaton self-loops are possible via local-only layers.
}

\begin{table}[t]
\caption{Evaluation of \Tool{} discovering candidate layers from source code.}
\label{fig:eval} 
\vspace{-4mm}
\begin{center}
{
\begin{tabular}{l|ccc|ccrr}
\toprule
 & {\bf States} & \multicolumn{2}{c|}{{\bf \# Paths}} & {\bf \# Trans.} & {\bf \# Layers} & {\bf Time} & {\bf \# Solver} \\
{\bf Example} & $\mid\!\QQ\!\mid$ & {\bf \# $\kat_l$} & {\bf \# $\kat_w$} & $\mid\!\delta\!\mid$ & $\mid\!\Lambda(O)\!\mid$ &  (s) & {\bf Queries}\\
\midrule
\texttt{evenodd.c} &    2 &    2 &    2 &    6 &    3 &  52.2 &   32\\
\texttt{counter.c} &    2 &    3 &    2 &    6 &    5 &  67.8 &   36\\
\texttt{descriptor.c} &    4 &    6 &    2 &    6 &    6 & 160.2 &   74\\
\texttt{treiber.c} &    2 &    3 &    2 &    6 &    5 &  71.4 &   37\\
\texttt{msq.c    } &    4 &    9 &    3 &   17 &    7 & 441.6 &  314\\
\texttt{listset.c} &    7 &    6 &    2 &   59 &    7 & 603.8 &  494\\
\bottomrule
\end{tabular}}
\end{center}
\vspace{-5mm}
\end{table}

\subsection{Implementation and Experiments} 

We built a proof-of-concept implementation of our algorithm, called \Tool{} in $\sim$1,000 lines of OCaml code, using CIL and Ultimate~\cite{DBLP:conf/tacas/HeizmannCDGHLNM18}. \Tool{} is publicly available\footnote{\url{https://github.com/quotientprovers/cion}}.
\ARXIV{
Our implementation consumes an input concurrent data structure implemented in C, with demarcation of atomic read/writes. For simplicity, our implementation requires user-provided automaton states (also written in C), but this can be automated and there are tools for doing so. We parse the implementation, enumerate paths tracking whether they are local-vs-write, parse automaton states, and then begin the algorithm.
\Tool{} constructs  feasibility problems for a solver of the form
$\texttt{assume}(\qq); \kat; \texttt{assert}(\neg \qq')$.
In each problem $\kat$ is an element of the support $\supp(\lambda)$ of the layer, constructed by concatenating the local path prefix, the write path, and the local path suffix, \ie~$\kat_{l1} \kdot \kat_w \kdot \kat_{l2}$.
Each of these problems (feasibility ``queries'') is emitted as a C reachability problem and discharged with Ultimate.}
We applied \Tool{} to some of the Sec.~\ref{sec:cases} objects that were amenable to layers. 
Experiments were run on Ubuntu 22.04 within a Parallels VM on a MacBook Pro M2 with 32GB RAM.
Benchmarks are available in \Tool{} repository. 
We used Ultimate v0.2.1 (\texttt{54a68f4}) as a reachability solver, with its default configuration.
The results are summarized in Table~\ref{fig:eval}.
For each benchmark, we report the number of automaton {\bf States} $\mid\!\QQ\!\mid$,
the number of local {\bf Paths} $\#\kat_l$ and number of write paths $\#\kat_w$. We then report the number of {\bf Trans}itions $\mid\!\delta\!\mid$ in the automata constructed by \Tool{} and the number of 
{\bf Layers}, as well as the wall-clock {\bf Time} in seconds, and the number of {\bf Queries} made to the solver (Ultimate).
\ARXIV{For MSQ, Ultimate needed to use MathSAT for one feasibility check (the 62nd solver query). Since there is no distribution of MathSAT 5 compiled to ARM architecture, we ran Ultimate on an X86 Linux machine for that single check but the solving time for that one query was negligible.}
\ARXIV{
Even with small numbers of states and paths, the tasks are challenging because there are many $\qq,\lambda,\qq'$ triples and, for each one, the support involves all interleavings of all writers. This is why, for example, \texttt{counter} involves 36 queries. As the paths get longer, there are more interleavings to consider, as seen in MSQ and List Set. We found it most efficient to search for a possibly failing reader $\kat_l$ by interleaving $\kat_w$ starting at the end of $\kat_l$ and then progressing backwards to the first interleaving position. (Typically the ``failure'' conditions are checked at the end of a read path.)
We also improved efficiency by searching first for the readers in a $\qq$ local layer and then using only the remaining ones (that are not enabled from $\qq$) in the search for interleavings with write paths.}
The results show that \Tool{} is able to efficiently generate candidate layer automata for some important and challenging concurrent objects.

\section{Related work}\label{sec:related}

\emph{Linearizability proofs.}
Program logics for compositional reasoning about concurrent programs and data structures have been studied extensively, as mentioned in Sec.~\ref{subsec:intro-quo}. 
Improving on the classical \citet{DBLP:journals/cacm/OwickiG76} and Rely-Guarantee~\cite{DBLP:conf/ifip/Jones83} logics, numerous extensions of Concurrent Separation Logic~\cite{DBLP:conf/popl/BornatCOP05,DBLP:conf/concur/Brookes04,DBLP:conf/concur/OHearn04,DBLP:conf/popl/ParkinsonBO07} have been proposed in order to reason compositionally about different instances of fine-grained concurrency, e.g.~\cite{DBLP:journals/jfp/JungKJBBD18,DBLP:journals/pacmpl/JungLPRTDJ20,DBLP:conf/popl/Ley-WildN13,DBLP:conf/ecoop/PintoDG14,DBLP:conf/esop/SergeyNB15,DBLP:journals/pacmpl/Nanevski0DF19,DBLP:conf/esop/RaadVG15,DBLP:conf/icfp/TuronDB13,conf/cav/DragoiGH13,conf/vmcai/Vafeiadis09,phd/Vafeiadis08,DBLP:journals/pacmpl/KrishnaSW18}. 
We build on the success of such program logics toward improving the confidence in the correctness of concurrent objects. In the current paper we alternatively focus on the scenario-based reasoning found in the distributed computing literature, and have aimed to capture those scenarios as formally-defined representative executions.
In future work it could be interesting to combine the benefits of program logics with those of quotients.
%
%
Other more distantly related works include:
\citet{DBLP:conf/cav/BerdineLMRS08},
\citet{DBLP:conf/cav/Vafeiadis10},
\citet{DBLP:conf/esop/BouajjaniEEH13},
\citet{DBLP:journals/corr/ChakrabortyHSV15},
\citet{DBLP:conf/cav/ZhuPJ15}, and
\citet{DBLP:conf/sas/AbdullaJT16}.

\emph{Reduction.}
The reduction theory of~\citet{DBLP:journals/cacm/Lipton75} introduced the concept of \emph{movers}
to define a program transformation that creates atomic blocks of code.
QED~\cite{DBLP:conf/popl/ElmasQT09} expanded Lipton's theory by introducing iterated
application of reduction and abstraction over gated atomic actions.
CIVL~\cite{DBLP:conf/cav/HawblitzelPQT15} builds upon the foundation of QED,
adding invariant reasoning and refinement layers~\cite{DBLP:conf/cav/KraglQ18,DBLP:conf/concur/KraglQH18}.
Reasoning via simplifying program transformations has also been
adopted in the context of mechanized proofs, e.g.,~\cite{DBLP:conf/osdi/ChajedKLZ18}.
Inductive sequentialization~\cite{DBLP:conf/pldi/KraglEHMQ20} builds upon this prior
work, and introduces a new scheme for reasoning inductively over unbounded concurrent executions.
The main focus of these works is to define generic proof rules to prove soundness of such
program transformations, whose application does however require 
carefully-crafted artifacts such as abstractions of program code or invariants.
Our work takes a different approach and tries to distill common syntactic patterns of concurrent objects into
a simpler reduction argument. Our reduction is \emph{not} a form of program transformation
since quotient executions are interleavings of actions in the implementation.

\section{Conclusion}

We have shown that scenario-based reasoning about concurrent objects has a formal grounding, answering an open question.
The key insight is the concept of a quotient, defined so that it admits \emph{only} representative traces and all other traces are merely equivalent to one of those representatives, up to commutativity.
We then gave a language for finitely expressing abstractions of those quotients (as regular or context-free languages) and an inductive and automata-theoretical way of describing them.
Our results show that quotients provide a succinct formal foundation for scenario-based reasoning, are capable of capturing a wide range of tricky objects, enhance original authors' correctness arguments, and that discovery of candidate quotient expressions can be automated.
In the future will explore further mechanization and other application domains.

\section*{Data-Availability Statement}
Software that supports Sec.~\ref{sec:automation}
is available on GitHub~\cite{github} and Zenodo~\cite{zenodo}.

\begin{acks}
We thank Matthew Parkinson and the anonymous reviewers for their feedback on this draft. Koskinen was partially supported by 
\grantsponsor{GSNSF}{NSF}{https://www.nsf.gov} award \grantnum{GSNSF}{CCF-2008633} and 
\grantsponsor{GSNSF}{NSF}{https://www.nsf.gov} award \grantnum{GSNSF}{CCF-2315363}.
Enea was partially supported by 
\grantsponsor{GSANR}{ANR}{https://anr.fr} award \grantnum{ANR}{SCEPROOF}.
\end{acks}


\bibliography{biblio,refs-oopsla20,refs-pldi20,dblp}

\ifarxiv
\vfill\pagebreak
\appendix

\begin{center}
{\Large\bf Scenario-based Proofs for Concurrent Objects}

\medskip
{\large\bf Appendix}
\end{center}

\bigskip
\section{Unabridged Concurrent Object Semantics}
\label{apx:semantics}

We define an operational semantics for concurrent objects as sets of executions that interleave steps of a number of method invocations executed by different threads. For simplicity, we assume that every thread invokes a single method, which is without loss of generality provided that thread ids are modeled as additional inputs. Method implementations are assumed to be given as KAT expressions, as described in Section~\ref{sec:prelim}. 

\smartpar{Client environments}
An object $O$ is acted on by a finite \dt{environment} $\env: \tids \rightarrow O \times \overline{\mathit{Val}}$, specifying which threads invoke which methods, with which argument values. We use $\tids$ to denote the set of thread ids, and $\mathit{Val}$ an unspecified set of values ($\overline{\mathit{Val}}$ denotes the set of tuples of values). We assume that $\tids$ is equipped with a total order $<$ that will be used to define representatives of equivalence classes up to symmetry (renaming of thread ids).

\smartpar{States}
We assume that each test or action in a method implementation acts on a local state, whose content can be accessed only by the thread executing that test/action, and possibly a shared state which can be read or modified by any thread in the environment. As expected, we assume that each local state contains a valuation for the arguments of an invocation $\As$ and, once a thread has finished its execution, its local state contains the return values $\vec{\rv}$. A precise formalization of local/shared states is irrelevant to our development and we omit it for readability. Let $\Lstate$ and $\Gstate$ denote the set of local and shared states, respectively.

A thread executes an implementation given by a KAT expression $k$, according to the rules below.
We assume that semantics of tests $\sem{b}:(\Lstate \times \Gstate) \rightarrow \mathbb{B}$ and actions $\sem{a}:(\Lstate \times \Gstate) \rightarrow (\Lstate \times \Gstate)$ is provided (or generated from the language/program). 

\smartpar{Non-deterministic single-thread execution} 
Given an environment $\env$, a step of a single thread $\tid$ is a relation on 
$\Lstate\times\Gstate\times(\Kat\cup\{\bot\})$ where $\bot$ indicates that a thread has completed.
We denote this relation as $\lstate,\gstate,k \downarrow_\lbl \lstate',\gstate',k'$, which optionally involve label $\lbl$.
\dt{Labels} are taken from the set of possible labels 
$\Lbls \subseteq \kAs \cup \kBs \cup \Inv{\tid}{m(\vec{v})} \cup \Ret{\tid}{\rvs} \cup \condwr{b}{a}$
which includes primitive actions, primitive tests, invocations, returns or ARWs. (We here write $\Inv{\tid}{m(\vec{v})}$ with free variables to refer to the set of all invocations and similar for returns and ARWs.)
The labeled single-step semantics are now defined inductively on $k$ as follows:
$$
\infer{\lstate^0,\gstate,\epsilon \downarrow_{\Inv{\tid}{m(\vec{v})}} \lstate^0[\mathit{args}_i\mapsto v_i],\gstate,k_m}{\env(\tid)=(m(\As)/\rvs:k_m,\vec{v})}
\qquad
\infer{\lstate,\gstate,\Ret{\tid}{\rvs} \downarrow_{\Ret{\tid}{\vec{v}}} \lstate,\gstate,\bot}{\vec{v} = \lstate(\rvs)}
$$
$$
\infer{\lstate,\gstate,k + k' \downarrow \lstate,\gstate,k}{}
\qquad
\infer{\lstate,\gstate,k + k' \downarrow \lstate,\gstate,k'}{}
\qquad
\infer{\lstate,\gstate,k \kdot k' \downarrow \lstate',\gstate',k'}{
\lstate,\gstate,k \downarrow \lstate',\gstate',1
}
$$
$$
\infer{\lstate,\gstate,k^* \downarrow \lstate,\gstate,k\kdot k^*}{}
\qquad
\infer{\lstate,\gstate,k^* \downarrow \lstate,\gstate,1}{}
$$
$$
\infer{\lstate,\gstate,a \downarrow_{a} \lstate',\gstate',1}{
a \neq \condLL
&
(\lstate',\gstate') = \sem{a}(\lstate,\gstate)
}
\qquad
\infer{\lstate,\gstate,b \downarrow_{b} \lstate,\gstate,1}{\sem{b}(\lstate,\gstate) = \textsf{true}}
$$
$$
\infer{\lstate,\gstate,\condwr{b}{a} \downarrow_{\condwr{b}{a}} \lstate',\gstate',1}{\sem{b}(\lstate,\gstate) & (\lstate',\gstate') = \sem{a}(\lstate,\gstate)}
$$
The first rule is for invocation, assuming that the environment for this thread specifies that $m$ should be invoked with arguments $\vec{v}$. These arguments are recorded in the local state, and an invocation label is generated; $\lstate^0$ is a fixed initial local state.
The second rule applies when execution reaches a return statement and a label is generated with the values provided in the local state variables $\rvs$.
Note that invocation/return labels are not invocation/return actions because they contain \emph{values} rather than arguments.
The subsequent 5 rules are built atop the standard non-deterministic semantics of KAT expressions, without any labels being generated.
The last 3 rules are for atomic actions $a$, atomic tests $b$, and ARWs $\condwr{b}{a}$, with the respective labels generated. 
When a test $b$ does not hold, the successor is undefined  (and similar when atomic test $b$ in $\condwr{b}{a}$ does not hold). 
We further define $\lstate,\gstate,k \Downarrow_{\lbl} \lstate^n,\gstate^n,k^n$, relating triples from a sequence
$(\lstate,\gstate,k)\downarrow
(\lstate^1,\gstate^1,k^1)\downarrow\cdots
\downarrow_{\lbl}
(\lstate^n,\gstate^n,k^n)$ where only the final $\downarrow$ transition produces a label. (\ie~the intermediate label-free nondeterminism has been resolved.)

The rules above give a semantics to steps of a thread assuming a certain shared state $\sigma_g$, and can be extended to sequences of steps assuming that the shared state can be changed arbitrarily in between every two steps. Formally, given a KAT expression $\kat$, an \dt{execution of} $\kat$ starting from a local state $\lstate$ and global state $\gstate$ is defined as a sequence of triples $\lstate^i,\gstate^i,\kat^i$ with $0\leq i\leq n$ such that: (1) $\lstate^0=\lstate$, $\gstate^0=\gstate$, $\kat^0=k$, (2) $\lstate^i,\gstate^i,\kat^i \Downarrow_{\lbl} \lstate^{i+1},\gstate^{i+1},\kat^{i+1}$ for all $i$ even, (3) $\lstate^i= \lstate^{i+1}$ and $\kat^i = \kat^{i+1}$ for all $i$ odd, and (4) $k^n=1$. Note that $\gstate$ is unconstrained in (3).

\begin{example}
Consider $k_{inc}$ as defined in Example~\ref{ex:ctrobj} and $\lstate^0= [\ttc \mapsto \textsf{undef}]$ and $\gstate^0=[\code{ctr} \mapsto 0]$. Here is one execution of $k_{inc}$:
\[\begin{array}{rll}
\textrm{Step 0}: & (\lstate^0,\gstate^0,(\code{c:=ctr}\cdots)^*),\\
\textrm{Step 1}: & (\lstate^0,\gstate^0,(\code{c:=ctr}\cdots)\cdot (\code{c:=ctr}\cdots)^*), & \textrm{Unfold * via }\downarrow\\
\textrm{Step 2}: & (\lstate^0,\gstate',(\code{c:=ctr}\cdots)\cdot (\code{c:=ctr}\cdots)^*), & \textrm{Global state changed arbitrarily}\\
\textrm{Step 3}: & (\lstate^0[\ttc\mapsto \gstate'(\ttctr)],\gstate', 1\cdot (\code{c:=ctr}\cdots)^*), \ldots & \textrm{Reduce an action via } \downarrow_{\code{c:=ctr}} \\
\end{array}\]
\end{example}

%


\subsection{Executions, Traces}

The set of executions of a concurrent object $O$ in the context of an environment $\env$ are defined as interleavings of single-thread executions, acting on the shared state and their local states, with nondeterministic scheduling.

A configuration $C \in (\gstate,T)$ where $T : \tids \rightharpoonup (\Lstate\times (\Kat\cup\{\bot\}))$ comprises a shared state $\gstate\in\Gstate$ and a mapping for each active thread to its local state and current code.
The initial configuration is defined by
$C_0=(\gstate^0,\emptyset)$
where $\gstate^0$ is a fixed initial shared state.
Let $\Conf$ denote the set of configurations.

An \dt{execution} of $O$ is a sequence of configurations and labeled transitions over the threads specified by an environment $\env$. The transition relation $\Rightarrow: \Conf \times (\tids\times\Lbls) \times \Conf$ is defined as: 
$$
\infer{
(\gstate,T) \exectrans{(\tid:\Inv{\tid}{m(\vec{v})})} (\gstate',T[\tid \mapsto (\lstate^0[\mathit{args}_i\mapsto v_i],k_m)])
}{
\env(\tid)=(m(\As)/\rvs:k_m,\vec{v})
\quad
T(t)\mbox{ undefined}
}
$$
$$
\infer{
(\gstate,T) \exectrans{(\tid:\lbl)} (\gstate',T[\tid \mapsto (\lstate',k')])
}{
T(\tid) = (\lstate,k)
&
\lstate,\gstate,k \Downarrow_{\lbl} \lstate',\gstate',k'
}
$$
A transition $\exectrans{\;}$ is possible for any thread whose $\kat$ is not $\bot$. The first rule models a new thread invoking a method according to the environment. In the second, the thread $\tid$ takes a $\Downarrow_{\lbl}$ step, producing label $\ell$, and the configuration is updated with the new global state and the new $(\lstate',k')$ for thread $\tid$.

\section{Proof of Theorem~\ref{def:completeness_methodology}}
\label{apx:proof_methodology_sound}

We reason by induction on the number of paths in $\Pths(\expr)$ in a completed execution $\exec$ of $O$ that are either (1) write paths but they are interleaved with actions of other threads, or (2) local paths but they are interleaved with more than two write paths in $\Pths(\expr)$, or with a single write path in $\Pths(\expr)$ but together with this path, it does not form a support of a layer in $\expr$.

The base case of the induction is trivial: since every label of a transition in $\exec$ belongs to some path in $\Pths(\expr)$, and all paths are interleaved as prescribed by the layers, then clearly, the trace of $\exec$ is in the interpretation of $\expr$.

For the induction step, consider first a write path that is interleaved with actions of other threads. Let $\exec'$ be the minimal subsequence of $\exec$ that contains only steps of that path and all the other write paths that interleave with it.
By the induction hypothesis, the latter write paths execute without interruption. This execution is feasible starting from the first configuration of $\exec'$, because we removed only local actions that do not affect enabled-ness of other concurrently executing steps. Applying the WPC condition, there exists an execution $\exec''$ strongly equivalent to $\exec'$ where all paths execute without interruption. Since $\exec''$ passes through the same sequence of shared states (modulo stuttering) it can ``replace'' $\exec'$ in $\exec$. The trace of the obtained execution is a sequence of layers which ends the proof.

Second, consider a local path that is interleaved with more than two write paths. Similarly to the previous case, one can extract only the steps of that path and all the other write paths with which it interleaves, apply the LPC condition to produce an equivalent sub-execution where that path interleaves with at most one other write path, and then, re-insert the obtained sub-execution into the original execution.


\section{SLS queue source code}
\label{apx:resqueue}

Below is the implementation of the Scherer {\it et al.}~\cite{Scherer2006} queue. Path labels such as $\circled{1t}$ or $\ropath{1f}$ are included to indicate which paths from Sec.~\ref{subsec:resqueue} correspond to those program locations, where possible. (We have slightly refactored the second portion of the implementation in our path graph.) 
\begin{center}
\lstset{xleftmargin=1.0ex,numbersep=2pt}
\begin{lstlisting}
public void enq(T e) {
  Node offer = new Node(e, NodeType.ITEM);
  while (true) {
    Node t = tail.get(), h = head.get();
    if (h == t || t.type == NodeType.ITEM) {
      Node n = t.next.get();
      if (t == tail.get()) {
        if (n != null) {
          tail.compareAndSet(t, n); `\circled{1t},\ropath{1f}`
        } else if (t.next.compareAndSet(n, offer) ) { `\label{ln:insertOffer}`
            `\circled{3t}`
            tail.compareAndSet(t, offer); `\circled{3't},\ropath{3'f}`
            while (offer.item.get() == e); `\ropath{3''}` `\label{ln:taken}`
            h = head.get();
            if (offer == h.next.get()) {
               head.compareAndSet(h, offer); return; `\circled{3''at},\ropath{3''af}`
            } else { return; `\circled{3''b}` }
        } else { restart `\ropath{3f}` }
      } else { restart `\ropath{2}` }
   } else {
      Node n = h.next.get(); `\label{ln:readHeadNext}`
      if (t != tail.get() || h != head.get() || n == null `\ropath{4}`) { `\label{ln:verifyHeadNext}`
         continue;  
      }
      boolean success = n.item.compareAndSet(null, e); `\circled{5t},\ropath{5f}` `\label{ln:appendHead}`
      head.compareAndSet(h, n); `\circled{\{6,7\}bt},\ropath{\{6,7\}bf}` `\label{ln:removeDummy}`
      if (success) 
        return; `\ropath{7a},\circled{7bt},\ropath{7bf}`
      else
        restart `\ropath{6a},\circled{6bt},\ropath{7bf}`
      }
   }
}
\end{lstlisting}
\end{center}

\subsection{SLS queue implementation graph}
\label{apx:resqueue-impl-explain}

\begin{figure}[t]
\includegraphics[width=\columnwidth]{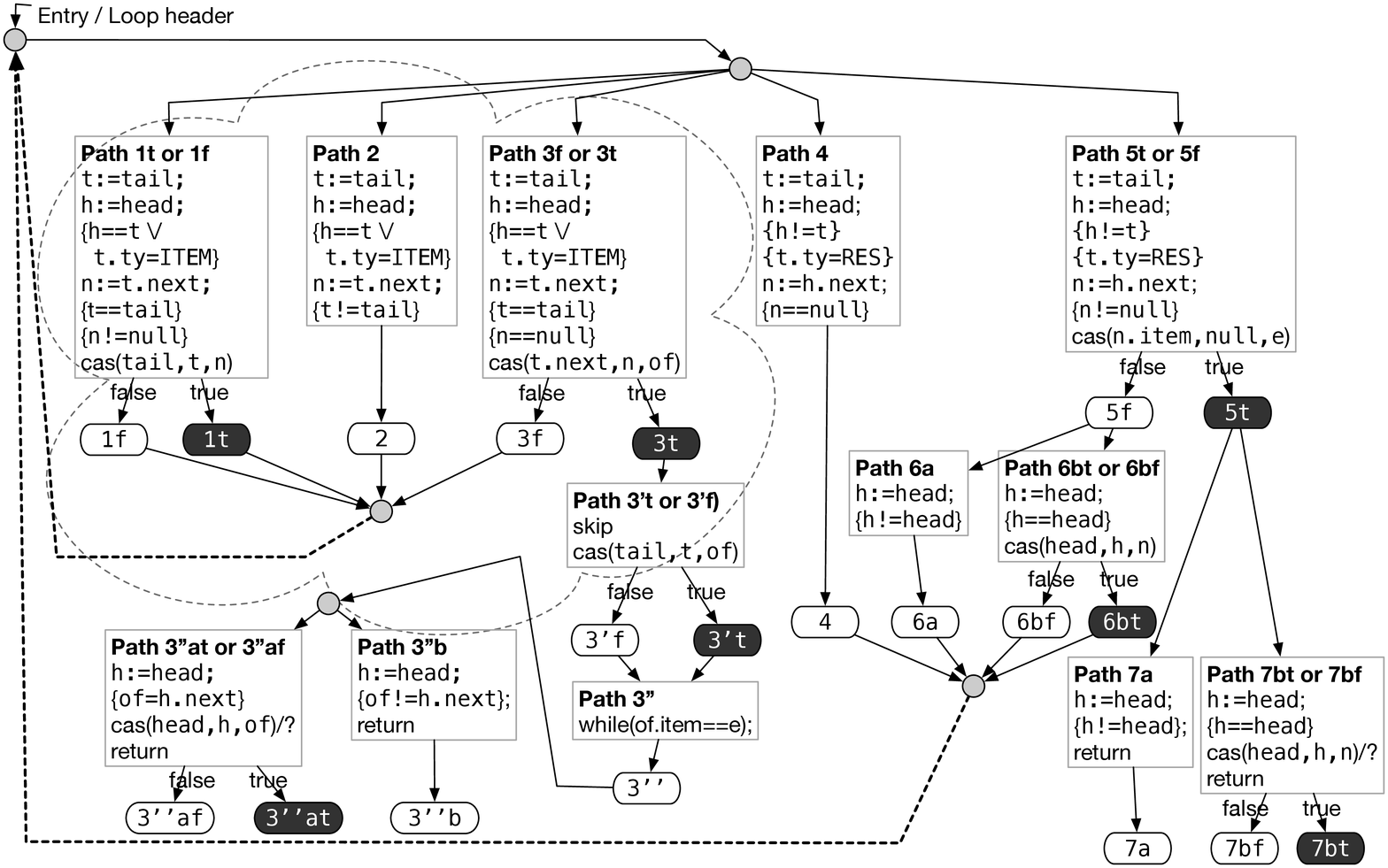}
\caption{(Reproduction of Fig.~\ref{fig:resqueue-scfg}: The implementation of
a synchronous queue due to~\citet{Scherer2006}.)}
\end{figure}

We describe the implementation in Fig.~\ref{fig:resqueue-scfg}, beginning with the cloud-surrounded area in upper left-hand half of the diagram which is, essentially, the Michael-Scott queue. In this region the queue is a list of \emph{items} (with a dummy head node), whereas the new portions of the implementation apply when the queue is a list of reservations. 
Paths $\circled{1t}$ and $\ropath{1f}$ attempt to advance the tail pointer. Path $\circled{2}$ is interrupted by a recently changed tail pointer.
Paths $\circled{3t}$ and $\ropath{3f}$ attempt to swap tail's next to their new item offer node. If successful, paths $\circled{3't}$ and $\ropath{3'f}$ attempt to advance the tail pointer. 

Path $\ropath{3''}$ is the synchronous part of the algorithm: waiting for an enqueued item to be consumed by a dequeuer. At that point, the head pointer may be stale, and paths $\circled{3''at}$, $\ropath{3''af}$ and $\ropath{3''b}$ try to advance the head pointer.

Alternatively the queue may be a list of reservations, again with a dummy head node. Paths $\circled{5t}$ and $\ropath{5f}$ attempt to fulfill a dequeuer's reservation by swapping null for an element. Path $\ropath{5f}$ is doomed to restart, while path $\circled{5t}$ will soon return. In either case, the enqueue first attempts to advance the head pointer (paths 
$\ropath{6a}, \ropath{6bf}, \circled{6bt},
\ropath{7a}, \ropath{7bf}, \circled{7bt}$).

The implementation of \lstinline|dequeue| is a sort of dual operation. 
When the queue is a non-empty list of \emph{items}, dequeue tries to take the first item by swapping the head's next value for null (and then tries to advance the head pointer). When the queue is empty or a list of reservations, dequeue redirects the tail's next to its new reservation node (and then tries to advance the tail pointer). After appending the reservation, dequeue spins until a value is swapped in, and then tries to advance the head pointer before returning.
%

\section{Detailed explanation for Counter}
\label{apx:detail-counter-case}

\begin{figure}[t]
\includegraphics[width=0.90\textwidth]{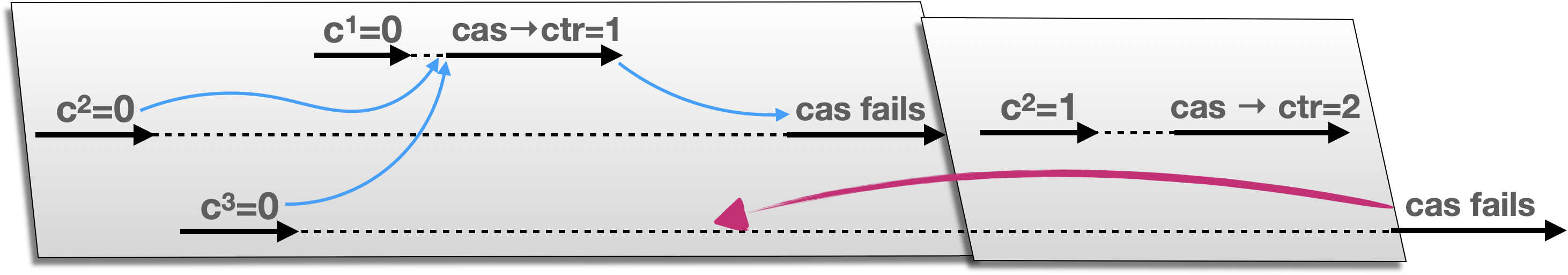}
\vspace{-3mm}
\caption{An increment-only execution for which there is an equivalent representative execution
(as suggested by the large wavy arrow) that is in the layer quotient.} 
\label{fig:counter-case1}
\vspace{-4mm}
\end{figure}

To explain the equivalence between arbitrary interleavings of increment invocations and 
representative executions in quotient $\equo{O}$,
we consider the execution pictured in Figure~\ref{fig:counter-case1}. This execution is \emph{not} 
in $\equo{O}$
because the unsuccessful iteration of thread 3 is interleaved with \emph{two} successful CASs: it reads \code{ctr} before the first successful CAS (in thread 1) and after the second successful CAS (in thread 2). Yet, as explained above, a layer interleaves an unsuccessful iteration with a \emph{single} successful CAS.

However, the second read of \code{ctr}, corresponding to the unsuccessful CAS in thread 3, is enabled even if executed earlier just after the first successful CAS. Moreover, since retry-loop iterations are ``forgetful'', \ie~there is no flow of data from one iteration to the next (the value of \code{ctr} is read anew in the next iteration), executing the unsuccessful CAS earlier would not affect the future behavior of this thread (and any other thread because it is a read) even if this reordering makes this unsuccessful CAS  read a different value of \code{ctr} (value 1 instead of 2). This reasoning extends even when the iteration of thread 3 is interleaved with more than two other iterations. 

The case of increment-only executions is simpler because it does not include the so-called ABA scenarios in which \code{ctr} is changed to a new value and later restored to a previous value. 
Every successful CAS will write a new value to \code{ctr} and will make all the other invocations that read \code{ctr} just before to restart. 

\begin{figure}[t]
\includegraphics[width=0.90\textwidth]{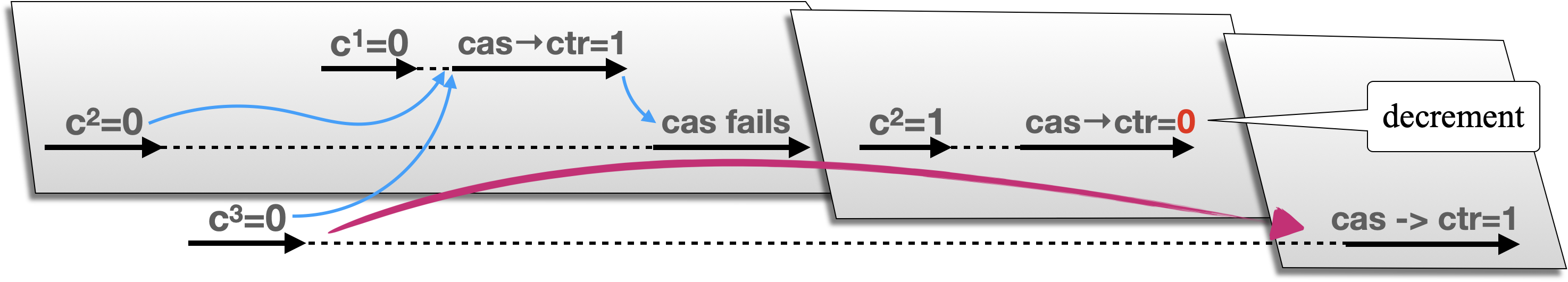}
\vspace{-3mm}
\caption{An execution where the second thread executes a decrement, which 
is equivalent to the representative execution suggested by the wavy arrow.}
\label{fig:counter-case2}
\vspace{-5mm}
\end{figure}

Interleavings of increment and decrement invocations can exhibit the ABA scenario described above, as exemplified in Figure~\ref{fig:counter-case2}. The value 0 read by thread 3 before the first successful CAS is restored by the second successful CAS (performed in a decrement invocation). This execution is \emph{not} 
a representative execution in $\equo{O}$
because the successful retry-loop iteration in thread 3 interleaves with other two successful iterations while 
in $\equo{O}$ executions,
every successful iteration is executed in isolation w.r.t. other successful iterations.
However, the equality test in the successful CAS of thread 3 (\code{ctr} == 0)  is enough to conclude that the previous read can be commuted to the right and just before the CAS. This allows to group together the actions of the third layer and obtain a 
representative execution from $\equo{O}$, extended to include decrements.
To this end, we introduce decrement layers of the form
$[(\code{c:=ctr}_{inc}^n \cdot \code{c:=ctr}_{dec}^m) \cdot \code{c:=ctr;cas(ctr,c,c-1)/true} \cdot (\code{cas/false}_{dec}^m \cdot \code{cas/false}_{inc}^n)$.
In this expression, $n$ concurrent increment threads and $m$ concurrent decrement threads interleave with a single successful decrement thread (we also subscripted with $inc/dec$ to indicate where the action came from.). All unsuccessful threads' operations commute with each other and are put in a canonical form (later the interpretation of $a^n$ will order $a$'s according to thread ids). We similarly augment the increment layers with concurrent decrement threads.


Decrement invocations can also be formed exclusively of \emph{read-only} iterations when they observe that \code{ctr} is 0. The last iteration in such invocations returns 0 and performs no write to the shared memory. Such loop iterations that read the value of \code{ctr} at the same time (after the same number of successful CASs) are grouped in a layer as well. They can be assumed to execute in isolation because they execute a single memory access.

\section{Layer automaton for Counter}

Overall, a quotient of the counter contains sequences of layers as described above. The order in which layers can occur in an execution can be constrained using regular expressions or equivalently, automata representations as shown in Fig.~\ref{fig:layers_automaton}. In this \emph{layer automaton}, states are properties of the shared memory that identify preconditions enabling shared-memory writes (successful CASs), and transitions represent layers.  

This automaton consists of two states depicted in dark gray, distinguishing shared-memory configurations where the precondition of a successful CAS in decrement invocations (\code{ctr} $>$ 0) holds. The self-loop on the initial state represents a layer (Layer 1) formed of an arbitrary number of decrement iterations returning value 0, executed by possibly different threads. ``\textsf{dec:10-11-ret(0)}'' refers to the control-flow path of decrement from Line 10 to Line 11 to return.  This is just an abbreviation; formally it is represented with KAT expressions. Layer 2 occurs on the outgoing transition from the initial state and this layer is formed from a successful increment iteration interleaved with an arbitrary number of unsuccessful increment iterations executed by different threads (when \code{ctr} equals 0 all decrement retry-loop iterations reach the \code{return} statement). Iterations are represented as control-flow paths in the code of the methods. 
{\sffamily inc:3-5-cas()/true} summarizes the single successful write path in the layer: an increment control-flow path that begins on Line 3, proceeds to the CAS, succeeds the CAS and returns. The final expression in Layer 2 summarizes an arbitrary number of threads failing the test on Line 4 (due to the successful write path), and loop back to Line 3.
The outgoing transitions from the second state represent layers containing a successful increment (Layer 3) or decrement iteration (Layer 4), each interleaved with an arbitrary number of unsuccessful increment or decrement iterations. Finally, the transition from \code{ctr}$>$0 to \code{ctr}=0  involves the same Layer 4, despite landing in a new automaton state.

\begin{figure}[t]\centering
\includegraphics[width=0.9\textwidth]{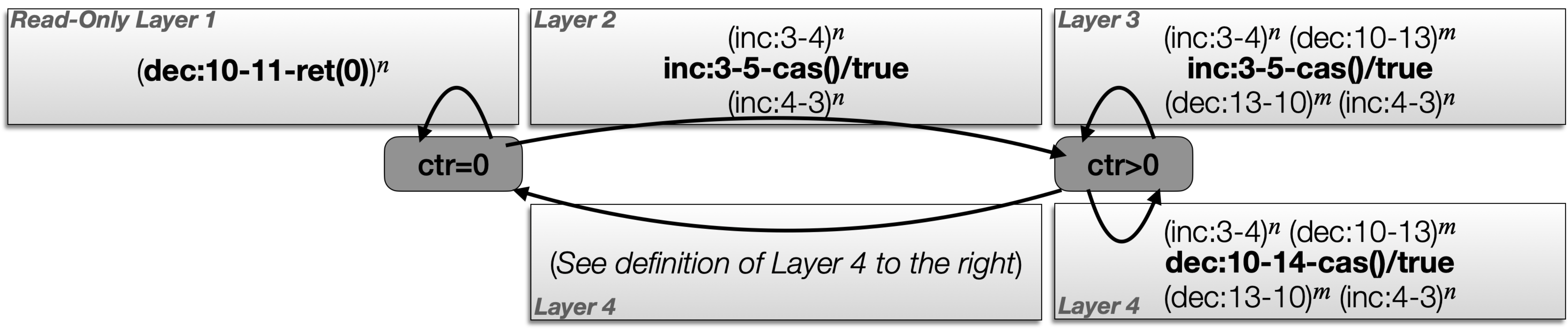}
\caption{An automaton representation of layer-serialized executions of the counter.}
\label{fig:layers_automaton}
\end{figure}


\section{Quotient for the SLS Queue}
\label{apx:resqueue-quo}

\begin{figure}[t]
\includegraphics[width=\columnwidth]{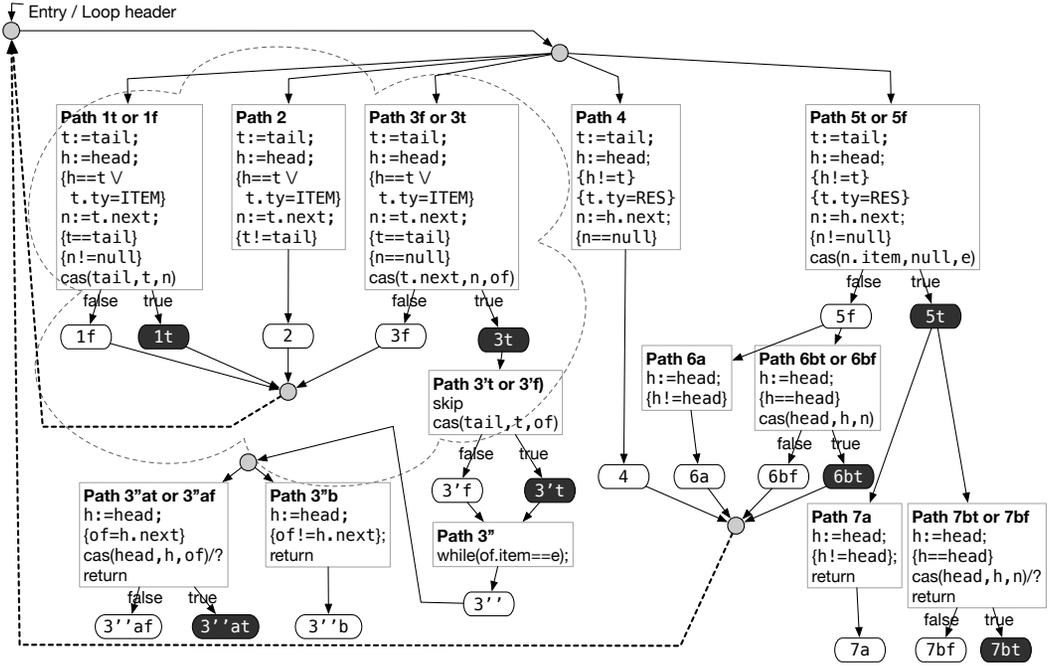}
\caption{\label{fig:resqueue-scfg} The implementation of a synchronous queue due to~\citet{Scherer2006}.}
\end{figure}

{\bf Implementation.}
The implementation of the SLS queue 
is illustrated in Fig.~\ref{fig:resqueue-scfg} (the source code is given in Apx.~\ref{apx:resqueue}).
This diagram is like a control-flow graph (entry point, loop header, branch/merge points, etc.),
but with some flattening to make paths more explicit. Paths are identified where they end, with write paths denoted as $\circled{1t}$ and local paths as $\ropath{4}$. Two paths that share a prefix and differ only based on a CAS result are denoted with a single box, but with true/false exit arcs, \eg~$\circled{1t}$ and $\ropath{1f}$. 
Later below we will write $\circled{D1t}$ versus $\circled{E1t}$ when we are referring specifically to dequeue versus enqueue.
This is the source for enqueue (which appends item nodes or fulfills reservation nodes) and the
source for dequeue is identical (except dequeue appends reservation nodes or consumes item nodes).

SLS, like MSQ, involves manipulating a list of nodes that are \emph{items} with a dummy head node. 
There is a synchronous blocking on path  $\ropath{3''}$. However, for SLS, alternatively the queue may be a list of reservations, and the right-hand paths attempt to fulfill a dequeuer's reservation by swapping null for an element. 
The implementation of \code{dequeue} is a sort of dual, omitted for lack of space. Below we denote paths such as $\circled{D5t}$ to mean the dequeue dual of enqueue's $\circled{5t}$.
Note that, unlike Treiber's stack or the MSQ, in the SLS queue a method invocation could involve a series of paths, \eg~the sequence $\circled{3t};\circled{3't};\ropath{3''};\circled{3''at}$, that involves multiple write operations.

The cloud-surrounded area in the upper left-hand half of the diagram is essentially MSQ and it involves manipulating a list of nodes that are \emph{items} with a dummy head node. 
There is a synchronous blocking on path  $\ropath{3''}$. Alternatively the queue may be a list of reservations, and the right-hand paths attempt to fulfill a dequeuer's reservation by swapping null for an element. 
\removed{The implementation of \code{dequeue} is a sort of dual, omitted for lack of space. Below we denote paths such as $\circled{D5t}$ to mean the dequeue dual of enqueue's $\circled{5t}$;}

The SLS queue demonstrates that a method implementation could consist of sequentially composed paths which define different layers. As we will see, advancing the tail pointer (and the head pointer) are subpaths of method implementation. Moreover, the synchronous behavior involves busy-wait/blocking during the implementation, after which point further paths are executed.

\begin{figure}[t]
\includegraphics[width=\columnwidth]{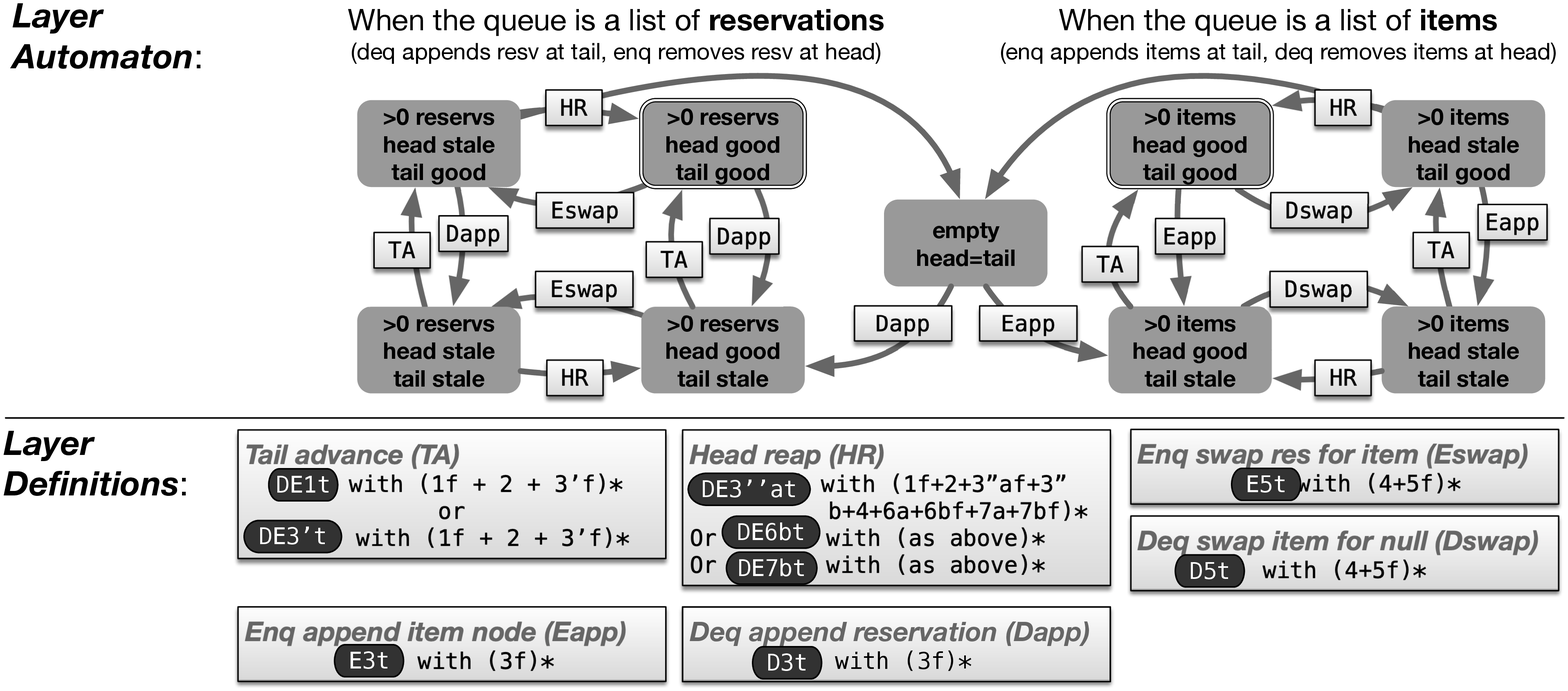}
\caption{\label{fig:resautomaton-cfgnodes} Layer automaton for the synchronous
SLS queue. Layers' acronyms and their definitions are given in the lower half of the figure. \ARXIV{States are depicted in dark boxes.} For conciseness, layer definitions do not split the prefix/suffix of the read paths.
}
\end{figure}

{\bf Quotient.} The quotient for SLS is discussed in Sec.~\ref{subsec:resqueue}. In this appendix, we show Fig.~\ref{fig:resautomaton-cfgnodes} which is similar to the Sec.~\ref{subsec:resqueue} quotient, but with precise CFG locations in the layer definitions.

Technically the states require one further predicate to indicate whether there is currently a thread at location $\circled{3t}$, omitted for lack of space. This is needed because the subsequent paths use the local variable \code{t} which is an old read done in the previous path. This is acceptable because it is only possible that one thread can be at location $\circled{3t}$ so the old read is still valid during the subsequent path. Typically we have found that implementations do not perform such ``old reads'' which are only correct as a result of very delicate reasoning.

 \begin{theorem}\label{thm:apxresqueue-linearizability} The SLS queue is linearizable.
  \begin{proof}
  The following expression uses the same layers, some marked $E$ or $D$ for linearization points:
\[\begin{array}{ll}
(
  \;\;\;\;\;[\Dapp\cdot\TA\cdot (\Eswap_{E;D}\cdot\HR \;\;\cdot\;\; \Dapp\cdot\TA)^*\cdot \Eswap_{E;D}\cdot\HR] & \textrm{// LHS}\\
  \;\;+ \;\;[\Eapp_{E}\cdot\TA \cdot(\Dswap_{D}\cdot\HR \;\;\cdot\;\; \Eapp_{E}\cdot\TA)^* \cdot \Dswap_{D}\cdot\HR] & \textrm{// RHS}\\
)^* \;\;\;\cdot\;\;\;
( 
  \;[\Dapp\cdot\TA]^* 
  \;\;\;+
  \;\;\;[\Eapp_{E}\cdot\TA]^* 
)
\end{array}\]
This expression captures 
iterating through the left and righthand sides of the automaton (passing through the empty ADT state in between),
followed by either unmatched appended dequeue reservations or unmatched appended enqueue items.
When the queue consists of reservations, the \Eswap{} layer provides the linearization point for enqueue, but also the corresponding dequeue.
\TA{} and \HR{} layers are positioned next to a corresponding  app and swap (resp.). 

We thus prove (\#1) This expression is an abstraction of the quotient: by induction on any execution,  feasible actions can be reordered into layers and those layers can be ordered into the  above expression. (\#2) For linearizability, we project out the LP operations to obtain simply $(E\!\cdot\! D)^* \cdot (E^*\! +\! D^*)$. Thus, combining with \#1, all executions meet the sequential spec.~of a queue.
\end{proof}
\end{theorem}

\section{Quotient for Treiber's Stack}
\label{apx:treiber}

Recall the implementation of Treiber's stack~\cite{IBMTR:Treiber86}, stored as a linked list from a global pointer \code{top}, and manipulated as follows:\\
\begin{tabular}{l|l}
\begin{minipage}{1.95in}
\lstset{xleftmargin=3.0ex,numbersep=2pt}
\begin{lstlisting}[morekeywords={ret}]
void push(int item){ while(1){
  node_t* n = malloc(...);
  n->val = item;
  node_t* oldTop = top;
  n->next = oldTop;
  if(CAS(top,oldTop,n) ret;
} }
\end{lstlisting}
\end{minipage}
&
\begin{minipage}{2.3in}
\lstset{xleftmargin=3.0ex,numbersep=2pt}
\begin{lstlisting}[morekeywords={ret}]
int pop(){ while(1){
  node_t* oldTop = top;
  if(oldTop==NULL) { ret 0; }
  newTop = oldTop->next;
  if(CAS(top,oldTop,newTop) ret oldTop->val;
} }
\end{lstlisting}
\end{minipage}\\
\end{tabular}

The states for the layer-automaton of the Treiber's Stack (derived from the pre-conditions of successful \code{push} and \code{pop} operations) are simply 
\code{top=null}
and
\code{top$\neq$null}.
The Treiber stack can thus be decomposed into a layer automaton as follows:
\begin{center}
\includegraphics[width=\columnwidth]{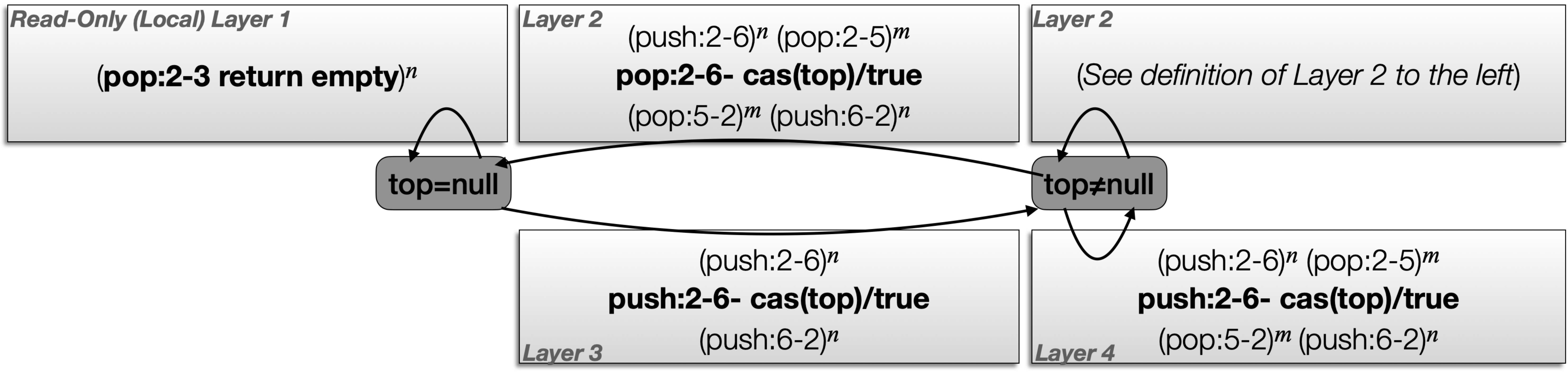}
\end{center}

Above the automaton states are given in rounded dark boxes, and edges are labeled with layers.
\ARXIV{In this diagram (and those below), we}
\CONF{We} abbreviate local paths using source code line numbers rather than KAT expressions. For example \code{pop:2-5} means the path starting at the beginning of Line 2 of \ttpop{} and proceeding to the beginning of Line 5.
Layer 1 is a local layer, in which the state is \code{top=null}. In this layer, there is only one local path from \code{pop} that is enabled for some $n$ threads and it pertains to returning 0 to indicate empty. Layer 2 occurs from a state where \code{top$\neq$null} and the \ttpop{} ARW action for the compare-and-swap occurs, causing $n$ other \ttpush{}es' and $m$ \ttpop{}s' CAS attempts to fail (on their lines 6 and 5, respectively) and thus they restart (transition back to their respective Line 2s). The write path is in bold. The other layers are similar, with a single  \ttpop{} or \ttpush{} ARW invalidating other \ttpop{}/\ttpush{} attempts. Layer 2 occurs as a label in two different transitions. Layer 5 self-loops at state \code{top$\neq$null}, which abstracts over all non-empty stacks.


\begin{lemma}\label{lemma:treiber}
The above layer automaton is an abstraction of a quotient
for Treiber's stack.
\begin{proof}
By the methodology of Def.~\ref{def:completeness_methodology}.
%
Per WPC, we must show that an old read of \code{top} and \code{top->next}, with then arbitrarily many write paths interleaved, can be moved to the right just before the successful CAS (an unsuccessful CAS belongs to a local path, discussed next). The successful CAS checks that \code{top} is unchanged since the old read. Moreover, since \code{top->next} is only written once, if \code{top} is unchanged then \code{top->next} must also be unchanged\footnote{We assume a semantics modeling garbage collection where memory cannot be reallocated. Without this assumption, it is possible that \code{top} is unchanged, but \code{top->next}  has changed. This is known as an ABA bug.}.
Therefore both old reads could be moved to the right just before the successful CAS, and a whole write path can be assumed to execute without interruption. 

Per LPC, requiring that each local path, \code{pop}s returning 0 or iterations with failed CASs, can be re-ordered to interleave with at most one write path. Iterations where a \code{pop} returns 0 perform a single access to shared-memory (reading \code{top}) and therefore, they can be assumed to execute without interruption. The failed CAS in an iteration can be re-ordered to occur just after the first successful CAS that follows the read of \code{top} in the same iteration. This holds because 
in Treiber's stack successful CAS operations always mutate \code{top} to a fresh value (assuming memory freshness). 
\end{proof}
\end{lemma}


%

\section{Quotient for Elimination Stack}
\label{apx:elimination}

The Elimination Stack~\cite{DBLP:conf/spaa/HendlerSY04} augments Treiber's stack with a protocol for ``colliding'' push and pop invocations so that the push passes its input directly to the pop without affecting the underlying data structure. An invocation starts this protocol after performing a loop iteration in Treiber's stack and failing (due to contention on \texttt{top}). The protocol uses two arrays: (1) a \texttt{location} array indexed by thread ids where a push or pop invocation publishes a descriptor object with fields \texttt{op} for the type of invocation (push or pop), \texttt{id} for the id of the invoking thread, and \texttt{input} for the input of a push operation, and (2) a \texttt{collision} array indexed by arbitrary integers which stores ids of threads announcing their availability to collide. 

Each invocation starts by publishing their descriptor in the \texttt{location} array (line~\ref{line:descriptor-ap}). Then, it reads a random cell of the \texttt{collision} array while also trying to publish their id at the same index using a CAS (lines~\ref{line:start_publish-ap}--\ref{line:end_publish-ap}). If it reads a non-NULL thread id, then it tries to collide with that thread. A successful collision requires 2 successful CASs on the \texttt{location} cells of the two threads (we require CASs because other threads may compete to collide with one of these two threads): the initiator of the collision needs to clear its cell (line~\ref{line:first_CAS-ap}) and modify the cell of the other thread (line~\ref{line:second_CAS-ap}) to pass its input if the other thread is a pop. The first CAS failing means that a third thread successfully collided with the initiator and the initiator can simply return (lines~\ref{line:start_finish-ap}--\ref{line:end_finish-ap}). Failing the second CAS leads to a restart (line~\ref{line:restart-ap}). If the invocation reads a NULL thread id from \texttt{collision}, then it tries to clear its cell before restarting (line~\ref{line:clear_CAS-ap}). If it fails, then as in the previous case, a collision happened with a third thread and the current thread can simply return (line~\ref{line:start_last_finish-ap}--\ref{line:end_last_finish-ap}).

\begin{minipage}{1.95in}
\lstset{xleftmargin=3.0ex,numbersep=2pt}
\begin{lstlisting}[morekeywords={ret}]
void push/pop(descriptor p){ while(1) {
    one iteration of Treiber stack
    location[mytid] = p; `\label{line:descriptor-ap}`
    pos = nondet(); `\label{line:start_publish-ap}`
    do { 
      him = collision[pos]
    } while (!CAS(&collision[pos], him, mytid)) `\label{line:end_publish-ap}`
    if him != NULL { `\label{line:first_test-ap}`
       q = location[him]
       if ( q != NULL & q.id = him & p.op != q.op ) { `\label{line:second_test-ap}`
          if (CAS(&location[mytid],p,NULL)) { `\label{line:first_CAS-ap}`
             if ( CAS (&location[him], q, p/NULL) ) `\label{line:second_CAS-ap}`
                return NULL/q.input
             else 
                continue `\label{line:restart-ap}`
          }
          else {
             val = NULL/location[mytid].input; `\label{line:start_finish-ap}`
             location[mytid] = NULL;
             return val `\label{line:end_finish-ap}`
    }  }  }
    if (!CAS(&location[mytid],p,NULL)) { `\label{line:clear_CAS-ap}`
       val = NULL/location[mytid].data; `\label{line:start_last_finish-ap}`
       location[mytid] = NULL;
       return val `\label{line:end_last_finish-ap}`
} } }
\end{lstlisting}
\end{minipage}

We use the automaton below to describe a sound abstraction of the quotient. Layers of Treiber's stack (defined in Section~\ref{apx:treiber}) interleave with layers of the collision protocol (some components are not exactly layers as in Definition~\ref{def:layer_exprs}, but very similar).
\begin{center}
\includegraphics[width=.6\columnwidth]{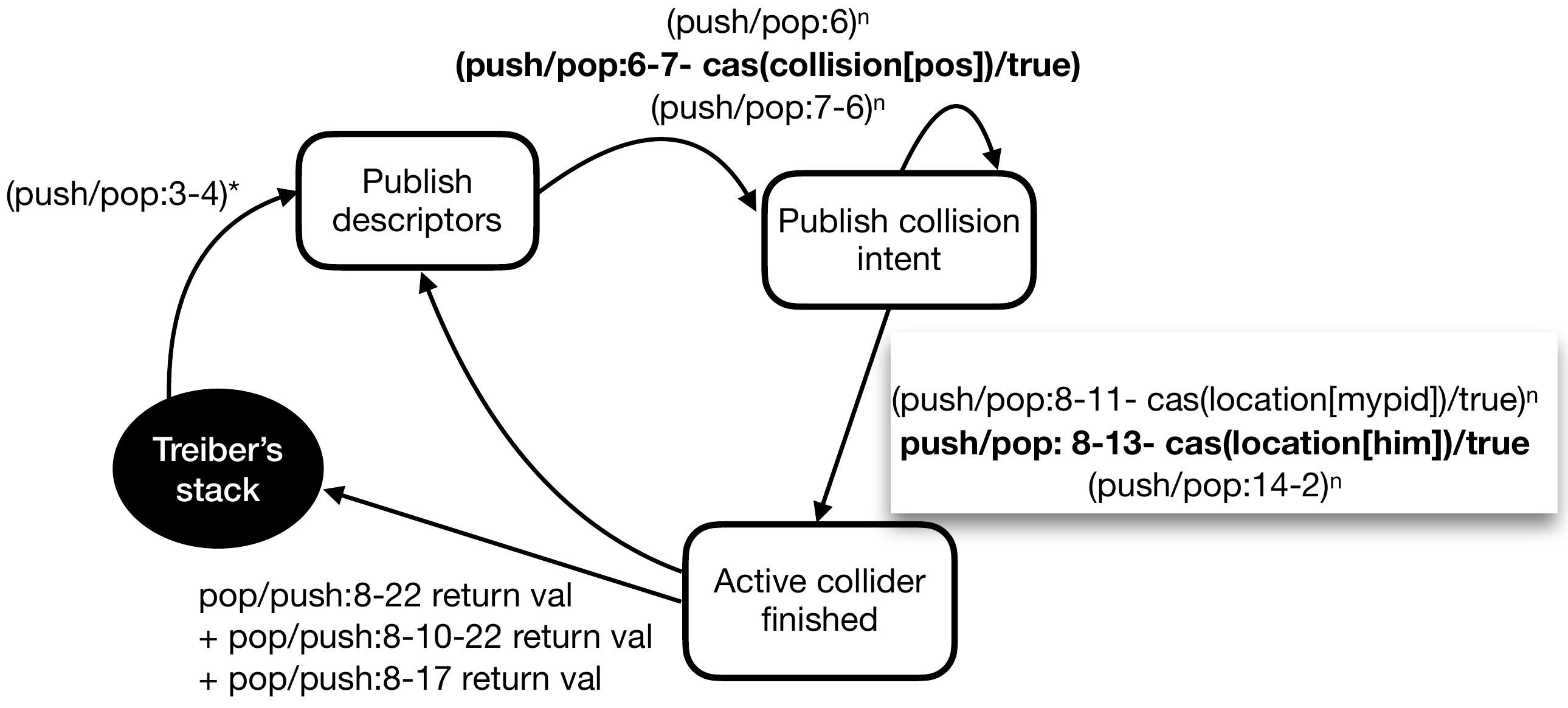}
\end{center}
Executions in the quotient serialize collisions and proceed as follows: (1) some number of threads publish their descriptor and choose a cell in the \texttt{collision} array, (2) some number of threads publish their id in the \texttt{collision} array (there may be more than one such thread -- note the self-loop on the top right state), (3) some number of threads succeed the CAS to clear their \texttt{location} cell but only one succeeds to also CAS the \texttt{location} cell of some arbitrary but fixed thread \texttt{him} and return, and (4) the thread \texttt{him} returns after possibly passing the tests at line~\ref{line:first_test-ap} or~\ref{line:second_test-ap}. We emphasize that collisions happen in a serial order, i.e., at any point there is exactly one thread that succeeds on both CASs required for a collision and immediately after the collided thread returns (publishing descriptors or collision intent interleaves arbitrarily with such serialized collisions).


\begin{theorem}\label{th:elimination}
The above automaton is an abstraction of a quotient
for the Elimination Stack.
\end{theorem}
\begin{proof} (Sketch) 
We need to show that every execution of the Elimination stack is equivalent to some execution represented by this automaton up to reordering of commutative actions. 
The interactions in the Treiber's stack component do not interfere with collisions (they use disjoint addresses in the shared memory) and therefore every execution can be assumed (up to commutativity) to execute in phases as follows: some number of invocations executing a sequence of layers as in the Treiber's stack layer automaton (competing on the \code{top} pointer) followed by some number of invocations trying to collide with each other, followed again by Treiber's stack layers, and so on. In the following we show that the collisions can be reordered to occur serially as in the above automaton. 

We proceed by induction on the number of successful CASs at line~\ref{line:second_CAS-ap} (the second CAS required for a successful collision). Consider the first such successful CAS, denoted as $CAS_2/T$ and let $\mathcal{F}_2$ be the set of threads whose next step in the execution after this point is a failed CAS on the same address. As in previous proofs, all these failed CASs can be reordered (to the left) to occur immediately after the successful one. Then, by the control-flow of an invocation, all threads in $\mathcal{F}_2$ executed the successful CAS at line~\ref{line:first_CAS-ap} before their failed CAS. All these successful CASs turn the \code{location} cell of those threads to NULL. Since no other thread (besides themselves) can turn it back to some non-NULL value (see the test at line~\ref{line:second_test-ap}), they can be reordered to occur immediately before $CAS_2/T$. This leads to an interleaving around $CAS_2/T$ that conforms to the expression that labels the transition leading to ``Active collider finished''. Then, looking at other steps before $CAS_2/T$, for every successful CAS on a \code{collision} cell, one can construct a layer as the one labeling the transitions leading to ``Publish collision intent'' and also serialize the steps \ref{line:descriptor-ap}--\ref{line:start_publish-ap} for every thread. This is possible because all these interactions concern different memory addresses. Finally, $CAS_2/T$ wrote on the \code{location} cell of a thread \code{him}, and no other thread can modify this value until \code{him} reads it, observes to have been collided and returns ($CAS_2/T$ writes either NULL to \code{location[him]} in which case the first conjunct at line~\ref{line:second_test-ap} will fail in another thread, or a descriptor with an \code{id} field different from \code{him} in which case the second conjunct at line~\ref{line:second_test-ap} will fail). Therefore, all those steps of \code{him} can be reordered to the left to occur immediately after the interaction around $CAS_2/T$, which completes the handling of this first collision. The subsequent collisions can be handled in a similar manner.
\end{proof}
\section{Quotient for RDCSS}
\label{apx:rdcss}

The Restricted Double-Compare Single-Swap (RDCSS)~\cite{DBLP:conf/wdag/HarrisFP02} is a restricted version of a double-word CAS (acting atomically on two addresses) which modifies a so-called data location provided that this location and another so-called control location have some given expected values. This is an instance of an atomic read-modify operation, i.e., the tests and the write should happen atomically. It is assumed that data and control locations are disjoint (i.e., the same address can not be a data address in some invocation and a control address in another). The code of the main RDCSS operation is given below (for simplicity, we omit the read operation):

\begin{tabular}{l|l}
\begin{minipage}{2.75in}
\lstset{xleftmargin=3.0ex,numbersep=2pt}
\begin{lstlisting}[morekeywords={ret}]
void RDCSS(descriptor *d){ 
  do {
     r = CAS(d->DATA_ADDR, d->exp_data, d); `\label{line:dataCAS1}`
     if ( isDescriptor(r) ) Complete(r); `\label{line:helpingComplete}`
  } while ( isDescriptor(r) )
  if ( r == d->exp_data ) Complete(d) `\label{line:regularComplete}`
  return r;
}
\end{lstlisting}
\end{minipage}
&
\begin{minipage}{2.3in}
\begin{lstlisting}[firstnumber=9]
void Complete(descriptor *d) {
  if ( * d->CONTROL_ADDR == d->exp_control )
     CAS(d->DATA_ADDR, d, d->new_data); `\label{line:dataCAS2T}`
  else
     CAS(d->DATA_ADDR, d, d->exp_data) `\label{line:dataCAS2F}`
}
\end{lstlisting}
\end{minipage} \\
\end{tabular}

The inputs of the operation are put inside a \code{descriptor} structure: \code{DATA\_ADDR} and \code{CONTROL\_ADDR} are the data and control addresses, respectively, \code{exp\_data} and \code{exp\_control} are the expected values of these addresses, and \code{new\_data} is the new value to be written to the data address (provided that the data and control addresses store the expected values). 

\code{RDCSS} attempts a standard CAS on the data address to change the old value into a pointer to the descriptor (line~\ref{line:dataCAS1}). This CAS checks that the data address has the expected value, and if it fails, the operation simply returns. In the context of this implementation, we assume that a CAS returns the value of the location before any modification (if any) and not just a Boolean. If the CAS succeeds, then the operation calls \code{Complete} in order to check the control location and finalize the modification if possible (line~\ref{line:regularComplete}). \code{Complete} checks the value of the control location and if it has the expected value, then it attempts a CAS to change the data address (line~\ref{line:dataCAS2T}); note that the data address currently stores a pointer to a descriptor. Otherwise, it attempts a CAS to revert the data address to its old value (line~\ref{line:dataCAS2F}). 

When multiple threads compete to change the same data address, it may happen that the thread succeeding the first CAS at line~\ref{line:dataCAS1} (the initiator) is slow and before it executes the call to \code{Complete}, another thread fails its CAS but finds a descriptor at this address (it is assumed that descriptor pointers can be distinguished from data values). Then, this other thread will try to help the slower one and call \code{Complete} itself (line~\ref{line:helpingComplete}). Note that all the information needed to help the slower thread is stored in the descriptor.

We use the expression below to describe a sound abstraction of the quotient:

\hspace{-4mm}
{\small
\begin{tabular}{ll}
\textcolor{blue}{\texttt{// }\emph{successful modification}} \\
\hspace{-2mm}\Big(\textbf{3-CAS/true} $\cdot$ (3-CAS/false-4)$^n$ $\cdot$ & \{ \textbf{4-6-11-CAS/true} $\cdot$ (4-11-CAS/false)$^n$ \\[1mm]
& + \textbf{3-CAS/false-4-11-CAS/true}$\cdot$ (4-11-CAS/false)$^n$ $\cdot$ 4-6-11-CAS/false \}\\
+ \\
\textcolor{blue}{\texttt{// }\emph{fail: wrong control value}} \\
\textbf{3-CAS/true} $\cdot$ (3-CAS/false-4)$^n$ $\cdot$ & \{ \textbf{4-6-13-CAS/true} $\cdot$ (4-13-CAS/false)$^n$ \\[1mm]
& + \textbf{3-CAS/false-4-13-CAS/true}$\cdot$ (4-13-CAS/false)$^n$ $\cdot$ 4-6-13-CAS/false \}\\
+ \\
\textcolor{blue}{\texttt{// }\emph{fail: wrong data value}} \\
(3-CAS/false)$^*$ \Big)$^*$
\\
\end{tabular}}

Executions in the quotient are iterations (note the outer $^*$) of three types of ``phases'' (note the outer union and read expressions from top to bottom): (1) a phase in which the data address is modified successfully (without or with help), (2) a phase in which the modification fails because the \emph{control} address does not have the expected value (noticed by the initiator of the modification or a helper thread), and (3) a phase in which the modification fails because the \emph{data} address does not have the expected value.

The first two phases have a common prefix: some initiator thread succeeding the CAS at line~\ref{line:dataCAS1} and some number $n$ of threads failing the same CAS and reading the descriptor written by the initiator. Next, for the first phase, there are two cases: (1) the initiator succeeds the second CAS at line~\ref{line:dataCAS2T} (after calling \code{Complete} at line~\ref{line:regularComplete}), and those $n$ threads will fail the same CAS (after calling \code{Complete} at line~\ref{line:helpingComplete}), or (2) some helping thread which fails the same CAS as the other $n$ threads will succeed the CAS at line~\ref{line:dataCAS2T} (after calling \code{Complete} at line~\ref{line:helpingComplete}), and the initiator together with the other $n$ threads fail the same CAS. For the second phase, there are two analogous cases in which either the initiator or a helping thread observes a wrong value for the control location and succeeds the CAS at line~\ref{line:dataCAS2F} to revert the value of the data location. The third phase is trivial and consists of an arbitrary number of failed instances of the CAS at line~\ref{line:dataCAS1}.

\begin{theorem}\label{th:rdcss}
The above expression is an abstraction of a quotient
for RDCSS.
\end{theorem}
\begin{proof}(Sketch)
Since steps of \code{RDCSS} invocations on different data addresses commute (the assumption that data and control addresses are disjoint is important here), we focus on invocations that act on the same data address. We follow the same strategy as for Elimination Stack, and proceed by  induction on the successful CASs in \code{Complete} (line~\ref{line:dataCAS2T} or line~\ref{line:dataCAS2F}). Consider the first such CAS and assume that it is at line~\ref{line:dataCAS2T}. This corresponds to the first phase above and the case of line~\ref{line:dataCAS2F} which corresponds to the second phase can be handled similarly. There are two cases to consider:
\begin{itemize}
	\item The thread $t$ performing this CAS called \code{Complete} at line~\ref{line:regularComplete}. If there are threads whose next step in the execution after this point is a failed CAS on the same address and expecting to find the same descriptor, then all of these steps can be reordered to the left to occur immediately after the successful one. By control-flow, these other threads arrived there by calling \code{Complete} at line~\ref{line:helpingComplete} which means that they fail their CAS at line~\ref{line:dataCAS1} and they read the same descriptor. All of these failed CASs can be reordered to occur immediately after $t$ succeeding its CAS at line~\ref{line:dataCAS1}. Overall, these reorderings lead to an execution fragment with the shape described in the first line of the expression above.
	\item The thread $t$ performing this CAS called \code{Complete} at line~\ref{line:helpingComplete}. Following a similar reasoning while taking into account that another thread $t'$ initiated this modification by succeeding a CAS at line~\ref{line:dataCAS1}, one can reorder steps to obtain a prefix with the shape given by the second line of the expression above.
\end{itemize}
While building serializations of phases of type (1) and (2) above, any failed CASs at line~\ref{line:dataCAS1} that return the same value can be reordered to occur one after another, thereby creating phases of type (3). And these phases of type (3) can occur ``outside'' of phases of type (1) and (2) since they have no effect on the shared memory.
\end{proof}
\section{Quotient for the List Set}
\label{apx:listset}

We here consider a List Set Object and  describe the layer expressions and proof that they are an abstraction of the List Set's quotient.
This example is a Set object implemented as a sorted linked list~\cite{DBLP:conf/podc/OHearnRVYY10},
which involves a read-only traversal \ttlocate, and then small atomic sections to link/unlink nodes (for \ttinsert/\ttdelete, respectively). \ttlocate\ traverses the list from the \code{head} and returns a pair of nodes (\code{x},\code{y}) such that  \code{y} has the key of interest or else \code{x} points to the last node whose key is below \code{k}. It is implemented as a loop that may perform an unbounded number of shared-memory reads.
We assume that it is abstracted with the postcondition at line 3 in \code{insert} stating that \code{y} is the successor of \code{x}, the input \code{k} is in between \code{x.key} and \code{y.key}, and that at some point between the invocation of the operation and “now”, \code{x} resides on a valid search path for \code{k} that starts at the head of the list, denoted as  $\Diamonddot\code{head}\stackrel{\code{k}}{\rightarrow}\code{x}$.
Recent work~\cite{FeldmanE0RS18,DBLP:journals/pacmpl/FeldmanKE0NRS20} shows that this postcondition can be derived easily by showing that roughly, list nodes are never updated once they become unreachable.
Therefore, the implementations of \ttinsert{} and \ttdelete{} are as follows:
\begin{center}
\begin{tabular}{l}
\begin{minipage}{2.7in} 
\lstset{xleftmargin=1.0ex,numbersep=2pt}
\begin{lstlisting}
int insert(int k) { while(1) {
    struct node_t *z = ...;
    assume $\code{x.next = y}\land \code{x.key} < \code{k}\leq \code{y.key} \land \Diamonddot \code{head}\stackrel{\code{k}}{\rightarrow}\code{x}$
    atomic {
      if (x->next == y && x->del == 0) {
        if (y->key != k) {
          z->next = y;
          x->next = z;
          return 1;
      } else { return 0; }
   }
} }
\end{lstlisting}
\end{minipage}\\
\hline
\begin{minipage}{2.3in}
\lstset{xleftmargin=3.0ex,numbersep=2pt}
\begin{lstlisting}
int delete(int k) { while(1) {
  assume $\code{x.next = y}\land \code{x.key} < \code{k}\leq \code{y.key} \land \Diamonddot \code{head}\stackrel{\code{k}}{\rightarrow}\code{x}$
  atomic {   
    if(x->next == y && x->del == 0) {
      if (y->key == k) {
        y->del = 1;
        x->next = y->next;
        return 1;
      } else { return 0; }
    }
  }
} }
\end{lstlisting}
\end{minipage}\\
\end{tabular}
\end{center}
The \ttinsert{} method will link a node \code{z} in between \code{x} and \code{y}, provided that \code{k} wasn't already in the list.
The \ttdelete\ method returns 0 if the element wasn't in the list and otherwise, marks node \tty\ for deletion, and then updates \ttx\ to skip past node \tty. 
The \ttdelete{} method marks deleted nodes with a \code{del} flag before they are unlinked. Because \ttdelete{} marks deleted nodes' \code{del} fields, a concurrent \ttlocate{} that has just found this node, but was then preempted by \ttdelete{}, will return a node that's marked as deleted and unlinked, not simply unlinked.

As we discuss below, for List Set the layer expressions based on interleavings of two threads generalizes to arbitrary threads. We thus define the states of the automaton in terms of the possible values from the perspective of one reader and one writer. In these states below, \ttx$_w$ denotes the writer's \ttx, \ttx$_r$ denotes the reader's \ttx\ and similar for the other variables. The \code{x} and \code{y} variables are existentially-quantified in the pre-conditions because they are method-local variables and not inputs. We omit the sub-formula $\Diamonddot \code{head}\stackrel{\code{k}}{\rightarrow}\code{x}$ because this condition does not affect the enabled status of a layer.
\[\footnotesize\arraycolsep=1.4pt\begin{array}{rll}
\qq^1 &=& \sem{\textrm{$\exists$\code{x$_r$},\code{y$_r$}.\ \code{x$_r$->next=y$_r$ $\wedge$ !x$_r$->del $\wedge$ k$_r$=y$_r$->key}}},\\
\qq^2 &=& \sem{\textrm{$\exists$\code{x$_r$},\code{y$_r$}.\ \code{x$_r$->next=y$_r$ $\wedge$ !x$_r$->del $\wedge$ x$_r$->key~<~k$_r$~<~y$_r$->key}}},\\
\qq^3 &=& \semL
   \textrm{$\exists$\code{x$_r$},\code{y$_r$},\code{x$_w$},\code{y$_w$}.\ \code{x$_w$->next=y$_w$ $\wedge$ !x$_w$->del $\wedge$ x$_w$->key~<~k$_w$~<~y$_w$->key 
    $\wedge$}}\\ 
&&  \textrm{\code{x$_r$=x\added{$_w$} $\wedge$ x$_r$->key~<~k$_r$~<~y$_r$->key}} \semR,\\
\qq^4 &=& \semL
   \textrm{$\exists$\code{x$_r$},\code{y$_r$},\code{x$_w$},\code{y$_w$}.\ \code{x$_w$->next=y$_w$ $\wedge$ !x$_w$->del $\wedge$ x$_w$->key~<~k$_w$~<~y$_w$->key $\wedge$ }}\\
   && \textrm{\code{x$_r$=x\added{$_w$} $\wedge$ k$_r$=y$_r$->key}} \semR,\\
\qq^5 &=& \semL
   \textrm{$\exists$\code{x$_r$},\code{y$_r$},\code{x$_w$},\code{y$_w$}.\ \code{x$_w$->next=y$_w$ $\wedge$ !x$_w$->del $\wedge$ k$_w$=y$_w$->key $\wedge$ x$_r$=x\added{$_w$} $\wedge$ }}\\
   && \textrm{\code{  x$_r$->key~<~k$_r$~<~y$_r$->key}} \semR,\\
\qq^6 &=& \sem{\textrm{$\exists$\code{x$_r$},\code{y$_r$},\code{x$_w$},\code{y$_w$}.\ \code{x$_w$->next=y$_w$ $\wedge$ !x$_w$->del $\wedge$ k$_w$=y$_w$->key $\wedge$ x$_r$=x\added{$_w$} $\wedge$ k$_r$=key}}},\\
\qq^7 &=& \sem{\textrm{$\exists$\code{x$_r$},\code{y$_r$}.\ \code{x$_r$->next=y$_r$ $\wedge$ x$_r$->del $\wedge$ x$_r$->key~<~k$_r~\leq~$y$_r$->key}}}\\
    \end{array}\]

With these 7 states, 2 write paths (one from \ttinsert, one from \ttdelete) and 6 read paths, there are many transitions to consider, although many of them are labeled with the same layer. In fact, the List Set can be decomposed into 8 layers, enumerated below.
For lack of space, we omit the automaton, but the definitions, including all 77 feasible transitions, can be seen in the output of our tool (which we discuss in the next section) shown in Apx.\apxA. Below we refer to example transitions in Apx\apxA, denoted $\delta_i$.

\newcommand\mysee[1]{(\eg~$\delta_{#1}$)}
\begin{enumerate}

\item A layer with a \ttdelete{} write path that updates \code{x.next} to point to \code{y.next}, causing one \ttinsert{} and one \ttdelete{} path to fail when finding  \code{x.next$\neq$y}.\mysee{2}

\item A layer with an \ttinsert{} write path that updates \code{x.next} to point to \code{z}, causing one \ttinsert{} and one \ttdelete{} paths to fail when finding \code{x.next$\neq$y}. \mysee{9}




\item A local layer consisting of one \ttdelete{} path, when the key is not in the set. \mysee{31}

\item A local layer consisting of one \ttinsert{} path, when the element is already in the set. \mysee{47}


\item Four local layers consisting of \ttinsert{} or \ttdelete{} paths when the node \code{x} is already marked for deletion. \mysee{63}

%
%
%
%
%
%
%
%
\end{enumerate}
Note that \code{insert} and \code{delete} have more than one control-flow path that ``fails'' because of the nested conditional inside the atomic read-write.


As in the Michael/Scott queue, here again the layer-quotient is 
optimal \emph{modulo} method-level commutativity. At the method-level, operations such as insertion/deletion of different elements commute and their corresponding linearizations can be commuted 
(different orders of write layers) in the layer quotient.

\begin{lemma}
The above layer automaton is an abstraction of a quotient for the List Set.
\end{lemma}

Proof by the methodology of Def.~\ref{def:completeness_methodology}.
To prove the lemma we first note that the post-condition of \ttlocate{} ensures that \code{x} was reachable and that \code{y=x->next}.
In all read and write paths, the ARW checks that \code{y=x->next} still holds. Furthermore, an invariant of the implementation is that if
\code{x} was reachable at some point in the past
(\ie~when \ttlocate{} executed) and
\code{!x->del} holds in the atomic section,
then \code{x} is still reachable in the atomic section (this holds because elements are marked before being unlinked). Therefore, if \ttlocate{}'s postcondition was true in the past, it remains true when the ARW succeeds and the \texttt{assume} can be reordered to occur just before it. For local paths, as in previous cases, a failed ARW can be commuted to the left to occur just after the first ARW that modifies the location \texttt{x}.

\section{Quotient for the Herlihy-Wing Queue}
\label{apx:hwq}

Recall the queue due to Herlihy and Wing~\cite{DBLP:journals/toplas/HerlihyW90}, reproduced below:
\begin{center}
\lstset{xleftmargin=3.0ex,numbersep=2pt}
\begin{tabular}{l|l}
\begin{minipage}{2.6in}\footnotesize
\begin{lstlisting}
int deq() { while(1) {
  assume 0 <= range < back; int j = 0;
  while(j<range) {
    v := swap(items[j],null); 
    if (v != null) return v;
    j++; }
} }      
\end{lstlisting}
\end{minipage}
&
\begin{minipage}{1.9in}\footnotesize
\begin{lstlisting}
void enq(int v) { 
   i := back++;
   items[i] = v;
}
\end{lstlisting}
\end{minipage}
\end{tabular}
\end{center}
Enqueue (on the right) reserves the next slot in the array \code{items} by \emph{atomically} reading and incrementing the shared variable \ttback{}, and then assigns the value to that slot in a \emph{second} write to the shared state.
Meanwhile, dequeue (on the left), in an outer loop reads into \code{range} any value strictly smaller than \ttback{} and then iterates from 0 to \code{range}, looking for a slot containing an item to atomically dequeue. For every \code{j}, it atomically reads \code{items[j]} into \code{v} and writes \code{null} (written as a \code{swap} instruction), and if the read value is not \code{null}, it returns it.
This is actually a sound abstraction of the original version which assigns \ttback{}-1 to \code{range} instead of any smaller value. Soundness follows easily from the fact that reading a smaller value will only make the dequeue restart more often (perform more traversals in which there is no occupied slot), but not affect safety. In the reasoning below, we will use the fact that such a non-deterministic read commutes to the right of any increment of \ttback{} (it is a right mover).

We show that  the Herlihy-Wing queue quotient can
be abstracted by an expression given below.
To this end we first, as below in Apx.~\ref{apx:hwqproof}, prove that
any iteration of the dequeue's outer \texttt{while(1)} loop can be considered atomic, modulo commutative re-orderings.
Consequently, there exists a quotient of Herlihy and Wing Queue where outer loop iterations in dequeue are atomic sections. Since \code{items[i] = v} steps in enqueues commute (they write on different slots of the array), there exists a quotient where additionally, every sequence of \code{items[i] = v} steps before a dequeue iteration that is successful (contains a non \code{null} swap) is ordered w.r.t. the array slots that they write. The following expression is an abstraction of such a quotient:
$
\left( \deqF^*\cdot (\enqinc)^+ \cdot\enqwr^* \cdot \deqT^*\right)^* 
$
where $\enqinc$ and $\enqwr$ denote the statements \code{i:=back++} and \code{items[i] = v} in enqueue, respectively, and $\deqT$ and $\deqF$ represent entire iterations of the outer loop in dequeue that end in a \code{return} and restarting the loop, respectively. The interpretation of $\enqwr^*$ is refined to be a set of sequences of labels \code{items[i] = v} (with thread ids) that are ordered w.r.t. the position $i$ that they write to. Above, we also use straightforward feasibility arguments like ``enqueues increment \code{back} before writing to \code{items},'' and ``a $\deqT$ must be preceded by a write to \code{items} in an enqueue.''

\begin{theorem}\label{th:hwq}
The above expression is an abstraction of the HWQ quotient.
\end{theorem}

\begin{theorem}\label{th:linearizability-hwq}
The HWQ is linearizable.
\end{theorem}

The set of traces represented by this expression admits a ``simple'' linearization point mapping which identifies $\enqwr$ and $\deqT$ steps with linearization points of enqueues and dequeues, respectively. The restriction to traces in this quotient is instrumental for such a simple linearization point mapping. For arbitrary traces, the Herlihy and Wing Queue is known for having linearization points that depend on the future and that \emph{can not} be associated to fixed statements, see e.g.~\cite{DBLP:conf/cav/SchellhornWD12} !


%
%

\subsection{Proof of atomicity of outer loop}
\label{apx:hwqproof}

We prove that any iteration of the outer \texttt{while(1)} loop can be considered to be an atomic section, modulo re-orderings of commutative actions. That is, there exists a quotient of this object formed of traces where all steps of such an iteration occur consecutively one after another. 

We proceed by induction on the number of steps executing the swap at line 4 in dequeue and that find a non-\code{null} value in \code{items[j]}. In the base case, i.e., the number of such steps is 0, all the swaps at line 4 in all dequeue invocations find \code{null} values. Therefore, any possible write to an \code{items} slot (in enqueues) can be re-ordered after all swaps. Now all steps in the same outer loop iteration of a dequeue (the non-deterministic read of \code{back} and swaps returning \code{null}) can be re-ordered to occur consecutively one after another. In particular this relies on the fact that the non-deterministic read of \code{back} can return the same value even if executed after more increments of \code{back}.
For a trace with $n+1$ swaps returning non-\code{null} values, we focus on the first such step. Assume that it is a swap on some position \code{k}. All the writes to \code{items} slots strictly before \code{k} can be re-ordered to the right of this first non-\code{null} swap. This relies on the fact that all the other previous swaps return a \code{null} value and anyway, do not ``observe'' these writes. Similarly to the base case, all steps in the same outer loop iteration of a dequeue that completes before this first non-\code{null} swap (including) can be re-ordered to occur consecutively one after another.
We are again using the fact that swaps read \code{null} values and there is no more write on the slots that they read. Now, removing the write to $\code{items[k]}$ in enqueue and the dequeue iteration that removes this value from the current trace, we get another feasible trace that has $n$ non-\code{null} swaps, for which one can apply the induction hypothesis.


\section{Evaluation: Algorithm Authors' Correctness Arguments}
\label{apx:revisit}

As discussed in Sec.~\ref{sec:intro}, our goal is to provide a formal foundation for the scenario-based correctness arguments found in the literature. In this section, we evaluate our work by revisiting various such arguments in the literature, and comparing them with the quotient-based proofs presented in this paper. At the high level, our comparison shows that quotients make scenario-based reasoning more explicit and ensure that all cases are considered.

\subsection{Treiber's Stack}

Treiber's stack is fairly straight-forward. As such, it provided a good starting point for defining quotients yet the prose correctness arguments are fairly minimal. For example, the following is a comment on linearizability:

\QuoteAuthor{The linearization point of both the push() and the pop() methods is the successful compareAndSet(), or the throwing of the exception in case of a pop() on an empty stack.}{\TAOMPP}

\noindent
This prose identifies specific linearization points as (1) the ``successful compareAndSet'' and (2) the not-found exception. These LPs correspond to the layers in the quotient shown in Apx.~\ref{apx:treiber}. Layers 2, 3, 4 are ``successful compareAndSet'' linearization points, and read-only Layer 1 is the linearization point for the not-found exception.

\paragraph{Summary.} 
The following table summarizes the various elements of the correctness argument/proof, and identifies examples of where they occur in the \citet{TAOMPP} proof, and where they occur in the quotient proof.

\noindent
\begin{tabular}{|p{1.2in}|p{1.9in}|p{1.9in}|}
\hline
{\bf Proof Element} & \citet{TAOMPP} {\bf Proof} & {\bf Quotient Proof} \\
\hline
\hline
ADT states & ``empty stack'' & ADT states, e.g. (\texttt{top=null})\\
\hline
Concurrent threads & (general description) & Superscripting $(...)^n$ \\
\hline
Thread-local step seq. & ``try to swing [top] ... if [] succeeds, push() returns, and if not, the [] attempt is repeated'' & Layer paths, e.g., \textsf{push:2-6}\\
\hline
Linearization pts. & ``The linearization point of both the push() and the pop() methods is the successful compareAndSet(), ...'' & The successful CAS in Layers 2, 3 and 4. \\
\hline
(continued) & ``...or the throwing of the exception in case of a pop() on an empty stack.'' &  Read-Only Layer 1\\
\hline
\end{tabular}

\noindent
The layer quotient and, especially, the layer automaton (shown in Apx.~\ref{apx:treiber}) helps make the \citet{TAOMPP} proof more explicit. 
The layer automaton makes the ADT states explicit. 
From each ADT state, one can consider which (i.e. all possible) layers are enabled, and which target states are reached via those layers. 
Linearization points are explicit in the layer quotient, occurring once with each layer transition. 
The layer quotient automaton also has the benefit of explicitly showing all of the linearizable executions: i.e. all the possible runs of the automaton. This is left as implicit in the \citet{TAOMPP} proof.


\subsection{Elimination Stack}
\label{apx:elimination-stack}

Section 5 of \citet{DBLP:conf/spaa/HendlerSY04} gives a correctness proof for the elimination stack. We now review the proof and compare it with the quotient given in Apx.~\ref{apx:elimination}.
For reference, the following is a replication of the quotient automaton:

\begin{center}
\includegraphics[width=\columnwidth]{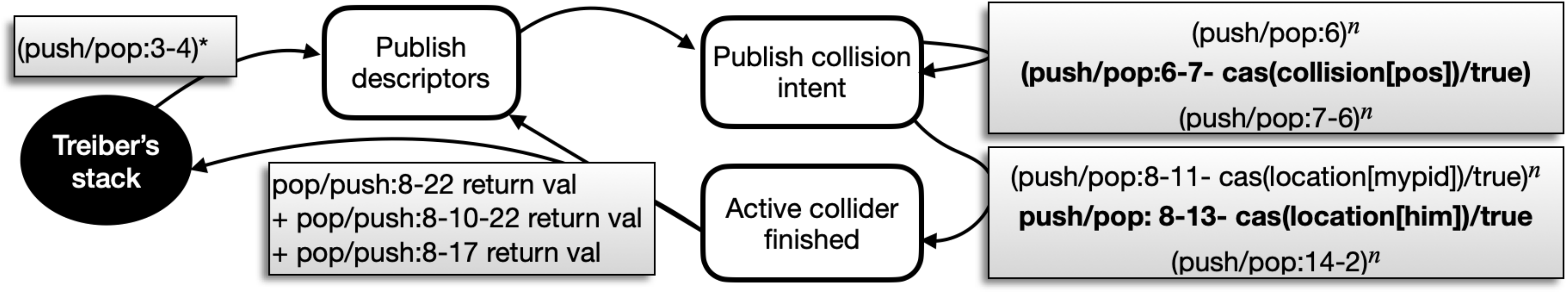}
\end{center}

\HSYprose{We note that a set is a relaxation of a stack that does not require LIFO ordering. We begin by proving that our algorithm implements a concurrent set, without considering a linearization order. We then prove that our stack implementation is linearizable to the sequential stack specification of Definition 5.1. Finally we prove that our implementation is lock-free.}

\HSYprose{We now prove that our algorithm has correct set semantics, i.e. that pop operations can only pop items that were previously pushed, and that items pushed by push operations are not duplicated. This is formalized in the following definition [omitted Set semantics for methods Push/Pop]}

\noindent
\citet{DBLP:conf/spaa/HendlerSY04} decompose the proof into first Set semantics and then ordering considerations. In the quotient this is unnecessary because the layers capture the ordering and the elements in them. 
In the bottom right layer in the {\bf bold} action, a single push or pop succeeds, colliding with another operation of the oppose type, and passing the element from the push to the pop. 
(Note that the quotient automaton could also have been written in a more verbose way where the bottom right layer is replaced with two layers: (1) a layer where a push's successful CAS takes with it a corresponding pop, and (2) a layer where a pop's successful CAS takes with it a corresponding push. For succinctness, we have combined those layers using the ``push/pop'' notation.) 
As discussed below, the thread that succeeds its CAS in the bottom right later is referred to as the ``active'' thread, and the thread with which the active thread collides is referred to as ``passive.''
These concepts are explicit in the quotient: the thread taking the bold action in the bottom right is the ``active'' thread, and the thread that finds itself collided with in the layers on the arcs that exit the ``Active Collider Finished'' state, are ``passive.''

We now continue to examine the authors' proof:

\newcommand\repeatedQuo{As discussed above, the bottom right layer in the {\bf bold} action, a single push or pop succeeds, colliding with another operation of the oppose type, and passing the element from the push to the pop.}

\HSYprose{In the following, we prove that operations that exchange their values through collisions are also correct set operations, thus we show that our algorithm has correct set semantics.
...
We say that a colliding operation op is
\underline{active} if it executes a successful CAS in lines C2 or C7. We
say that a colliding operation is \underline{passive} if op fails in the CAS
of line S10 or S19. [underlines added]}

\noindent
The authors lay out a few definitions, which are also captured by the layer quotient. Above the authors' intuitive concept of ``active'' is captured by the paths in a layer that succeed their CAS. Likewise for ``passive'' and CAS failure. As mentioned above, the active thread is captured as the bold thread that succeeds its CAS in the bottom right layer; the passive thread is the thread that finds itself collided with in the layers on arcs exiting the bottom right layer.


\HSYprose{We say that op is trying to collide at state s, if, in s,
the value of t’s program counter is pointing at a statement
of one of the following procedures: LesOP, TryCollision,
FinishCollision. Otherwise, we say that op is not trying
to collide at s.}

\noindent
Here the authors' intuitive concept of ``trying to collide'' is captured by the ``Publish collision intent'' quotient automaton state, as compared to the other states. 

\HSYprose{We next prove that operations can only collide with operations of the opposite type. First we need the following technical lemma. Lemma 5.2. Every colliding operation op is either active or passive, but not both.}

\noindent
\repeatedQuo{} Furthermore, the bottom right layer shows that the colliding operations cannot be both active and passive.

\HSYprose{Lemma 5.3. Operations can only collide with operations of the opposite type: an operation that performs a push can only collide with operations that perform a pop, and vice versa.}

\noindent
\repeatedQuo{} 

\HSYprose{Lemma 5.4. An operation terminates without modifying the central stack object, if and only if it collides with another operation.}

\noindent
This is captured by the bottom left layer, which (1) involves \code{return val} statements, avoiding the central stack and (2) is only reachable after a successful collision.

\HSYprose{Lemma 5.5. For every thread p and in any state s, if p is not trying to collide in s, then it holds in s that the element corresponding to p in the location array is NULL.}

\noindent
Captured by the initial conditions and the (only possible) paths through the automaton.

\HSYprose{Lemma 5.6. Let op be a push operation by some thread p; if location[p] $\neq$ NULL, then op is trying to push the value \texttt{location[p]->cell.pdata}.}

\noindent
Captured by the initial conditions and the (only possible) paths through the automaton.

\HSYprose{we show that push and pop operations are paired correctly during collisions. Lemma 5.7. Every passive collider collides with exactly one active collider.}

\noindent
\repeatedQuo{} 

\HSYprose{Lemma 5.8. Every active collider op1 collides with exactly one passive collider.}

\noindent
\repeatedQuo{} 

\HSYprose{Lemma 5.9. Every colliding operation op participates in exactly one collision with an operation of the opposite type.}

\noindent
\repeatedQuo{} 

\HSYprose{We now prove that, when colliding, opposite operations exchange values in a proper way. Lemma 5.10. If a pop operation collides, it obtains the value of the single push operation it collided with. [Lemma 5.11 analogous for push-pop.]}

\noindent
\repeatedQuo{} 

\HSYprose{We can now finally prove that our algorithm has correct set semantics. Theorem 5.12. The elimination-backoff stack has correct set semantics.}

\noindent
As discussed above, separately proving Set semantics is unnecessary.

{\bf Linearizability.}

\HSYprose{we choose the following linearization points for all operations, except for
passive-colliders: Lines T4, C2 (for a push operation), Lines T10, T14, C7 (for a pop operation)}

\noindent
The authors give linearization points for ``active'' threads as the time when the second CAS succeeds, and linearization points for ``passive'' threads ``the time of linearization of the matching active-collider operation, and the push colliding-operation is linearized before the pop colliding-operation.'' 
The linearization points in the quotient are:
(1) the bold successful CAS in the bottom right layer in the quotient automaton, and 
(2) the subsequent automaton transition where a corresponding passive thread finds it has been collided with.
Importantly, every run of the quotient automaton gives a serial linearization order that is a
repetition of pairs of active/passive threads. All other executions are equivalent to one such serialized run, up to commutativity.

\HSYprose{For a passive-collider operation, we set the linearization point to be at the time of linearization of the matching active-collider operation, and the push colliding-operation is linearized before the pop colliding-operation.}

\noindent
Same as above.

\HSYprose{Each push or pop operation consists of a while loop that
repeatedly attempts to complete the operation. An iteration
is successful if its attempt succeeds, in which case the operation returns at that iteration; otherwise, another iteration
is performed . Each completed operation has exactly one
successful attempt (its last attempt), and the linearization
of the operation occurs in that attempt.
In other words, the
operations are linearized in the aforementioned lineanirazation points only in case of a successful CAS, which can only
be performed in the last iteration of the while loop.
}

\noindent
Same as above.

\HSYprose{To prove that the aforementioned lines are correct linearization points of our algorithm, we need to prove that
these are correct linearization points for the two types of operations: operations that complete by modifying the central
stack object, and operations that exchange values through collisions.}

\noindent
Same as above.

\HSYprose{Lemma 5.13. For operations that do not collide, we can choose the following linearization points: Line T4 (for a push operation). Line T10 (in case of empty stack) or line T14 (for a pop operation)}

\noindent
Follows from the quotient automaton for the Treiber central stack.

\HSYprose{We still have to prove that the linearization points for
collider-operations are consistent, both with one another,
and with non-colliding operations. We need the following
technical lemma, whose proof is omitted for lack of space.
Lemma 5.14. Let op1, op2, be a colliding operations-pair,
and assume w.l.o.g. that op1 is the active-collider and op2
is the passive collider, then the linearization point of op1 (as
defined above) is within the time interval of op2.}

\noindent
Same as above.

\HSYprose{Lemma 5.15. The following are legal linearization points for collider-operations. • An active-collider, op1, is linearized at either line C2 (in case of a push operation) or at line C7 (in case of a pop operation). • A passive-collider, op2, is linearized at the linearization time of the active-collider it collided with. If op2 is a push operation, it is linearized immediately before op1, otherwise it is linearized immediately after op1.}

\noindent
Same as above.

\paragraph{Summary.} 
The quotient naturally and succinctly captures the key concept of the Elimination stack: that a single successful CAS of one type of operation is the linearization point for that operation as well as the corresponding matched operation (order with the push before the pop). Specifically, every run of the quotient automaton gives a serial linearization order that is a repetition of pairs of active/passive threads. All other executions are equivalent to one such serialized run, upto commutativity.

Many of the lemmas and reasoning in the~\citet{DBLP:conf/spaa/HendlerSY04} proof are used to set up a bijection between active and passive threads. The quotient instead simplifies the proof through the serialized representative executions. The quotient similarly simplifies the other logistics of threads preparing/completing in the other quotient automaton states.


\subsection{Michael-Scott Queue}

The layer quotient for the MSQ is given in Sec.~\ref{ssec:MSQ}.  We will refer to the layers defined there.

\newcommand\MSQprose[1]{\QuoteAuthor{#1}{\citet{TAOMPP}}}
\MSQprose{We now review all the steps in detail. An enqueuer creates a new node with
the new value to be enqueued (Line 10), reads tail, and finds the node that
appears to be last (Lines 12–13). To verify that node is indeed last, it checks whether
that node has a successor (Line 15). If so, the thread attempts to append the new
node by calling compareAndSet() (Line 16). (A compareAndSet() is required
because other threads may be trying the same thing.) }

\noindent
The above scenario involves a single successful enqueuer and unboundedly many other enqueuers attempting. This scenario is captured by the \emph{Enqueue Succeed Layer}, and the automaton transitions shown in Fig.~\ref{fig:aut-msq}.

\MSQprose{If the compareAndSet()
succeeds, the thread uses a second compareAndSet() to advance tail to the
new node (Line 17). Even if this second compareAndSet() call fails, the thread
can still return successfully because, as we will see, the call fails only if some
other thread “helped” it by advancing tail.}

\noindent
The above scenario corresponds to the \emph{Advancer Succeed Layer}, where some advancer succeeds.

\MSQprose{If the tail node has a successor (Line
20), then the method tries to “help” other threads by advancing tail to refer
directly to the successor (Line 21) before trying again to insert its own node.}

\noindent
The above scenario corresponds to the \emph{Advancer Succeed Layer}, and the fact that ``trying again to insert'' occurs in a subsequent layer.

\MSQprose{This enq() is total, meaning that it never waits for a dequeuer. A successful enq()
is linearized at the instant where the executing thread (or a concurrent helping
thread) calls compareAndSet() to redirect the tail field to the new node at
Line 21.}

\noindent
This linearization point occurrence is preserved in the layer quotient abstraction, at the point where the tail is advanced.

\MSQprose{The deq() method is similar to its total counterpart from the UnboundedQueue.
If the queue is nonempty, the dequeuer calls compareAndSet() to change head
from the sentinel node to its successor, making the successor the new sentinel
node. The deq() method makes sure that the queue is not empty in the same way
as before: by checking that the next field of the head node is not null.}

\noindent
This scenario is captured by the \emph{Dequeue Succeed Layer} (when the queue is non-empty) and by \emph{Read Only Layer 1} (where dequeue returns because the queue was empty).

Regarding ADT states, the correctness argument mentions two aspects: (1) whether the queue was empty and (2) whether the tail pointer was lagged. This is captured in the automaton representation in Fig.~\ref{fig:aut-msq}, where the states capture both 1 and 2.

\begin{center}
\includegraphics[width=2.0in]{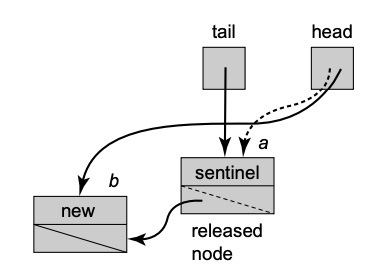}
\end{center}
\MSQprose{There is, however, a subtle issue in the lock-free case, depicted [above]:
before advancing head one must make sure that tail is not left referring to the
sentinel node which is about to be removed from the queue. To avoid this problem
we add a test: if head equals tail (Line 31) and the (sentinel) node they refer to
has a non-null next field (Line 32), then the tail is deemed to be lagging behind.
As in the enq() method, deq() then attempts to help make tail consistent by
swinging it to the sentinel node’s successor (Line 35), and only then updates head
to remove the sentinel (Line 38). As in the partial queue, the value is read from
the successor of the sentinel node (Line 37). If this method returns a value, then
its linearization point occurs when it completes a successful compareAndSet()
call at Line 38, and otherwise it is linearized at Line 33.} 

\noindent
There are multiple layers discussed above. First, there is an \emph{Advancer Succeed Layer} as part of a dequeue. Second, a \emph{Dequeue Succeed Layer} (or \emph{Read Only Layer 1})  may occur, but only after (``only then'') the  \emph{Advancer Succeed Layer}. 
This scenario is focused on the refers to $q_3$ in the quotient automaton, where \texttt{Q.tail=Q.head} and yet \texttt{Q.tail->next$\neq$null}. The automaton helps illuminate this case because the states and arcs require one to consider all possible cases, which layers are enabled, and where the arcs land after the layer.


\paragraph{Summary.} See Sec.~\ref{ssec:MSQ}.

\subsection{SLS Queue}
\label{apx:resqueue-eval}

We review the correctness argument presented by the original
authors~\cite{Scherer2006}, quoting their prose and discussing how those statements correspond to our quotient proof methodology with layer expressions.


\SLS{The reservation linearization point for this code path occurs at line 
\ref{ln:insertOffer} when we successfully insert our offering into the queue}

\noindent
First, this prose indicates that $\circled{3t}$ is a linearization point. The write of $\circled{3t}$ is atomic and so this line number has a corresponding location in the layer expression, which is this same linearization point.
Second, this prose identifies a layer as an important state change: inserting an offer node into the queue. This is the EAIN layer in our decomposition.
Third, the prose describes what kind of data change is important: the tail changing to non-null, a distinction we make in the states of our layer automaton.

\SLS{``A successful followup linearization point occurs when we notice at line \ref{ln:taken} that our data has been taken.}

\noindent
Similarly here this linearization point appears in a layer where a dequeue mutates the state, and local path $\ropath{3''}$ is feasible. This prose also identifies important state change: from an item to null.



\SLS{The other case occurs when the queue consists of reservations
(requests for data), and is depicted [below]. }
\begin{center}
\includegraphics[height=1.7in]{resqueue.png}
\end{center}

\SLS{In this case, after originally reading the head node (step A), we read its successor
(line~\ref{ln:readHeadNext}/step B) and verify consistency (line~\ref{ln:verifyHeadNext}). Then, we attempt to supply our data to the \underline{head-most reservation} (line~\ref{ln:appendHead}/C). If this succeeds, we dequeue the former dummy node (~\ref{ln:removeDummy}/D) and return}

\noindent
This prose again indicates important state changes, which are reflected as 
distinct states (and transitions between them) in our layer automata: 
whether head-most reservation has data supplied and
whether the head dummy node needs to be advanced.

%

\SLS{If it fails, we need to go to the next reservation, so we dequeue
the old dummy node anyway (28) and retry the entire operation (32,
05).}

\noindent
This is a description of the failure path $\circled{5f} \cdot (\circled{6bt} + \ropath{6bf}$) 
and that interference (implicitly) caused by  a concurrent cas from $\circled{5t}$.


\SLS{The reservation linearization point for this code path occurs
when we successfully supply data to a waiting consumer at line~\ref{ln:appendHead};
the followup linearization point occurs immediately thereafter.}

\noindent
Again, this prose indicates the important state transition
at $\circled{5t}$, replacing a null with an item (as seen in the states of our layer automaton), and corresponding automaton transition for layer EFHR.

%
%
%

\paragraph{Summary.} A summary is given in Sec.~\ref{subsec:resqueue}.

\subsection{Herlihy-Wing Queue}
\label{apx:revisit-hwq}

We now examine the author's proof of this object. As discussed in Sec.~\ref{subsec:hwq}, the quotient can be abstracted as: $\left( \deqF^*\cdot (\enqinc)^+ \cdot\enqwr^* \cdot \deqT^*\right)^*$.
A proof of correctness is given in Appendix II of~\citet{DBLP:journals/toplas/HerlihyW90}. A key challenge of this object is that linearization points are non-fixed. 

\QuoteAuthor{An Enq execution occurs in two distinct steps, which may be interleaved with steps of other concurrent operations: an array slot is reserved by atomically incrementing back, and the new item is stored in items.}{Sec 4.1 of ~\citet{DBLP:journals/toplas/HerlihyW90}}

\noindent
This describes an execution scenario with unboundedly many threads, though is not yet an argument for why that scenarios is correct. This scenario appears in the quotient as the fact that $\enqinc$ and $\enqwr$ are distinct.

To cope with non-fixed linearization points (in this and other objects), the authors introduce a proof methodology based on tracking all possible linearizations that could happen in the future:

\QuoteAuthor{For each linearized value, it is sometimes useful to keep track of which invocations
were completed in the linearization that yielded that value, and what their responses were. A possibility for a history $H$ is a triple $(v, P, R)$, where $v$ is a linearized value of $H$, $P$ is the subset of pending invocations in $H$ not completed when forming the linearization that yielded $u$, and $R$ is the set of responses appended to $H$ to form $u$. }{Appendix I of ~\citet{DBLP:journals/toplas/HerlihyW90}}

\noindent
This is a rather general method for linearizability. The quotient, however, allows one to consider scenarios along the lines of ``one or more enqueuers increment back, possibly some of them write to the array, and then some dequeuers succeed,'' following the quotient's regular expression. 

Importantly, while the Appendix I of ~\citet{DBLP:journals/toplas/HerlihyW90} methdology maintains a history to allow for all possible linearization orders, 
quotient-based reasoning instead involves 
representative executions (those that are accepted by the regular expression) with \emph{fixed} linearization orders and all other executions are equivalent to one such representative execution upto commutativity.

\fi

\end{document}